\documentclass{amsart}
\usepackage{amssymb}
\usepackage{mathrsfs}
\usepackage{graphicx}
\usepackage[unicode]{hyperref}

\numberwithin{equation}{section}

\newtheorem{theorem}{Theorem}[section]
\newtheorem{lemma}[theorem]{Lemma}
\newtheorem{proposition}[theorem]{Proposition}
\newtheorem{corollary}[theorem]{Corollary}

\theoremstyle{definition}

\theoremstyle{remark}

\DeclareMathOperator{\sinc}{sinc}
\DeclareMathOperator{\ch}{ch}
\DeclareMathOperator{\sh}{sh}
\DeclareMathOperator{\tg}{tg}
\DeclareMathOperator{\thyp}{th}
\DeclareMathOperator{\arctg}{arctg}
\DeclareMathOperator{\Cos}{Cos}
\DeclareMathOperator{\Sinc}{Sinc}

\newcommand{\R}{{\mathbb R}}
\newcommand{\Z}{{\mathbb Z}}
\renewcommand{\C}{{\mathbb C}}

\renewcommand{\Im}{\mathrm{Im}\,}
\renewcommand{\Re}{\mathrm{Re}\,}

\begin{document}

\title[Coupling constant dependence]{Coupling constant dependence \\for the Schr\"odinger equation \\with an inverse-square potential}

\author{A.G.~Smirnov}
\address{I.~E.~Tamm Theory Department, P.~N.~Lebedev
Physical Institute, Leninsky prospect 53, Moscow 119991, Russia}
\email{smirnov@lpi.ru}

\begin{abstract}
We consider the one-dimensional Schr\"odinger equation $-f''+q_\alpha f = Ef$ on the positive half-axis with the potential $q_\alpha(r)=(\alpha-1/4)r^{-2}$. It is known that the value $\alpha=0$ plays a special role in this problem: all self-adjoint realizations of the formal differential expression $-\partial^2_r + q_\alpha(r)$ for the Hamiltonian have infinitely many eigenvalues for $\alpha<0$ and at most one eigenvalue for $\alpha\geq 0$. We find a parametrization of self-adjoint boundary conditions and eigenfunction expansions that is analytic in $\alpha$ and, in particular, is not singular at $\alpha = 0$. Employing suitable singular Titchmarsh--Weyl $m$-functions, we explicitly find the spectral measures for all self-adjoint Hamiltonians and prove their smooth dependence on $\alpha$ and the boundary condition. Using the formulas for the spectral measures, we analyse in detail how the ``phase transition'' through the point $\alpha=0$ occurs for both the eigenvalues and the continuous spectrum of the Hamiltonians. 
\end{abstract}

\maketitle

\section{Introduction}
\label{intro}

This paper is devoted to eigenfunction expansions connected with the one-dimensional Schr\"odinger equation
\footnote{The shift of the coupling constant by $1/4$ in the potential term is a matter of technical convenience: it allows us to get simpler expressions for the solutions of~(\ref{1}).}
\begin{equation}\label{1}
-\partial^2_r f(r) + \frac{\alpha-1/4}{r^2} f(r) = Ef(r),\quad r\in\R_+,
\end{equation}
where $\alpha$ and $E$ are real parameters and $\R_+$ denotes the positive half-axis $(0,\infty)$.

There are two special values of the coupling constant $\alpha$ at which this problem undergoes a structural change. One of them is $\alpha=1$. For $\alpha<1$, all solutions of~(\ref{1}) are square-integrable on the interval $(0,a)$ for every $a>0$. At the same time,  only one solution (up to a constant factor) for each $E$ possesses this property for $\alpha\geq 1$. In terms of the well-known Weyl alternative, this means that the differential expression
\begin{equation}\label{2}
-\partial^2_r + \frac{\alpha - 1/4}{r^2}
\end{equation}
corresponds at $r=0$ to the limit point case for $\alpha\geq 1$ and to the limit circle case for $\alpha<1$. As a consequence, (\ref{2}) has a unique self-adjoint realization in $L_2(\R_+)$\footnote{Here and subsequently, we let $L_2(\R_+)$ denote the Hilbert space of (equivalence classes of) square-integrable complex functions on $\R_+$.} for $\alpha\geq 1$ and infinitely many self-adjoint realizations in $L_2(\R_+)$ for $\alpha<1$. The latter correspond to various self-adjoint boundary conditions at $r=0$.

Another special value of $\alpha$ is $\alpha=0$. It has long been known~\cite{Case, Meetz} that the spectrum of all self-adjoint realizations of~(\ref{2}) is not bounded from below and contains infinitely many negative eigenvalues for $\alpha<0$. On the other hand, every self-adjoint realization of~(\ref{2}) has at most one eigenvalue for $\alpha\geq 0$ (the continuous spectrum is $[0,\infty)$ for all real $\alpha$).

If $\kappa\in\R$ and $\alpha=\kappa^2$, then the function\footnote{In this paper, we use the symbol $\sqrt{x}$ only for nonnegative $x$; the notation $z^{1/2}$ will be used for a suitable branch of the square root in the complex plane.} $f(r) = \sqrt{r} J_\kappa(\sqrt{E}r)$, where $J_\kappa$ is the Bessel function of the first kind of order $\kappa$, is a solution of~(\ref{1}) for every $E>0$ (this follows immediately from the fact that $J_\kappa$ satisfies the Bessel equation). These solutions can be used to expand square-integrable functions on $\R_+$. More precisely, given $\kappa\geq 0$ and a square-integrable complex function $\psi$ on $\R_+$ that vanishes for large $r$, we can define the function $\hat \psi$ on $\R_+$ by setting
\begin{equation}\label{3}
\hat\psi(E) = \frac{1}{\sqrt{2}}\int_0^\infty \sqrt{r}J_{\kappa}(\sqrt{E} r)\psi(r)\,dr,\quad E>0.
\end{equation}
The map $\psi\to\hat\psi$ then coincides up to a change of variables with the well-known Hankel transformation~\cite{Hankel} and induces a uniquely determined unitary operator in $L_2(\R_+)$. Since the development of a general theory of singular Sturm--Liouville operators by Weyl~\cite{Weyl}, this transformation has been used by many authors to illustrate various approaches to eigenfunction expansions for problems of this type~\cite{Weyl1,Titchmarsh,Naimark,GesztesyZinchenko,Fulton,KST}.

For $\alpha\geq 1$, transformation~(\ref{3}) with $\kappa=\sqrt{\alpha}$ provides an eigenfunction expansion (i.e., a diagonalizing unitary operator) for the unique self-adjoint realization of~(\ref{2}). If $0\leq\alpha<1$, then it is an eigenfunction expansion for one of infinitely many self-adjoint realizations of~(\ref{2}), namely, for the Friedrichs extension of the minimal operator $h_\alpha$ associated with~(\ref{2}) (see~\cite{EverittKalf}; the precise definition of $h_\alpha$ will be given later in this section). As we shall see, the latter is not bounded from below and, therefore, has no Friedrichs extension for $\alpha<0$. Accordingly, the right-hand side of~(\ref{3}) as
a function of $\alpha$ has a branch point at $\alpha=0$ and cannot be analytically continued to the region $\alpha<0$.

For $\alpha\geq 0$, eigenfunction expansions corresponding to all self-adjoint realizations of~(\ref{2}) were found in~\cite{Titchmarsh} (however, without explicitly using the language of operators in Hilbert space). In~\cite{Meetz}, all self-adjoint Hamiltonians associated with~(\ref{2}) and corresponding eigenfunction expansions were constructed for every real $\alpha$ using the theory of self-adjoint extensions of symmetric operators (a somewhat different treatment of this problem in the framework of self-adjoint extensions can be found in~\cite{GTV2010,GTV2012}).

The generalized eigenfunctions used in~\cite{Titchmarsh,Meetz,GTV2010,GTV2012} had the same type of branch point singularity at $\alpha=0$ as that appearing in Hankel transformation~(\ref{3}).
As a result, the cases $0<\alpha<1$, $\alpha=0$, and $\alpha<0$ were treated separately and eigenfunction expansions for $\alpha = 0$ could not be obtained from those for $0<\alpha<1$ and $\alpha<0$ by taking the limit $\alpha\to 0$. In~\cite{Smirnov2016}, we considered problem~(\ref{1}) with $\alpha=\kappa^2$ and constructed a parametrization of self-adjoint realizations of~(\ref{2}) and corresponding eigenfunction expansions that is continuous in $\kappa$ on the interval $(-1,1)$ (and, in particular, at $\kappa=0$). This work was motivated by our previous research~\cite{Smirnov2015} of the Aharonov--Bohm model, where zero and nonzero $\kappa$ correspond to integer and noninteger values of the dimensionless magnetic flux through the solenoid. In terms of $\alpha$, the results of~\cite{Smirnov2016} give a continuous transition from the region $0<\alpha<1$ to $\alpha=0$.

In this paper, we extend the treatment in~\cite{Smirnov2016} to also cover the region $\alpha<0$. We parametrize all eigenfunction expansions associated with~(\ref{2}) in such a way that the generalized eigenfunctions turn out to be analytic in $\alpha$ for $\alpha<1$, while the corresponding spectral measures are infinitely differentiable in $\alpha$ on the same interval. Using explicit formulas for the spectral measures, we analyse in detail how the transition through the point $\alpha=0$ occurs for both the eigenvalues and the continuous spectrum of self-adjoint realizations of~(\ref{2}) in this parametrization.

We now give a brief structural description of our results.

For every $\alpha\in\C$, we define the function $q_\alpha$ on $\R_+$  by setting
\begin{equation}\label{qalpha}
q_\alpha(r) = \frac{\alpha-1/4}{r^2}, \quad r>0.
\end{equation}
For real $\alpha$, $q_\alpha$ is the potential term in~(\ref{2}).

Let $\lambda_+$ be the restriction to $\R_+$ of the Lebesgue measure $\lambda$ on $\R$ and $\mathcal D$ be the space of all complex continuously differentiable functions on $\R_+$ whose derivative is absolutely continuous on $\R_+$ (i.e., absolutely continuous on every segment $[a,b]$ with $0<a\leq b<\infty$).
Given $\alpha,z\in\C$, we let $\mathscr L_{\alpha,z}$ denote the linear operator from $\mathcal D$ to the space of complex $\lambda_+$-equivalence classes such that
\begin{equation}\label{lalphaz}
(\mathscr L_{\alpha,z} f)(r) = -f''(r)+ q_\alpha(r) f(r) - z f(r)
\end{equation}
for $\lambda$-a.e.\footnote{Throughout the paper, a.e. means either ``almost every'' or ``almost everywhere.''} $r\in \R_+$ and set
\begin{equation}\label{lalpha0}
\mathscr L_\alpha = \mathscr L_{\alpha,0}.
\end{equation}

Let $\alpha\in\R$. We define the linear subspace $\Delta_\alpha$ of $\mathcal D$ by setting
\begin{equation}\label{Delta}
\Delta_\alpha = \{f\in \mathcal D : f \mbox{ and } \mathscr L_\alpha f \mbox{ are both square-integrable on } \R_+\}.
\end{equation}
For every linear subspace $Z$ of $\Delta_\alpha$, let $H_\alpha(Z)$ be the linear operator in $L_2(\R_+)$ defined by the relations\footnote{Here and subsequently, we let $D_F$ denote the domain of definition of a map $F$.}
\begin{equation}\label{H_alpha(Z)}
\begin{split}
& D_{H_\alpha(Z)} = \{[f]: f\in Z\},\\
& H_\alpha(Z) [f] = \mathscr L_\alpha f,\quad f\in Z,
\end{split}
\end{equation}
where $[f]=[f]_{\lambda_+}$ denotes the $\lambda_+$-equivalence class of $f$.
We clearly have $C_0^\infty(\R_+)\subset \Delta_\alpha$, where $C_0^\infty(\R_+)$ is the space of all smooth functions on $\R_+$ with compact support. The operator
\begin{equation}\label{checkh}
\check h_\alpha = H_\alpha(C_0^\infty(\R_+))
\end{equation}
is obviously symmetric and, hence, closable.  The closure of $\check h_\alpha$ is denoted by $h_\alpha$,
\begin{equation}\label{hkappadef}
h_\alpha = \overline{\check h_\alpha}.
\end{equation}
We shall see that the adjoint $h_\alpha^*$ of $h_\alpha$ is given by
\begin{equation}\label{halpha*}
h_\alpha^* = H_\alpha(\Delta_\alpha).
\end{equation}
If $T$ is a symmetric extension of $h_\alpha$, then $h^*_\alpha$ is an extension of $T^*$ and hence of $T$. By~(\ref{halpha*}), we conclude that $T$ is of the form $H_\alpha(Z)$ for some subspace $Z$ of~$\Delta_\alpha$.

Self-adjoint operators of the form $H_\alpha(Z)$ can be naturally viewed as self-adjoint realizations of differential expression~(\ref{2}). If $H_\alpha(Z)$ is self-adjoint, then equality~(\ref{halpha*}) and the closedness of $h_\alpha$ imply that $H_\alpha(Z)$ is an extension of $h_\alpha$ because $H_\alpha(\Delta_\alpha)$ is an extension of $H_\alpha(Z)$. Therefore, the self-adjoint realizations of~(\ref{2}) are precisely the self-adjoint extensions of $h_\alpha$ (or, equivalently, of $\check h_\alpha$).

For every $\alpha,z\in\C$, we shall construct real-analytic functions $\mathcal A^\alpha(z)$ and $\mathcal B^\alpha(z)$ on $\R_+$ such that
\begin{equation}\label{LAB}
\mathscr L_{\alpha,z} \mathcal A^\alpha(z) = \mathscr L_{\alpha,z} \mathcal B^\alpha(z) = 0,\quad \alpha,z\in\C.
\end{equation}
The functions $\mathcal A^\alpha(z)$ and $\mathcal B^\alpha(z)$ are real for real $\alpha$ and $z$.
Moreover, the quantities\footnote{Given a map $F$ whose values are also maps, we let $F(x|y)$ denote the value of $F(x)$ at a point $y$: $F(x|y) = (F(x))(y)$.} $\mathcal A^\alpha(z|r)$ and $\mathcal B^\alpha(z|r)$ are entire analytic in $\alpha$ and $z$ for every fixed $r>0$ and, in particular, are not singular at $\alpha=0$. If $\alpha<1$ and $z\in\C$, then $\mathcal A^\alpha(z)$ and $\mathcal B^\alpha(z)$ are linear independent and are both square-integrable on the interval $(0,a)$ for every $a>0$ (as mentioned above, we have the limit circle case for $\alpha<1$).

Given $f,g\in \mathcal D$, we let $W_r(f,g)$ denote the Wronskian of $f$ and $g$ at a point $r>0$,
\begin{equation}\label{wronskian}
W_r(f,g) = f(r)g'(r) - f'(r)g(r).
\end{equation}
Clearly, $r\to W_r(f,g)$ is an absolutely continuous function on $\R_+$.

For every $\alpha,\vartheta,z\in\C$, we define the function $\mathcal U^\alpha_\vartheta(z)$ on $\R_+$ by the relation
\begin{equation}\label{Ualphatheta}
\mathcal U^\alpha_\vartheta(z) = \mathcal A^\alpha(z)\cos\vartheta + \mathcal B^\alpha(z)\sin\vartheta.
\end{equation}
By~(\ref{LAB}), we obviously have
\begin{equation}\label{lalpha}
\mathscr L_{\alpha,z} \mathcal U^\alpha_\vartheta(z) = 0,\quad \alpha,\vartheta,z\in\C.
\end{equation}
The properties of $\mathcal A^\alpha(z)$ and $\mathcal B^\alpha(z)$ imply that
$\mathcal U^\alpha_\vartheta(z)$ is real for real $\alpha,\vartheta$, and $z$ and the quantity $\mathcal U^\alpha_\vartheta(z|r)$ is entire analytic in $\alpha,\vartheta$, and $z$ for every fixed $r>0$. If $\alpha<1$, then $\mathcal U^\alpha_\vartheta(z)$ is nontrivial for every $\vartheta,z\in\C$. We shall show that $\lim_{r\downarrow 0} W_r(\mathcal U^\alpha_\vartheta(z),\mathcal U^\alpha_\vartheta(z')) = 0$ for all $\alpha<1$, $\vartheta\in\R$, and $z,z'\in\C$. This condition means that $\mathcal U^\alpha_\vartheta(z|r)$ for various $z$ have the same asymptotics as $r\downarrow 0$.

Let $f\in \mathcal D$ and $\alpha,\vartheta,z\in\C$. In view of~(\ref{lalpha}), integration by parts yields
\[
\int_r^a (\mathscr L_{\alpha,z}f)(r') \mathcal U^\alpha_\vartheta(z|r')\,dr' = W_r(\mathcal U^\alpha_\vartheta(z),f) - W_a(\mathcal U^\alpha_\vartheta(z),f)
\]
for every $a, r>0$. If $\alpha<1$, then $\mathcal U^\alpha_\vartheta(z)$ is square-integrable on $(0,a)$ for every $a>0$ and this equality implies that $W_r(\mathcal U^\alpha_\vartheta(z),f)$ has a limit as $r\downarrow 0$ for every $\vartheta,z\in\C$ and $f\in \Delta_\alpha$. Given $\alpha<1$ and $\vartheta\in\R$, we define the operator $h_{\alpha,\vartheta}$ in $L_2(\R_+)$ by the relation
\begin{equation}\label{halphatheta}
h_{\alpha,\vartheta} = H_\alpha(Z_{\alpha,\vartheta}),
\end{equation}
where the linear subspace $Z_{\alpha,\vartheta}$ of $\Delta_\alpha$ is given by
\begin{equation}\label{zalphatheta}
Z_{\alpha,\vartheta} = \{f\in \Delta_\alpha : \lim_{r\downarrow 0} W_r(\mathcal U^\alpha_\vartheta(0),f) = 0\}.
\end{equation}
By~(\ref{Ualphatheta}) and the definition of $h_{\alpha,\vartheta}$, we have
\begin{equation}\label{periodh}
h_{\alpha,\vartheta+\pi} = h_{\alpha,\vartheta},\quad \alpha<1,\,\vartheta\in\R.
\end{equation}

The next statement gives a complete description of the self-adjoint extensions of $h_\alpha$ for every $\alpha\in\R$.
\begin{theorem}\label{t_sa}
For $\alpha\geq 1$, the operator $h_\alpha$ is self-adjoint. If $\alpha<1$, then $h_{\alpha,\vartheta}$ is a self-adjoint extension of $h_\alpha$ for every $\vartheta\in \R$ and, conversely, every self-adjoint extension of $h_\alpha$ is equal to $h_{\alpha,\vartheta}$ for some $\vartheta\in\R$. Given $\vartheta,\vartheta'\in\R$, we have $h_{\alpha,\vartheta} = h_{\alpha,\vartheta'}$ if and only if $\vartheta-\vartheta'\in\pi\Z$.
\end{theorem}

As will be obvious from the explicit definitions of $\mathcal A^\alpha(z)$ and $\mathcal B^\alpha(z)$ in Sec.~\ref{s_def}, these functions are actually square integrable on intervals $(0,a)$ with $a>0$ for all $\alpha$ belonging to the domain
\[
\Pi = \{\alpha\in\C: \alpha =\kappa^2 \mbox{ for some $\kappa\in\C$ such that }|\Re\kappa|<1\}.\
\]
and all $z\in\C$. Hence, the above definition of $h_{\alpha,\vartheta}$ can be naturally extended to all $\alpha\in\Pi$ and $\vartheta\in\C$. Moreover, it is possible to show that such an extended family of operators is holomorphic on $\Pi\times\C$ in the sense of Kato (see Ch.~7 in~\cite{Kato}) and, therefore, $h_{\alpha,\vartheta}$ with $\alpha<1$ and $\vartheta\in\R$ constitute a real-analytic family of operators. This can be proved using a technique similar to that employed in~\cite{BruneauDerezinskiGeorgescu} for the case of extensions of $h_\alpha$ homogeneous with respect to dilations of $\R_+$. The spectral analysis of a holomorphic family similar to $h_{\alpha,\vartheta}$ can be found in~\cite{DerezinskiRichard2017}, where, however, $\kappa$ rather than $\alpha$ was used as a parameter and the case $\alpha=0$ was treated separately (see also~\cite{DFNR2020} for an analogous treatment of the Coulomb potential; a possibility of removing the singularity at $\kappa = 0$ was indicated in Remark~2.5 in~\cite{DerezinskiRichard2017}). However, the analysis of $h_{\alpha,\vartheta}$ for complex $\alpha$ and $\vartheta$ is beyond the scope of this paper. In the sequel, we confine ourselves to the self-adjoint case $\alpha<1$, $\vartheta\in\R$.

Given a positive Borel measure $\nu$ on $\R$ and a $\nu$-measurable complex function $g$, we let $\mathcal T^\nu_g$ denote the operator of multiplication by $g$ in the Hilbert space $L_2(\R,\nu)$ of $\nu$-square-integrable complex functions on $\R$.\footnote{More precisely, $\mathcal T^\nu_g$ is the operator in $L_2(\R,\nu)$ whose graph consists of all pairs $(\varphi_1,\varphi_2)$ such that $\varphi_1,\varphi_2\in L_2(\R,\nu)$ and $\varphi_2(E) = g(E)\varphi_1(E)$ for $\nu$-a.e. $E$.} If $g$ is real, then $\mathcal T^\nu_g$ is self-adjoint. We let $L_2^c(\R_+)$ denote the subspace of $L_2(\R_+)$ consisting of all its elements vanishing $\lambda$-a.e. outside some compact subset of $\R_+$.

It turns out that the functions $\mathcal U^\alpha_\vartheta(E)$ with real $E$ can be used as generalized eigenfunctions for constructing eigenfunction expansions for $h_{\alpha,\vartheta}$. More precisely, for every $\alpha<1$ and $\vartheta\in\R$, we shall construct a positive Radon measure\footnote{We recall that a Borel measure $\nu$ on $\R$ is called a Radon measure on $\R$ if $\nu(K)<\infty$ for every compact set $K\subset \R$.} $\mathcal V_{\alpha,\vartheta}$ on $\R$ such that
\begin{equation}\label{est_meas}
\int_\R \frac{d\mathcal V_{\alpha,\vartheta}(E)}{E^2+1} < \infty
\end{equation}
and the following statement holds.

\begin{theorem}\label{leig2}
Let $\alpha<1$ and $\vartheta\in \R$. Then there is a unique unitary operator $U_{\alpha,\vartheta}\colon L_2(\R_+)\to L_2(\R,\mathcal V_{\alpha,\vartheta})$ such that
\begin{equation}\nonumber
(U_{\alpha,\vartheta}\psi)(E) = \int_0^\infty \mathcal U^\alpha_\vartheta(E|r)\psi(r)\,dr,\quad \psi\in L_2^c(\R_+),
\end{equation}
for $\mathcal V_{\alpha,\vartheta}$-a.e. $E$, and we have
\begin{equation}\label{diag}
h_{\alpha,\vartheta} = U_{\alpha,\vartheta}^{-1} \mathcal T^{\mathcal V_{\alpha,\vartheta}}_\iota U_{\alpha,\vartheta},
\end{equation}
where $\iota$ is the identity function on $\R$ (i.e., $\iota(E)=E$ for all $E\in\R$).
\end{theorem}

Clearly, the measures $\mathcal V_{\alpha,\vartheta}$ (which will be referred to as the spectral measures) contain all information about the spectral properties of the operators $h_{\alpha,\vartheta}$. In particular, $E\in \R$ is an eigenvalue of $h_{\alpha,\vartheta}$ if and only if the measure $\mathcal V_{\alpha,\vartheta}$ of the one-point set $\{E\}$ is strictly positive. In agreement with~(\ref{Ualphatheta}) and~(\ref{periodh}), $\mathcal V_{\alpha,\vartheta}$ is $\pi$-periodic in $\vartheta$,
\begin{equation}\label{periodicV}
\mathcal V_{\alpha,\vartheta+\pi} = \mathcal V_{\alpha,\vartheta},\quad \alpha<1,\,\vartheta\in \R.
\end{equation}

Let $\omega$ be the $\pi$-periodic function on $\R$ such that
\begin{equation}\label{omega}
\omega(\vartheta) = \left(1-\frac{2|\vartheta|}{\pi}\right)^2,\quad -\frac{\pi}{2}\leq\vartheta\leq \frac{\pi}{2}.
\end{equation}
Clearly, we have $0\leq\omega(\vartheta)\leq 1$ for all $\vartheta\in\R$.
We define the subsets $Q_0$, $Q_1$, and $Q_\infty$ of $\R^2$ by the relations (see Fig.~\ref{f0})
\begin{align}
& Q_0 = \{ (\alpha,\vartheta)\in\R^2 : \omega(\vartheta)\leq \alpha<1\},\label{Q0}\\
& Q_1 = \{ (\alpha,\vartheta)\in\R^2 : 0\leq \alpha<\omega(\vartheta)\},\label{Q11}\\
& Q_\infty = \{ (\alpha,\vartheta)\in\R^2 : \alpha<0\}.\label{Qi}
\end{align}
The analysis of the measures $\mathcal V_{\alpha,\vartheta}$ shows that $h_{\alpha,\vartheta}$ has no eigenvalues for $(\alpha,\vartheta)\in Q_0$, one eigenvalue for $(\alpha,\vartheta)\in Q_1$, and infinitely many eigenvalues that are not bounded from below for $(\alpha,\vartheta)\in Q_\infty$. Using a parametrization of generalized eigenfunctions that is analytic on the entire domain
\begin{equation}\label{domain}
Q = Q_0\cup Q_1\cup Q_\infty = \{(\alpha,\vartheta)\in\R^2: \alpha<1\}
\end{equation}
allows us to understand in detail what happens to eigenvalues as we pass from $Q_\infty$ to $Q_1$ through the line $\alpha=0$. It turns out that there is one eigenvalue that crosses this line in an analytic manner, while the rest infinitely many eigenvalues tend either to $-\infty$ or to zero as $\alpha\uparrow 0$ and die away there. Moreover, the density of $\mathcal V_{\alpha,\vartheta}$ corresponding to the continuous part of the spectrum turns out to be real-analytic on $Q$.

If $\alpha<0$, then the operator $h_\alpha$ is not bounded from below because otherwise it would have self-adjoint extensions (e.g., its Friedrichs extension) that are bounded from below, in contradiction to Theorem~\ref{t_sa} and the described properties of eigenvalues of $h_{\alpha,\vartheta}$ for $(\alpha,\vartheta)\in Q_\infty$. On the other hand, it is easy to see that $h_\alpha$ is positive for $\alpha\geq 0$. Indeed, let $f\in C_0^\infty(\R_+)$, $\psi = [f]$, and $\varrho$ be a real number. Using the integration by parts, we easily derive from~(\ref{checkh}) that
\[
\langle \psi,\check h_\alpha\psi\rangle = \int_0^\infty r^{2\varrho}\left(|\partial_r\tilde f(r)|^2 + \frac{\alpha-1/4-\varrho^2+\varrho}{r^2}|\tilde f(r)|^2\right)dr,
\]
where $\langle\cdot,\cdot\rangle$ is the scalar product in $L_2(\R_+)$ and the function $\tilde f$ on $\R_+$ is given by $\tilde f(r) = r^{-\varrho} f(r)$, $r>0$. The maximum of $-\varrho^2+\varrho$ is attained at $\varrho=1/2$ and is equal to $1/4$. Substituting this value to the above equality, we deduce that $\langle \psi,\check h_\alpha\psi\rangle\geq 0$ for all $\alpha\geq 0$ and $\psi\in D_{\check h_\alpha}$. The positivity of $h_\alpha$ for $\alpha\geq 0$ now follows from~(\ref{hkappadef}).

\begin{figure}
  \includegraphics[scale=0.75]{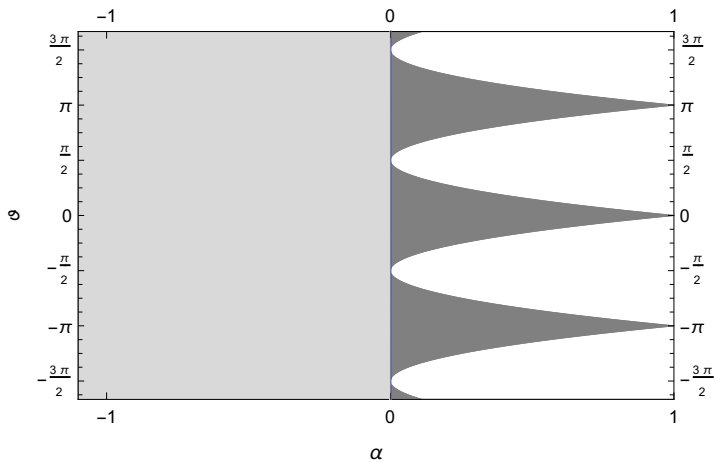}
  \caption{The sets $Q_0$, $Q_1$, and $Q_\infty$ are represented by white, dark gray, and light gray regions respectively.}
  \label{f0}
\end{figure}

It seems probable that the dependence of $\mathcal V_{\alpha,\vartheta}$ on $\alpha$ and $\vartheta$ is not analytic at the boundaries between the regions $Q_0$, $Q_1$, and $Q_\infty$, where eigenvalues arise and disappear (however, we do not prove this claim in this paper). At the same time, we shall show that this dependence is smooth on $Q$ in a suitable sense. To formulate this result precisely, we make use of the Schwartz space $\mathscr S$ of rapidly decreasing smooth functions. More specifically, $\mathscr S$ consists of all infinitely differentiable functions $\varphi$ on $\R$ such that
\[
\sup_{E\in\R,\,k\in \{0,\ldots,n\}} |\varphi^{(k)}(E)|(1+|E|)^n < \infty
\]
for every nonnegative integer $n$, where $\varphi^{(k)}$ stands for the $k$th derivative of $\varphi$. The space $\mathscr S$ is widely used in the theory of generalized functions as the test function space for tempered distributions. In view of~(\ref{est_meas}), every $\varphi \in\mathscr S$ is $\mathcal V_{\alpha,\vartheta}$-integrable for all $\alpha<1$ and $\vartheta\in\R$.

\begin{theorem}\label{t_smooth}
For every $\varphi\in \mathscr S$, the function $(\alpha,\vartheta)\to\int \varphi(E)\,d\mathcal V_{\alpha,\vartheta}(E)$ is infinitely differentiable on the domain $Q$ given by~$(\ref{domain})$.
\end{theorem}

Thus, our construction of eigenfunction expansions is, as a whole, at least infinitely differentiable.

When considering equation~(\ref{1}), it is convenient to set $\alpha=\kappa^2$ and find its solutions as functions of $\kappa$ (we have actually done so in the case of Hankel transformation~(\ref{3})). To return to the initial variable $\alpha$, it is then necessary to replace $\kappa$ with the square root of $\alpha$. As was discussed above, this may lead to the appearance of branch points and the loss of analyticity. This does not happen, however, if the solution in question is an even holomorphic function of $\kappa$. Indeed, suppose we have an even holomorphic function $g$, which will be assumed for simplicity to be entire. Then $g$ has the power series expansion of the form $g(w) = \sum_{k=0}^\infty c_k w^{2k}$ for every $w\in\C$. If we define the entire analytic function $G$ by the formula $G(\zeta) = \sum_{k=0}^\infty c_k\zeta^k$, $\zeta\in\C$, then we have
\begin{equation}\label{Gg}
G(w^2) = g(w)
\end{equation}
for all $w\in\C$ and, hence, $G(\zeta)$ can be viewed as a result of ``substituting the square root of $\zeta$'' in $g$. More generally, representations of type~(\ref{Gg}) can be obtained for even holomorphic functions on arbitrary reflection-symmetric domains and for the case of several complex variables (see Appendix~\ref{app_even}). Our construction of the solutions $\mathcal A^\alpha(z)$ and $\mathcal B^\alpha(z)$ is based on the described technique. We shall first find functions $\mathfrak a^\kappa(z)$ and $\mathfrak b^\kappa(z)$ that are even in $\kappa$ and satisfy~(\ref{1}) with $\alpha$ and $E$ replaced with $\kappa^2$ and $z$ respectively and then define $\mathcal A^\alpha(z)$ and $\mathcal B^\alpha(z)$ by requiring that $\mathcal A^{\kappa^2}(z)=\mathfrak a^\kappa(z)$ and $\mathcal B^{\kappa^2}(z)=\mathfrak b^\kappa(z)$ for every $\kappa,z\in\C$.

Simple examples of representations of type~(\ref{Gg}), which will be important for us, are obtained if we choose $g$ to be equal either to the cosine or the entire function $\sinc$ that is defined by the formula
\begin{equation}\nonumber
\sinc w = \left\{
\begin{matrix}
w^{-1}\sin w,& w\in\C\setminus\{0\},\\
1,& w=0.
\end{matrix}
\right.
\end{equation}
Proceeding as above, we find that
\begin{equation}\label{Cos1}
\cos w = \Cos(w^2),\quad \sinc w = \Sinc(w^2)
\end{equation}
for every $w\in\C$, where the entire functions $\Cos$ and $\Sinc$ are given by
\begin{equation}\label{CosSinc}
\Cos\zeta = \sum_{k=0}^\infty \frac{(-\zeta)^k}{(2k)!},\quad \Sinc\zeta = \sum_{k=0}^\infty \frac{(-\zeta)^k}{(2k+1)!},\quad \zeta\in\C.
\end{equation}
It follows from~(\ref{Cos1}) that
\begin{equation}\label{Cos2}
\Cos(-w^2) = \cos(iw) = \ch w,\quad w\Sinc(-w^2) = w\sinc(i w) = \sh w
\end{equation}
for every $w\in\C$. In particular, we have
\begin{align}
&\Cos\xi = \cos(\sqrt{\xi}),\quad \Sinc\xi = \sinc(\sqrt{\xi}), \quad \xi\geq 0,\label{Cos3}\\
&\Cos\xi = \ch(\sqrt{|\xi|}), \quad \Sinc\xi = |\xi|^{-1/2}\sh(\sqrt{|\xi|}),\quad \xi < 0. \label{Cos4}
\end{align}
The graphs of $\Cos \xi$ and $\Sinc\xi$ are shown in Fig.~\ref{f00}.

\begin{figure}
  \includegraphics[scale=0.75]{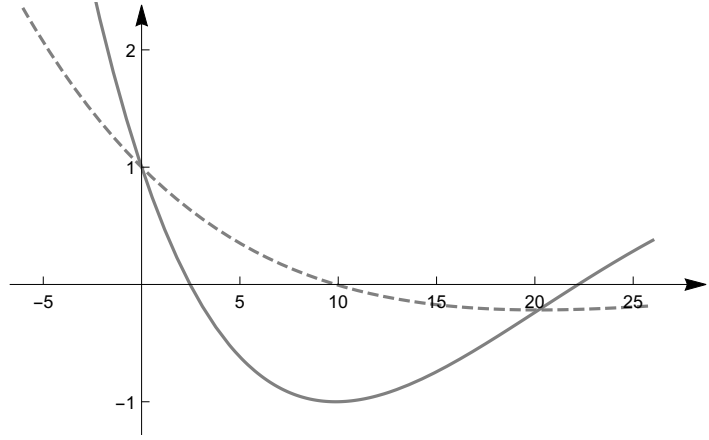}
  \caption{Solid and dashed lines correspond to the functions $\Cos$ and $\Sinc$ respectively.}
  \label{f00}
\end{figure}

Formulas~(\ref{Cos3}) and~(\ref{Cos4}) show that, in spite of being analytic, the functions $\Cos$ and $\Sinc$ are expressed in a piecewise way in terms of the standard trigonometric and hyperbolic functions.
We shall see that various quantities related to the spectral measures (such as eigenvalues and the density of the absolutely continuous part of $\mathcal V_{\alpha,\vartheta}$) can be conveniently expressed through $\Cos$ and $\Sinc$. Accordingly, the formulas for these quantities in terms of the ordinary elementary functions are of a piecewise nature. This suggests that $\Cos$ and $\Sinc$ are more suitable as ``elementary functions'' for our problem. Some properties of these functions that will be necessary for us are summarized in Appendix~\ref{app_CosSinc}.

The paper is organized as follows. In Sec.~\ref{s_def}, we define the solutions $\mathcal A^\alpha(z)$ and $\mathcal B^\alpha(z)$ and the spectral measures $\mathcal V_{\alpha,\vartheta}$, thus completing the formulation of our results. The definition of the measures $\mathcal V_{\alpha,\vartheta}$ in Sec.~\ref{s_def} is given via Herglotz representations (see Appendix~\ref{s_Herglotz}) of suitable holomorphic functions in the upper complex half-plane and is not quite explicit. In Sec.~\ref{s_meas}, we obtain concrete formulas for the point and absolutely continuous parts of $\mathcal V_{\alpha,\vartheta}$. In particular, this allows us to justify the ``phase diagram'' in Fig.~\ref{f0} and analyze the dependence of eigenvalues of $h_{\alpha,\vartheta}$ on $\alpha$ and $\vartheta$. In Sec.~\ref{s_one}, we recall the general theory concerning self-adjoint extensions of one-dimensional Schr\"odinger operators and apply it to the proof of Theorem~\ref{t_sa}. Our treatment of eigenfunction expansions relies on the method of singular Titchmarsh--Weyl $m$-functions~\cite{KST}. In Sec.~\ref{s_eig}, we briefly describe this method and then use it to prove Theorem~\ref{leig2}. Sec.~\ref{s_smooth} is devoted to the proof of Theorem~\ref{t_smooth}.

\section{Definition of generalized eigenfunctions and spectral measures}
\label{s_def}

\subsection{Definition of \texorpdfstring{$\mathcal A^\alpha(z)$}{A(z)} and \texorpdfstring{$\mathcal B^\alpha(z)$}{B(z)}.} For any $z,\kappa\in\C$, we define the function $\mathfrak u^\kappa(z)$ on $\R_+$ by the relation
\begin{equation}\label{ukappa}
\mathfrak u^\kappa(z|r) = r^{1/2+\kappa}\mathcal X_\kappa(r^2 z),\quad r>0,
\end{equation}
where the entire function $\mathcal X_\kappa$ is given by
\begin{equation}\label{Xkappa}
\mathcal X_\kappa(\zeta) = \frac{1}{2^\kappa}\sum_{n=0}^\infty \frac{(-1)^n\zeta^n}{\Gamma(\kappa+n+1)n!2^{2n}},\quad \zeta\in\C.
\end{equation}
The function $\mathcal X_\kappa$ is closely related to Bessel functions: for $\zeta\neq 0$, we have
\begin{equation}\label{bessel}
\mathcal X_\kappa(\zeta) = \zeta^{-\kappa/2}J_\kappa(\zeta^{1/2}).
\end{equation}
Because $J_\kappa$ satisfies the Bessel equation, it follows that
\begin{equation}\label{eqf}
-\partial^2_r \mathfrak u^{\pm\kappa}(z|r) +  \frac{\kappa^2-1/4}{r^2} \mathfrak u^{\pm\kappa}(z|r) = z\mathfrak u^{\pm\kappa}(z|r),\quad r>0,
\end{equation}
for every $\kappa\in\C$ and $z\neq 0$. By continuity, this is also true for $z=0$. We therefore have
\begin{equation}\label{Lu}
\mathscr L_{\kappa^2,z} \mathfrak u^{\pm\kappa}(z) =0,\quad \kappa,z\in\C.
\end{equation}

For every $\kappa,z\in\C$, we define the function $\mathfrak a^\kappa(z)$ on $\R_+$ by setting
\begin{equation}\label{akappa}
\mathfrak a^\kappa(z) = \frac{\mathfrak u^{\kappa}(z)-\mathfrak u^{-\kappa}(z)}{\kappa}\cos\vartheta_\kappa,\quad \kappa\in\C\setminus\{0\},
\end{equation}
and
\begin{equation}\label{a0}
\mathfrak a^0(z|r) = \lim_{\kappa\to 0} \mathfrak a^\kappa(z|r)= 2\left[\left(\ln\frac{r}{2} + \gamma\right)\mathfrak u^0(z|r) -\sqrt{r}\,\mathcal Y(r^2 z)\right],\quad r>0,
\end{equation}
where
\begin{equation}\label{varthetakappa}
\vartheta_\kappa=\frac{\pi\kappa}{2},
\end{equation}
the entire function $\mathcal Y$ is given by
\begin{equation}\nonumber
\mathcal Y(\zeta) = \sum_{n=1}^\infty \frac{(-1)^nc_n}{(n!)^2 2^{2n}} \zeta^n,\quad c_n = \sum_{j=1}^n \frac{1}{j},
\end{equation}
and $\gamma = \lim_{n\to\infty} (c_n -\ln n)=0,577\ldots$ is the Euler constant.\footnote{To compute the limit of $\mathfrak a^\kappa(z|r)$ as $\kappa\to 0$, we can apply L'H\^{o}pital's rule and use the equality $\Gamma'(1+n)/\Gamma(1+n)=c_n-\gamma$ (see formula~(9) in~Sec.~1.7.1 in~\cite{Bateman}).}

Further, for every $\kappa,z\in\C$, we define the function $\mathfrak b^\kappa(z)$ on $\R_+$ by the formula
\begin{equation}\label{bkappa}
\mathfrak b^\kappa(z) = \frac{\pi}{2}(\mathfrak u^\kappa(z) + \mathfrak u^{-\kappa}(z))\sinc\vartheta_\kappa,
\end{equation}
where $\vartheta_\kappa$ is given by~(\ref{varthetakappa}).

Given $\varphi\in \R$, we set $R_\varphi = \{z\in\C: z= re^{i\varphi} \mbox{ for some }r\geq 0\}$ and
\begin{equation}\nonumber
\C_\varphi = \C\setminus R_\varphi.
\end{equation}
Hence, $\C_\varphi$ is the complex plane with a cut along the ray $R_\varphi$.

The next statement shows that, notwithstanding a piecewise definition of $\mathfrak a^\kappa(z)$, both quantities $\mathfrak a^\kappa(z|r)$ and $\mathfrak b^\kappa(z|r)$ are actually analytic in all their arguments.

\begin{lemma}\label{l_analyt0}
There are unique holomorphic functions $F_1$ and $F_2$ on $\C\times\C\times\C_\pi$ such that
\begin{equation}\label{fgab}
F_1(\kappa,z,r) = \mathfrak a^\kappa(z|r),\quad F_2(\kappa,z,r) = \mathfrak b^\kappa(z|r),\quad \kappa,z\in\C,\,r>0.
\end{equation}
\end{lemma}
The proof of Lemma~\ref{l_analyt0} is given in Appendix~\ref{app_analyt}.
\medskip

By~(\ref{akappa}) and~(\ref{bkappa}), we have
\begin{equation}\label{reflsym}
\mathfrak a^\kappa(z) = \mathfrak a^{-\kappa}(z),\quad \mathfrak b^\kappa(z) = \mathfrak b^{-\kappa}(z),\quad \kappa,z\in\C.
\end{equation}
For every $\alpha,z\in\C$, we define the functions $\mathcal A^\alpha(z)$ and $\mathcal B^\alpha(z)$ on $\R_+$ by the relations $\mathcal A^\alpha(z) = \mathfrak a^\kappa(z)$ and $\mathcal B^\alpha(z) = \mathfrak b^\kappa(z)$, where $\kappa\in \C$ is such that $\kappa^2 = \alpha$ (by~(\ref{reflsym}), this definition does not depend on the choice of $\kappa$). We therefore have
\begin{equation}\label{Aa}
\mathcal A^{\kappa^2}(z) = \mathfrak a^\kappa(z),\quad \mathcal B^{\kappa^2}(z) = \mathfrak b^\kappa(z),\quad \kappa,z\in\C.
\end{equation}

Equalities~(\ref{eqf}) and~(\ref{akappa}) imply that
\begin{equation}\label{eqa}
-\partial^2_r \mathfrak a^\kappa(z|r) +  \frac{\kappa^2-1/4}{r^2}\mathfrak a^\kappa(z|r) = z\,\mathfrak a^\kappa(z|r),\quad r>0,
\end{equation}
for every $\kappa\in \C\setminus\{0\}$ and $z\in\C$. By Lemma~\ref{l_analyt0}, we can take the limit $\kappa\to 0$ and conclude that (\ref{eqa}) also holds for $\kappa=0$. We hence have $\mathscr L_{\kappa^2,z} \mathfrak a^\kappa(z) = 0$ for all $\kappa,z\in\C$. Since $\mathscr L_{\kappa^2,z} \mathfrak b^\kappa(z) = 0$ for every $\kappa,z\in\C$ by~(\ref{Lu}) and~(\ref{bkappa}), it follows from~(\ref{Aa}) that (\ref{LAB}) is valid for all $\alpha,z\in\C$.

We now use Lemma~\ref{l_analyt0} to prove that the quantities $\mathcal A^\alpha(z|r)$ and $\mathcal B^\alpha(z|r)$ enjoy the same analyticity properties as $\mathfrak a^\kappa(z|r)$ and $\mathfrak b^\kappa(z|r)$.

\begin{lemma}\label{l_analyt}
There are unique holomorphic functions $G_1$ and $G_2$ on $\C\times\C\times\C_\pi$ such that $G_1(\alpha,z,r) = \mathcal A^\alpha(z|r)$ and $G_2(\alpha,z,r) = \mathcal B^\alpha(z|r)$ for every $\alpha,z\in\C$ and $r>0$.
\end{lemma}
\begin{proof}
Let $F_1$ and $F_2$ be as in Lemma~\ref{l_analyt0}. It follows from~(\ref{fgab}), (\ref{reflsym}), and the uniqueness theorem for holomorphic functions that $F_{1,2}(\kappa,z,\zeta) = F_{1,2}(-\kappa,z,\zeta)$ for all $\kappa,z\in\C$ and $\zeta\in \C_\pi$. The existence of $G_1$ and $G_2$ with the required properties is therefore ensured by Lemma~\ref{l_root1}, (\ref{fgab}), and~(\ref{Aa}). The uniqueness of $G_1$ and $G_2$ follows from the uniqueness theorem for holomorphic functions.
\end{proof}

It follows from~(\ref{Ualphatheta}), (\ref{akappa}), (\ref{bkappa}), and~(\ref{Aa}) that
\begin{equation}\label{wkappa}
\mathcal U^{\kappa^2}_\vartheta(z) = \frac{\mathfrak u^{\kappa}(z)\cos(\vartheta-\vartheta_\kappa)-\mathfrak u^{-\kappa}(z)\cos(\vartheta+\vartheta_\kappa)}{\kappa}
\end{equation}
for every $z,\vartheta\in\C$ and $\kappa\in\C\setminus\{0\}$, where $\vartheta_\kappa$ is given by~(\ref{varthetakappa}).

By~(\ref{ukappa}) and~(\ref{Xkappa}), we have
\begin{equation}\nonumber
\overline{\mathfrak u^\kappa(z)} = \mathfrak u^{\bar \kappa}(\bar z),\quad z,\kappa\in\C,
\end{equation}
where the bar means complex conjugation.
In view of~(\ref{akappa}), (\ref{a0}), and~(\ref{bkappa}), this implies that $\mathfrak a^\kappa(z)$ and $\mathfrak b^\kappa(z)$ are real if $z$ is real and $\kappa$ is either real or purely imaginary. Since every $\alpha\in\R$ is equal to $\kappa^2$ for some $\kappa$ that is either real or purely imaginary, it follows from~(\ref{Aa}) that $\mathcal A^\alpha(z)$ and $\mathcal B^\alpha(z)$ are real for every $\alpha,z\in\R$.

If $f,g\in \mathcal D$ are such that $r\to W_r(f,g)$ is a constant function on $\R_+$ (in particular, this is the case when $f$ and $g$ are solutions of $\mathscr L_{\alpha,z}f=\mathscr L_{\alpha,z}g=0$ for some $\alpha,z\in\C$), then its value will be denoted by $W(f,g)$. Equality~(\ref{Lu}) implies that $W_r(\mathfrak u^\kappa(z),\mathfrak u^{-\kappa}(z))$ does not depend on $r$, and we derive from~(\ref{ukappa}) and~(\ref{Xkappa}) that
\begin{equation}\label{Wukappa}
W(\mathfrak u^\kappa(z),\mathfrak u^{-\kappa}(z)) = \lim_{r\downarrow 0} W_r(\mathfrak u^\kappa(z),\mathfrak u^{-\kappa}(z)) = -\frac{2}{\pi}\sin\pi\kappa,\quad \kappa,z\in\C.
\end{equation}
It follows from~(\ref{akappa}), (\ref{bkappa}), and~(\ref{Wukappa}) that $W(\mathfrak a^\kappa(z),\mathfrak b^\kappa(z)) = -2\pi\sinc^2 \pi\kappa$ for all $\kappa\in \C\setminus\{0\}$ and $z\in\C$. By Lemma~\ref{l_analyt0}, $W(\mathfrak a^\kappa(z),\mathfrak b^\kappa(z))$ is continuous in $\kappa$ at $\kappa=0$ and, therefore, this equality holds for all $\kappa,z\in\C$. In view of~(\ref{Aa}), this yields
\begin{equation}\label{WABalpha}
W(\mathcal A^\alpha(z),\mathcal B^\alpha(z)) = -2\pi\Sinc^2(\pi^2\alpha),\quad \alpha,z\in\C.
\end{equation}
Hence, $\mathcal A^\alpha(z)$ and $\mathcal B^\alpha(z)$ are linearly independent for all $\alpha,z\in\C$ such that $\alpha$ is not a square of a nonzero integer number and, in particular, for all $\alpha<1$ and $z\in\C$.

\subsection{Definition of \texorpdfstring{$\mathcal V_{\alpha,\vartheta}$}{V}.} We now turn to the definition of the spectral measures $\mathcal V_{\alpha,\vartheta}$. In what follows, we let $\ln$ denote the branch of the logarithm on $\C_{3\pi/2}$ satisfying the condition $\ln 1 = 0$ and set $z^\rho=e^{\rho\ln z}$ for all $z\in \C_{3\pi/2}$ and $\rho\in\C$.

\begin{lemma}\label{lR}
There is a unique holomorphic function $R$ on $\C\times\C\times \C_{3\pi/2}$ such that
\begin{equation}\label{R}
R(\kappa^2,\vartheta,z) = \frac{z^{-\kappa/2}e^{i\pi\kappa/2}\cos(\vartheta-\vartheta_\kappa) - z^{\kappa/2}e^{-i\pi\kappa/2}\cos(\vartheta+\vartheta_\kappa)}{\kappa}
\end{equation}
for every $\kappa\in\C\setminus\{0\}$, $\vartheta\in\C$ and $z\in \C_{3\pi/2}$, where $\vartheta_\kappa$ is given by~$(\ref{varthetakappa})$. The function $R$ satisfies the equality
\begin{equation}\label{R0}
R(0,\vartheta,z) = (\pi i - \ln z)\cos\vartheta + \pi\sin\vartheta
\end{equation}
for every $\vartheta\in\C$ and $z\in \C_{3\pi/2}$.
\end{lemma}
\begin{proof}
Let the function $\tilde R$ on $\C\times\C\times \C_{3\pi/2}$ be such that $\tilde R(\kappa,\vartheta,z)$ is equal to the right-hand side of~(\ref{R}) for nonzero $\kappa$ and $\tilde R(0,\vartheta,z)$ is equal to the right-hand side of~(\ref{R0}). For every $\vartheta\in\C$ and $z\in \C_{3\pi/2}$, the function $\kappa\to \tilde R(\kappa,\vartheta,z)$ is holomorphic on $\C\setminus\{0\}$ and is continuous at $\kappa=0$ (the calculation of the limit of the right-hand side of~(\ref{R}) shows that $\lim_{\kappa\to 0} \tilde R(\kappa,\vartheta,z) = \tilde R(0,\vartheta,z)$). This implies that the function $\kappa\to \tilde R(\kappa,\vartheta,z)$ is holomorphic on $\C$ for every $\vartheta\in\C$ and $z\in \C_{3\pi/2}$. On the other hand, the function $(\vartheta,z)\to \tilde R(\kappa,\vartheta,z)$ is obviously holomorphic on $\C\times \C_{3\pi/2}$ for every $\kappa\in\C$. By the Hartogs theorem, we conclude that $\tilde R$ is holomorphic on its domain. Moreover, we have $\tilde R(-\kappa,\vartheta,z) = \tilde R(\kappa,\vartheta,z)$ for every $\kappa,\vartheta\in\C$ and $z\in \C_{3\pi/2}$. Hence, the existence of $R$ follows from Lemma~\ref{l_root1}. The uniqueness of $R$ is ensured by the uniqueness theorem for holomorphic functions. Formula~(\ref{R0}) is obvious from the above.
\end{proof}

It follows from~(\ref{R}) and~(\ref{R0}) that
\begin{equation}\label{periodR}
R(\alpha,\vartheta+\pi,z) = -R(\alpha,\vartheta,z)
\end{equation}
for every $\alpha,\vartheta\in\C$ and $z\in\C_{3\pi/2}$.

Given $z\in\C_{3\pi/2}$, there is a unique $\phi\in(-\pi/2,3\pi/2)$ such that $z=|z|e^{i\phi}$. We shall denote this $\phi$ by $\phi_z$.

\begin{lemma}\label{lIm}
Let $R$ be as in Lemma~$\ref{lR}$. Then we have
\begin{equation}\label{ImR}
\Im\left(R(\alpha,\vartheta+\pi/2,z)\overline{R(\alpha,\vartheta,z)}\right) = \pi(\phi_z-\pi)\Sinc((\phi_z-\pi)^2\alpha)\Sinc(\pi^2\alpha)
\end{equation}
for every $\alpha,\vartheta\in\R$ and $z\in\C_{3\pi/2}$.
\end{lemma}
\begin{proof}
By~(\ref{R}), we have
\begin{multline}\nonumber
\kappa^2 R(\kappa^2,\vartheta+\pi/2,Ee^{i\phi})\overline{R(\bar\kappa^2,\vartheta,Ee^{i\phi})} = i \sin(\phi-\pi)\kappa \sin\pi\kappa + \\+\cos(\phi-\pi)\kappa \sin 2\vartheta - E^\kappa \cos(\vartheta+\vartheta_\kappa)\sin(\vartheta+\vartheta_\kappa) - E^{-\kappa}\cos(\vartheta-\vartheta_\kappa)\sin(\vartheta-\vartheta_\kappa)
\end{multline}
for all $\kappa\in\C\setminus\{0\}$, $\vartheta\in \R$, $E>0$, and $-\pi/2<\phi<3\pi/2$. If $\kappa$ is real or purely imaginary, then $\bar\kappa^2=\kappa^2$ and the sum of the last three terms in the right-hand side is real. This implies~(\ref{ImR}) for nonzero $\alpha$. By continuity, (\ref{ImR}) remains valid for $\alpha=0$.
\end{proof}

Let $R$ be as in Lemma~$\ref{lR}$. For every $\alpha,\vartheta\in\C$, we let $\mathcal O_{\alpha,\vartheta}$ denote the open subset of $\C_{3\pi/2}$, where the function $z\to R(\alpha,\vartheta,z)$ is nonzero,
\begin{equation}\label{Oalphatheta}
\mathcal O_{\alpha,\vartheta} = \{z\in \C_{3\pi/2} : R(\alpha,\vartheta,z)\neq 0\}.
\end{equation}
Suppose now that $\alpha<1$ and $\vartheta\in\R$. Then it follows from Lemma~\ref{lIm} that
\begin{equation}\label{subset}
\C_+\cup\R_+ \subset \mathcal O_{\alpha,\vartheta},
\end{equation}
where $\C_+$ denotes the open upper half-plane of the complex plane, $\C_+=\{z\in\C: \Im z>0\}$.
Let the holomorphic function $\mathscr M_{\alpha,\vartheta}$ on $\mathcal O_{\alpha,\vartheta}$ be defined by the equality
\begin{equation}\label{M}
\mathscr M_{\alpha,\vartheta}(z) = -\frac{R(\alpha,\vartheta+\pi/2,z)}{2\pi^2\Sinc^2(\pi^2\alpha)R(\alpha,\vartheta,z)},\quad z\in \mathcal O_{\alpha,\vartheta}.
\end{equation}
Lemma~\ref{lIm} implies that
\begin{equation}\label{ImM}
\Im\mathscr M_{\alpha,\vartheta}(z) = \frac{(\pi-\phi_z)\Sinc((\pi-\phi_z)^2\alpha)}{2\pi\Sinc(\pi^2\alpha)|R(\alpha,\vartheta,z)|^2},\quad z\in \mathcal O_{\alpha,\vartheta}.
\end{equation}
By Lemma~\ref{l_Sinc} and~(\ref{subset}), we conclude that $\Im\mathscr M_{\alpha,\vartheta}(z) > 0$ for every $z\in\C_+$ and, hence, $\mathscr M_{\alpha,\vartheta}|_{\C_+}$ is a Herglotz function (see Appendix~\ref{s_Herglotz}). We now define $\mathcal V_{\alpha,\vartheta}$ as the Herglotz (and, hence, Radon) measure associated with the function $\pi \mathscr M_{\alpha,\vartheta}|_{\C_+}$.
It follows from~(\ref{herglotz_meas}) and Lemma~\ref{l_herglotz} that (\ref{est_meas}) is valid and
\begin{equation}\label{nualphavartheta}
\int \varphi(E)\,d\mathcal V_{\alpha,\vartheta}(E) = \lim_{\eta\downarrow 0} \int \varphi(E) \Im\mathscr M_{\alpha,\vartheta}(E+i\eta)\,dE
\end{equation}
for every continuous complex function $\varphi$ on $\R$ satisfying the bound $|\varphi(E)|\leq C(1+E^2)^{-2}$, $E\in \R$, for some $C\geq 0$. In particular, (\ref{nualphavartheta}) holds for every
continuous function $\varphi$ on $\R$ with compact support. In view of the Riesz representation theorem, this implies that $\mathcal V_{\alpha,\vartheta}$ is uniquely determined by equality~(\ref{nualphavartheta}). It follows from~(\ref{periodR}) and~(\ref{M}) that $\mathscr M_{\alpha,\vartheta+\pi} = \mathscr M_{\alpha,\vartheta}$ and, therefore, $\mathcal V_{\alpha,\vartheta}$ has $\pi$-periodicity property~(\ref{periodicV}).

\section{Explicit formulas for the spectral measures}
\label{s_meas}

In this section, we assume that Theorems~\ref{t_sa} and~\ref{leig2} are valid and obtain explicit expressions for both the point and absolutely continuous parts of the spectral measures $\mathcal V_{\alpha,\vartheta}$. The proofs of Theorems~\ref{t_sa} and~\ref{leig2} in Secs.~\ref{s_one} and~\ref{s_eig} do not rely on the results of this section.

\subsection{General structure of \texorpdfstring{$\mathcal V_{\alpha,\vartheta}$}{V}.}
Given a positive Radon measure $\nu$ on $\R$, we let $L_2^c(\R,\nu)$ denote the subspace of $L_2(\R,\nu)$ consisting of all its elements vanishing $\nu$-a.e. outside some compact subset of $\R$.

\begin{lemma}\label{cor_eig}
Let $\alpha<1$, $\vartheta\in \R$, and $U_{\alpha,\vartheta}$ be as in Theorem~$\ref{leig2}$. Then we have
\begin{equation}\label{v-1}
(U_{\alpha,\vartheta}^{-1}\varphi)(r) = \int \mathcal U^\alpha_\vartheta(E|r) \varphi(E)\,d\mathcal V_{\alpha,\vartheta}(E),
\quad \varphi\in L_2^c(\R,\mathcal V_{\alpha,\vartheta}),
\end{equation}
for $\lambda$-a.e. $r\in \R_+$.
An $E\in\R$ is an eigenvalue of $h_{\alpha,\vartheta}$ if and only if $\mathcal V_{\alpha,\vartheta}(\{E\})>0$. For every eigenvalue $E$, the corresponding eigenspace is one-dimensional and is spanned by $[\mathcal U^\alpha_\vartheta(E)]$, and we have the equality $\|[\mathcal U^\alpha_\vartheta(E)]\| = \mathcal V_{\alpha,\vartheta}(\{E\})^{-1/2}$.
\end{lemma}
\begin{proof}
For brevity, we set $h=h_{\alpha,\vartheta}$, $U = U_{\alpha,\vartheta}$, $\mathcal V = \mathcal V_{\alpha,\vartheta}$, and $\mathcal U = \mathcal U^\alpha_\vartheta$.
Given $\varphi\in L_2^c(\R,\mathcal V)$ and $r>0$, let $\check \varphi(r)$ denote the right-hand side of~(\ref{v-1}). By the unitarity of $U$ and the Fubini theorem, we have
\begin{multline}\nonumber
\langle \psi, U^{-1}\varphi\rangle_{L_2(\R_+)} = \langle U\psi, \varphi\rangle_{L_2(\R,\mathcal V)} = \\=\int_\R d\mathcal V(E) \varphi(E)\int_{\R_+} \overline{\psi(r)} \mathcal U(E|r)\,dr = \int_{\R_+} \overline{\psi(r)}\check \varphi(r)\,dr
\end{multline}
for any $\psi\in L^c_2(\R_+)$, where $\langle\cdot,\cdot\rangle_{L_2(\R_+)}$ and $\langle\cdot,\cdot\rangle_{L_2(\R,\mathcal V)}$ are the scalar products in $L_2(\R_+)$ and $L_2(\R,\mathcal V)$ respectively. This implies (\ref{v-1}). Given $E\in\R$, let $G_E$ be the subspace of $L_2(\R_+)$ composed of all $\psi$ in the domain of $h$ such that $h\psi = E\psi$ and $\tilde G_E$ be the subspace of $L_2(\R,\mathcal V)$ composed of all $\varphi$ in the domain of $\mathcal T^{\mathcal V}_\iota$ such that $\mathcal T^{\mathcal V}_\iota\varphi = E\varphi$, where $\iota$ is the identity function on $\R$. By Theorem~\ref{leig2}, $U$ induces an isomorphism between $G_E$ and $\tilde G_E$ for every $E\in\R$. This means, in particular, that the operators $h$ and $\mathcal T^{\mathcal V}_\iota$ have the same eigenvalues. Hence, $E\in \R$ is an eigenvalue of $h$ if and only if $\mathcal V(\{E\})>0$. If  $\mathcal V(\{E\})>0$, then $\tilde G_E$ is one-dimensional and is spanned by $[\chi_E]_{\mathcal V}$, where $\chi_E$ is the characteristic function of the one-point set $\{E\}$. By~(\ref{v-1}), we have $U^{-1}[\chi_E]_{\mathcal V}=\mathcal V(\{E\})[\mathcal U^\alpha_\vartheta(E)]$. The space $G_E$ is therefore one-dimensional and is spanned by $[\mathcal U^\alpha_\vartheta(E)]$. Since the norm of $[\chi_E]_{\mathcal V}$ in $L_2(\R,\mathcal V)$ is equal to $\mathcal V(\{E\})^{1/2}$, the unitarity of $U$ implies that $\|[\mathcal U^\alpha_\vartheta(E)]\| = \mathcal V_{\alpha,\vartheta}(\{E\})^{-1/2}$.
\end{proof}

As in Sec.~\ref{intro}, let $\omega$ be the $\pi$-periodic function on $\R$ satisfying~(\ref{omega}) and let $Q_0$, $Q_1$, $Q_\infty$, and $Q$ be defined by~(\ref{Q0}), (\ref{Q11}), (\ref{Qi}), and~(\ref{domain}) respectively. We set
\begin{equation}\label{Q1}
Q_{\geq1} = Q_1\cup Q_\infty = \{ (\alpha,\vartheta)\in\R^2 : \alpha<\omega(\vartheta)\}.
\end{equation}
We obviously have
\begin{equation}\label{Q0Q1}
Q_0\cap Q_{\geq1} = \varnothing,\quad Q_0\cup Q_{\geq1} = Q.
\end{equation}
By~(\ref{omega}), we have
\begin{equation}\label{coco}
\Cos\frac{\pi^2 \omega(\vartheta)}{4} = \cos\left(\frac{\pi}{2}-|\vartheta|\right) = |\sin\vartheta|
\end{equation}
for every $\vartheta\in [-\pi/2,\pi/2]$. Since both sides of~(\ref{coco}) are $\pi$-periodic, (\ref{coco}) remains valid for all $\vartheta\in\R$. By Lemma~\ref{l_Sinc}, the function $\alpha\to \Cos(\pi^2\alpha/4)$ is nonnegative and strictly decreasing on $(-\infty,1]$. Hence, for every $\alpha<1$ and $\vartheta\in\R$, we have the chain of equivalent conditions
\[
\alpha < \omega(\vartheta) \Leftrightarrow \Cos \frac{\pi^2\omega(\vartheta)}{4} < \Cos\frac{\pi^2\alpha}{4} \Leftrightarrow |\sin\vartheta| < \Cos\frac{\pi^2\alpha}{4} \Leftrightarrow \sin^2\vartheta < \Cos^2\frac{\pi^2\alpha}{4}.
\]
In view of~(\ref{Q0}) and~(\ref{Q1}), it follows that
\begin{align}
& Q_0 = \{ (\alpha,\vartheta)\in\R^2 : \alpha<1\mbox{ and }\sin^2\vartheta\geq \Cos^2(\pi^2\alpha/4)\},\label{Q0'}\\
& Q_{\geq1} = \{ (\alpha,\vartheta)\in\R^2 : \alpha<1\mbox{ and }\sin^2\vartheta < \Cos^2(\pi^2\alpha/4)\}.\label{Q1'}
\end{align}
Given $\alpha,\vartheta\in\C$, we set
\begin{equation}\label{Nalphatheta}
\Sigma_{\alpha,\vartheta} = \{ E < 0 : R(\alpha,\vartheta,E) = 0\},
\end{equation}
where $R$ is as in Lemma~$\ref{lR}$. By~(\ref{R}), we have
\begin{equation}\label{eig_eq}
E\in \Sigma_{\kappa^2,\vartheta} \Longleftrightarrow E<0 \mbox{ and }|E|^{-\kappa/2}\cos(\vartheta-\vartheta_\kappa) = |E|^{\kappa/2}\cos(\vartheta+\vartheta_\kappa)
\end{equation}
for every $\kappa\in\C\setminus\{0\}$ and $\vartheta\in\C$.

In view of~(\ref{Oalphatheta}) and~(\ref{subset}), we have $R(\alpha,\vartheta,E)\neq 0$ for every $\alpha<1$, $\vartheta\in \R$ and $E>0$. For every $\alpha<1$ and $\vartheta\in\R$, we define the function $t_{\alpha,\vartheta}$ by the formula
\begin{equation}\label{density}
t_{\alpha,\vartheta}(E) = \left\{
\begin{matrix}
\frac{1}{2|R(\alpha,\vartheta,E)|^2},& E > 0,\\
0,& E \leq 0.
\end{matrix}
\right.
\end{equation}

Given a positive Radon measure $\nu$ on $\R$, we set
\[
\mathcal P(\nu) = \{ E\in\R: \nu(\{E\}) > 0\}.
\]
Since $\nu$ is $\sigma$-additive, the set $\mathcal P(\nu)$ is at most countable. We define the continuous part $\nu^c$ and the point part $\nu^p$ of $\nu$ by the relations
\begin{equation}\label{nuc}
\nu^c = (1-\chi_{\mathcal P(\nu)})\,\nu, \quad \nu^p = \chi_{\mathcal P(\nu)}\,\nu,
\end{equation}
where $\chi_{\mathcal P(\nu)}$ is the characteristic function of the set $\mathcal P(\nu)$ (i.e., it is equal to unity on $\mathcal P(\nu)$ and to zero on $\R\setminus\mathcal P(\nu)$). Clearly, $\nu^c(\{E\})=0$ for every $E\in\R$ and
\[
\nu = \nu^c + \nu^p.
\]
A function $\varphi$ on $\R$ is $\nu^p$-integrable if and only if $\{\nu(\{E\})\varphi(E)\}_{E\in\mathcal P(\nu)}$ is a summable family in $\C$, in which case we have
\begin{equation}\label{nup}
\int \varphi(E)\,d\nu^p(E) = \sum_{E\in \mathcal P(\nu)} \nu(\{E\})\varphi(E).
\end{equation}
Thus, to completely describe a positive Radon measure $\nu$ on $\R$, it suffices to find $\mathcal P(\nu)$ and $\nu^c$ and specify $\nu(\{E\})$ for every $E\in\mathcal P(\nu)$. The next theorem gives such a description for the measures $\mathcal V_{\alpha,\vartheta}$. As in Sec.~\ref{intro}, we let $\lambda$ denote the Lebesgue measure on $\R$.

\begin{theorem}\label{t_measure}
For every $\alpha<1$ and $\vartheta\in\R$, we have $\mathcal P(\mathcal V_{\alpha,\vartheta}) = \Sigma_{\alpha,\vartheta}$ and $\mathcal V^c_{\alpha,\vartheta} = t_{\alpha,\vartheta}\,\lambda$, where $t_{\alpha,\vartheta}$ is given by~$(\ref{density})$. If $(\alpha,\vartheta)\in Q_0$, then $\Sigma_{\alpha,\vartheta} = \varnothing$. If $(\alpha,\vartheta)\in Q_{\geq1}$ and $E\in \Sigma_{\alpha,\vartheta}$, then
\begin{equation}\label{point}
\mathcal V_{\alpha,\vartheta}(\{E\}) = \frac{|E|}{2\Sinc(\pi^2\alpha)(\Cos^2(\pi^2\alpha/4)-\sin^2\vartheta)}.
\end{equation}
\end{theorem}

\begin{corollary}\label{cor_measure}
Let $\alpha<1$ and $\vartheta\in \R$. Then the set of eigenvalues of $h_{\alpha,\vartheta}$ is precisely $\Sigma_{\alpha,\vartheta}$. For every $E\in \Sigma_{\alpha,\vartheta}$, the corresponding eigenspace is one-dimensional and is spanned by $[\mathcal U^\alpha_\vartheta(E)]$, and we have
\[
\int_0^\infty \mathcal U^\alpha_\vartheta(E|r)^2\,dr =  2|E|^{-1}\Sinc(\pi^2\alpha)(\Cos^2(\pi^2\alpha/4)-\sin^2\vartheta).
\]
\end{corollary}
\begin{proof}
The statement follows directly from Lemma~\ref{cor_eig} and Theorem~\ref{t_measure}.
\end{proof}

To prove Theorem~\ref{t_measure}, we shall need several auxiliary lemmas.

\begin{lemma}\label{l_nonint}
Let $\kappa\in\C$ be such that $-1<|\Re\,\kappa|<1$ and $f\in \mathcal D$ be a nontrivial solution of $\mathscr L_{\kappa^2} f = 0$. Then $f$ is not square-integrable on $\R_+$.
\end{lemma}
\begin{proof}
If $\mathscr L_0 f = 0$, then there exist $c_1, c_2\in\C$ such that $f(r) = c_1r^{1/2} + c_2 r^{1/2}\ln r$ for all $r>0$. It is straightforward to verify that such a function is square-integrable on $\R_+$ if and only if $c_1=c_2=0$. This proves our statement for $\kappa=0$. If $\kappa=\kappa'+i\kappa''$ is nonzero and $\mathscr L_{\kappa^2} f = 0$, then there exist $c_1, c_2\in\C$ such that $f(r) = c_1r^{1/2+\kappa} + c_2 r^{1/2-\kappa}$ for all $r>0$. Since $|\kappa'|<1$, $f$ is square-integrable on $(0,r]$ for every $r>0$, and we have
\begin{equation}\label{int2}
\int_0^r|f(r')|^2 dr' = \frac{g(r)}{2} + \Re\frac{c_1\bar c_2 r^{2+2i\kappa''}}{1+i\kappa''}
\end{equation}
for every $r>0$, where the function $g$ on $\R_+$ is given by
\begin{equation}\label{g(r)}
g(r) = \frac{|c_1|^2r^{2+2\kappa'}}{1+\kappa'} + \frac{|c_2|^2r^{2-2\kappa'}}{1-\kappa'},\quad r>0.
\end{equation}
Applying the inequality $2ab\leq a^2 + b^2$ to $g(r)$, we obtain
\[
\frac{g(r)}{2}\geq \frac{|c_1 c_2| r^2}{\sqrt{1-\kappa'^2}}= \frac{1}{\sigma}\left|\frac{c_1\bar c_2 r^{2+2i\kappa''}}{1+i\kappa''}\right|,\quad r>0,
\]
where $\sigma = [(1+\kappa''^2)^{-1}(1-\kappa'^2)]^{1/2}$. In view of~(\ref{int2}), we conclude that
\begin{equation}\label{int2'}
\int_0^r|f(r')|^2 dr' \geq \frac{1-\sigma}{2}g(r)
\end{equation}
for every $r>0$. Since $\kappa\neq 0$, we have $\sigma<1$. If $f$ is nontrivial, then $c_1$ and $c_2$ are not both zero and it follows from~(\ref{g(r)}) that $g(r)\to\infty$ as $r\to\infty$. By~(\ref{int2'}), this implies that $\int_0^r|f(r')|^2 dr'\to\infty$ as $r\to\infty$ and, hence, $f$ is not square-integrable on $\R_+$.
\end{proof}

In what follows, we set $\R_- = (-\infty,0)$.

\begin{lemma}\label{l_c}
Let $\alpha<1$, $\vartheta\in\R$, and $\chi$ be the characteristic function of $\Sigma_{\alpha,\vartheta}$. Then we have
$(1-\chi)\mathcal V_{\alpha,\vartheta} = t_{\alpha,\vartheta}\,\lambda$.
\end{lemma}
\begin{proof}
Let $\nu = (1-\chi)\mathcal V_{\alpha,\vartheta}$ and $\nu' = t_{\alpha,\vartheta}\lambda$. Let $O$ be the open subset of $\R$ defined by the relation $O = \R_+\cup (\R_-\setminus \Sigma_{\alpha,\vartheta})$. As $\R = O\cup \Sigma_{\alpha,\vartheta}\cup\{0\}$, it suffices to show that $\nu$ and $\nu'$ have the same restrictions to each of the sets $O$, $\Sigma_{\alpha,\vartheta}$, and $\{0\}$. Since the functions $t_{\alpha,\vartheta}$ and $1-\chi$ are locally bounded on $O$ and $\mathcal V_{\alpha,\vartheta}$ and $\lambda$ are Radon measures on $\R$, the restrictions of $\nu$ and $\nu'$ to $O$ are Radon measures on $O$. By~(\ref{Oalphatheta}), (\ref{subset}), and~(\ref{Nalphatheta}), we have $O\subset \mathcal O_{\alpha,\vartheta}$, and it follows from~(\ref{ImM}) and~(\ref{density}) that $\Im\mathscr M_{\alpha,\vartheta}(E) = t_{\alpha,\vartheta}(E)$ for every $E\in O$. Since $1-\chi$ is equal to unity on $O$, (\ref{nualphavartheta}) and the dominated convergence theorem imply that
\[
\int \varphi(E)\,d\nu(E) =  \int \varphi(E)\,d\mathcal V_{\alpha,\vartheta}(E) = \int t_{\alpha,\vartheta}(E)\varphi(E)\,dE = \int \varphi(E)\,d\nu'(E)
\]
for every continuous function $\varphi$ on $\R$ such that $\mathrm{supp}\,\varphi$ is a compact subset of $O$.
By the Riesz representation theorem, we conclude that $\nu|_O = \nu'|_O$. Because $1-\chi$ and $t_{\alpha,\vartheta}$ vanish on $\Sigma_{\alpha,\vartheta}$, both $\nu$ and $\nu'$ have zero restrictions to $\Sigma_{\alpha,\vartheta}$. We now note that $\mathcal V_{\alpha,\vartheta}(\{0\})=0$ because otherwise $\mathcal U^\alpha_\vartheta(0)$ would be a nontrivial square-integrable function on $\R_+$ by Lemma~\ref{cor_eig}, in contradiction to~(\ref{lalpha}) and Lemma~\ref{l_nonint}.
Since $\nu'(\{0\}) = 0$ and $\nu(\{0\}) = \mathcal V_{\alpha,\vartheta}(\{0\})$, we conclude that $\nu(\{0\}) = \nu'(\{0\})$.
\end{proof}

Using elementary trigonometric transformations, we find that
\begin{equation}\label{trigid}
\cos(\vartheta-\vartheta_\kappa)\cos(\vartheta+\vartheta_\kappa) = \cos^2(\pi\kappa/2) - \sin^2\vartheta = \Cos^2(\pi^2\kappa^2/4) - \sin^2\vartheta
\end{equation}
for all $\kappa,\vartheta\in\C$, where $\vartheta_\kappa$ is given by~(\ref{varthetakappa}). In view of~(\ref{Q1'}), this equality implies that
\begin{equation}\label{coscos}
\cos(\vartheta-\vartheta_\kappa)\cos(\vartheta+\vartheta_\kappa) > 0
\end{equation}
for all $\vartheta\in\R$ and $\kappa\in\C$ such that $(\kappa^2,\vartheta)\in Q_{\geq1}$.

\begin{lemma}\label{l_empty}
$\Sigma_{\alpha,\vartheta} = \varnothing$ for every $(\alpha,\vartheta)\in Q_0$.
\end{lemma}
\begin{proof}
Let $(\alpha,\vartheta)\in Q_0$. Suppose first that $\alpha=0$. By~(\ref{Q0'}), we have $\vartheta = \pi/2 + \pi k$ for some $k\in\Z$. Equality~(\ref{R0}) therefore implies that $R(\alpha,\vartheta,z) = (-1)^k\pi$ for every $z\in\C_{3\pi/2}$. This means that $\Sigma_{\alpha,\vartheta}=\varnothing$. Now let $\alpha\neq 0$. Since $\omega$ is nonnegative, it follows from~(\ref{Q0}) that $0<\alpha<1$ and, hence, $\alpha = \kappa^2$ for some $0<\kappa<1$. Suppose $\Sigma_{\alpha,\vartheta}\neq\varnothing$ and $E\in \Sigma_{\alpha,\vartheta}$. By~(\ref{eig_eq}) and~(\ref{trigid}), it follows that
\begin{equation}\label{eig_cond}
\Cos^2(\pi^2\alpha/4) - \sin^2\vartheta = \frac{1}{2}(|E|^{\kappa}\cos^2(\vartheta+\vartheta_\kappa) + |E|^{-\kappa}\cos^2(\vartheta-\vartheta_\kappa)).
\end{equation}
Since $\kappa$ is real, the right-hand side of~(\ref{eig_cond}) is nonnegative and can be zero only if $\cos(\vartheta-\vartheta_\kappa)=\cos(\vartheta+\vartheta_\kappa)=0$ and, hence, only if $\kappa\in\Z$. The condition $0<\kappa<1$ therefore implies that the right-hand side of~(\ref{eig_cond}) is strictly positive. In view of~(\ref{Q0'}), this contradicts the assumption that $(\alpha,\vartheta)\in Q_0$. Hence, $\Sigma_{\alpha,\vartheta}=\varnothing$.
\end{proof}

\begin{lemma}\label{l_residue}
Let $(\alpha,\vartheta)\in Q_{\geq1}$ and $E\in \Sigma_{\alpha,\vartheta}$. Then we have
\[
\lim_{z\to E} (z-E)\mathscr M_{\alpha,\vartheta}(z) = -\frac{|E|}{2\pi\Sinc(\pi^2\alpha)(\Cos^2(\pi^2\alpha/4)-\sin^2\vartheta)}.
\]
\end{lemma}
\begin{proof}
We separately consider the cases $\alpha\neq 0$ and $\alpha=0$.
\par\medskip\noindent
1. Let $\alpha\neq 0$ and $\kappa\in \C$ be such that $\kappa^2=\alpha$.
By~(\ref{R}) and~(\ref{eig_eq}), we have
\begin{multline}\label{partialR}
\partial_z R(\alpha,\vartheta,z)|_{z=E} = \frac{|E|^{-\kappa/2}\cos(\vartheta-\vartheta_\kappa) + |E|^{\kappa/2}\cos(\vartheta+\vartheta_\kappa)}{2|E|} =\\= |E|^{\kappa/2-1}\cos(\vartheta+\vartheta_\kappa).
\end{multline}
It follows from~(\ref{coscos}) that $\cos(\vartheta+\vartheta_\kappa)\neq 0$ and, hence, $\partial_z R(\alpha,\vartheta,z)|_{z=E}\neq 0$. By~(\ref{R}), we have
\[
R(\alpha,\vartheta+\pi/2,E) = \frac{|E|^{\kappa/2}\sin(\vartheta+\vartheta_\kappa) - |E|^{-\kappa/2}\sin(\vartheta-\vartheta_\kappa)}{\kappa}.
\]
Multiplying the numerator and denominator by $\cos(\vartheta-\vartheta_\kappa)$ (which is nonzero by~(\ref{coscos})) and using~(\ref{eig_eq}), we obtain
\begin{equation}\label{Rpi}
R(\alpha,\vartheta+\pi/2,E) = \frac{\pi|E|^{\kappa/2}\sinc\pi\kappa}{\cos(\vartheta-\vartheta_\kappa)}.
\end{equation}
In view of~(\ref{M}), we have
\[
\lim_{z\to E} (z-E)\mathscr M_{\alpha,\vartheta}(z) = -\frac{R(\alpha,\vartheta+\pi/2,E)}{2\pi^2\Sinc^2(\pi^2\alpha)\partial_z R(\alpha,\vartheta,z)|_{z=E}}.
\]
Combining this formula with~(\ref{trigid}), (\ref{partialR}), and~(\ref{Rpi}), we arrive at the required equality.
\par\medskip\noindent
2. Let $\alpha=0$. Since $E\in \Sigma_{0,\vartheta}$, it follows from~(\ref{R0}) that
\begin{align}
&\ln |E| \cos\vartheta = \pi\sin\vartheta, \label{eig_eq0} \\
& \partial_z R(0,\vartheta,z)|_{z=E} = \frac{\cos\vartheta}{|E|}.\label{partialR0}
\end{align}
Since $(0,\vartheta)\in Q_{\geq1}$ and $\Cos(0) = 1$, (\ref{Q1'}) implies that $\cos\vartheta\neq 0$ and, hence, $\partial_z R(0,\vartheta,z)|_{z=E}\neq 0$. By~(\ref{R0}) and~(\ref{eig_eq0}), we obtain
\begin{equation}\label{R0pi}
R(0,\vartheta+\pi/2,E) = \ln |E|\sin\vartheta +\pi\cos\vartheta = \frac{\pi}{\cos\vartheta}.
\end{equation}
In view of~(\ref{M}) and the equality $\Sinc(0)=1$, we have
\[
\lim_{z\to E} (z-E)\mathscr M_{0,\vartheta}(z) =- \frac{R(0,\vartheta+\pi/2,E)}{2\pi^2\partial_z R(0,\vartheta,z)|_{z=E}}.
\]
Combining this formula with~(\ref{partialR0}) and~(\ref{R0pi}) yields the required result.
\end{proof}

\begin{proof}[Proof of Theorem~$\ref{t_measure}$]
Let $(\alpha,\vartheta)\in Q_{\geq1}$ and $E\in \Sigma_{\alpha,\vartheta}$. Let $O=\mathcal O_{\alpha,\vartheta}\cup \{E\}$, where $\mathcal O_{\alpha,\vartheta}$ is given by~(\ref{Oalphatheta}). Clearly, $O$ is an open subset of $\C$ containing $E$. By Lemma~\ref{l_residue}, there exists a holomorphic function $g$ on $O$ such that
\begin{equation}\label{residue}
\mathscr M_{\alpha,\vartheta}(z) = -\frac{1}{\pi}\frac{A}{z-E} + g(z)
\end{equation}
for all $z\in \mathcal O_{\alpha,\vartheta}$, where $A$ denotes the right-hand side of~(\ref{point}). By~(\ref{ImM}), $\mathscr M_{\alpha,\vartheta}$ is real on $\mathcal O_{\alpha,\vartheta}\cap \R_-$ and, therefore, $g$ is real on $O\cap\R_-$. By~(\ref{nualphavartheta}), (\ref{residue}) and the dominated convergence theorem, we conclude that $\int\varphi(E')\,d\mathcal V_{\alpha,\vartheta}(E')$ is equal to $A\varphi(E)$ for every continuous function $\varphi$ on $\R$ such that $\mathrm{supp}\,\varphi$ is a compact subset of $O\cap\R_-$. Hence, $\mathcal V_{\alpha,\vartheta}(\{E\})=A$. Thus, formula~(\ref{point}) holds for every $(\alpha,\vartheta)\in Q_{\geq1}$ and $E\in \Sigma_{\alpha,\vartheta}$. This implies, in particular, that $\Sigma_{\alpha,\vartheta}\subset \mathcal P(\mathcal V_{\alpha,\vartheta})$ for all $(\alpha,\vartheta)\in Q_{\geq1}$. By Lemma~\ref{l_empty}, we have $\Sigma_{\alpha,\vartheta}=\varnothing$ for all $(\alpha,\vartheta)\in Q_0$. It follows that $\Sigma_{\alpha,\vartheta}\subset \mathcal P(\mathcal V_{\alpha,\vartheta})$ for all $\alpha<1$ and $\vartheta\in\R$. Since the opposite inclusion also holds by Lemma~\ref{l_c}, we conclude that $\Sigma_{\alpha,\vartheta}= \mathcal P(\mathcal V_{\alpha,\vartheta})$ for all $\alpha<1$ and $\vartheta\in\R$. The equality $\mathcal V^c_{\alpha,\vartheta} = t_{\alpha,\vartheta}\,\lambda$ now follows from~(\ref{nuc}) and Lemma~\ref{l_c}.
\end{proof}

Theorem~\ref{t_measure} implies, in particular, that $t_{\alpha,\vartheta}$ is a locally integrable function on $\R$ for every $\alpha<1$ and $\vartheta\in\R$. It is noteworthy that we established this property of $t_{\alpha,\vartheta}$ without explicitly estimating this function. Instead, we relied on the fact that $\mathcal V_{\alpha,\vartheta}$ is a Radon measure, which follows from its definition as a Herglotz measure. In Lemma~\ref{l_R1} below, we shall obtain an explicit estimate for $|R(\alpha,\vartheta,E)|^{-1}$ that, when substituted in~(\ref{density}), immediately implies the local integrability of $t_{\alpha,\vartheta}$.

\subsection{Eigenvalues of \texorpdfstring{$h_{\alpha,\vartheta}$}{h}.}

We now turn to obtaining an explicit description of the set $\Sigma_{\alpha,\vartheta}$ of eigenvalues of $h_{\alpha,\vartheta}$ for every $\alpha<1$ and $\vartheta\in\R$. To this end, it is convenient to use the logarithmic scale and pass from the set $\Sigma_{\alpha,\vartheta}$ to its inverse image $N_{\alpha,\vartheta}$ under the map $s\to -e^s$ from $\R$ to itself,
\begin{equation}\label{Sigma}
N_{\alpha,\vartheta} = \{ s\in\R: -e^s\in \Sigma_{\alpha,\vartheta}\}.
\end{equation}
We thus have $E\in \Sigma_{\alpha,\vartheta}$ if and only if $E<0$ and $\ln|E|\in N_{\alpha,\vartheta}$. Further, we define the open subsets $W_0$ and $W$ of $\R^2$ by the relations
\begin{align}
&W_0 = \{(\alpha,\vartheta)\in Q_1: \alpha>0 \mbox{ and } -\pi/2<\vartheta<\pi/2\}, \label{W0} \\
&W = \{(\alpha,\vartheta)\in Q_1:  -\pi/2<\vartheta<\pi/2\} \cup Q_\infty.\label{W}
\end{align}
Hence, $W_0$ is the interior of the central dark gray curvilinear triangular region in Fig.~\ref{f0}.
Let $(\alpha,\vartheta)\in W_0$. Since $W_0\subset Q_{\geq1}$, inequality~(\ref{coscos}) for $\kappa=\sqrt{\alpha}$ implies that $\cos(\vartheta+\pi\sqrt{\alpha}/2)\neq 0$ and
\begin{equation}\label{>0}
\frac{\cos(\vartheta-\pi\sqrt{\alpha}/2)}{\cos(\vartheta+\pi\sqrt{\alpha}/2)}>0.
\end{equation}
Hence, we can define a real-valued function $S_0$ on $W_0$ by the formula
\begin{equation}\label{S0}
S_0(\alpha,\vartheta) = \frac{1}{\sqrt{\alpha}} \ln \frac{\cos(\vartheta-\pi\sqrt{\alpha}/2)}{\cos(\vartheta+\pi\sqrt{\alpha}/2)}.
\end{equation}
Applying~(\ref{eig_eq}) to $\kappa=\sqrt{\alpha}$ and using~(\ref{>0}), we conclude that $-\exp(S_0(\alpha,\vartheta))$ belongs to $\Sigma_{\alpha,\vartheta}$ for every $(\alpha,\vartheta)\in W_0$. In view of~(\ref{Sigma}), this means that
\begin{equation}\label{S0N}
S_0(\alpha,\vartheta) \in N_{\alpha,\vartheta},\quad (\alpha,\vartheta)\in W_0.
\end{equation}
We shall see that the set $N_{\alpha,\vartheta}$ for every $(\alpha,\vartheta)\in Q_{\geq1}$ can actually be completely described in terms of the analytic continuation of $S_0$ from $W_0$ to $W$. To construct such an analytic continuation, we calculate the derivative of $S_0(\alpha,\vartheta)$ with respect to $\vartheta$. In view of~(\ref{trigid}), we find that
\begin{equation}\label{Sder}
\partial_\vartheta S_0(\alpha,\vartheta) = \frac{\pi\Sinc(\pi^2\alpha)}{\Cos^2(\pi^2\alpha/4)-\sin^2\vartheta},\quad (\alpha,\vartheta)\in W_0.
\end{equation}
We now observe that the right-hand side of~(\ref{Sder}) is actually well-defined and real-analytic on the entire domain $W$. The real-analytic continuation of $S_0$ to $W$ can therefore be obtained by integrating the right-hand side of~(\ref{Sder}). This argument is central to the proof of the next result.

\begin{lemma}\label{l_S}
Let $W_0$ and $W$ be given by~$(\ref{W0})$ and~$(\ref{W})$ respectively and the function $S_0$ on $W_0$ be defined by~$(\ref{S0})$. There is a unique real-analytic function $S$ on $W$ such that $S|_{W_0} = S_0$.
For every $\vartheta\in (-\pi/2,\pi/2)$, we have
\begin{align}
& S(0,\vartheta) = \pi\tg\vartheta,\label{S=}\\
& S(\alpha,\vartheta) = \frac{2}{\sqrt{|\alpha|}}\arctg\left(\tg\vartheta\thyp\frac{\pi\sqrt{|\alpha|}}{2}\right),\quad \alpha<0.\label{S<}
\end{align}
For all $\vartheta\in\R$ and $\alpha<0$, we have
\begin{equation}\label{Sshift}
S(\alpha,\vartheta+\pi) = S(\alpha,\vartheta) + \frac{2\pi}{\sqrt{|\alpha|}}.
\end{equation}
\end{lemma}

The graph of the function $S$ described by Lemma~\ref{l_S} is shown in Fig.~\ref{f_S}.

\begin{figure}
  \includegraphics[width=\linewidth]{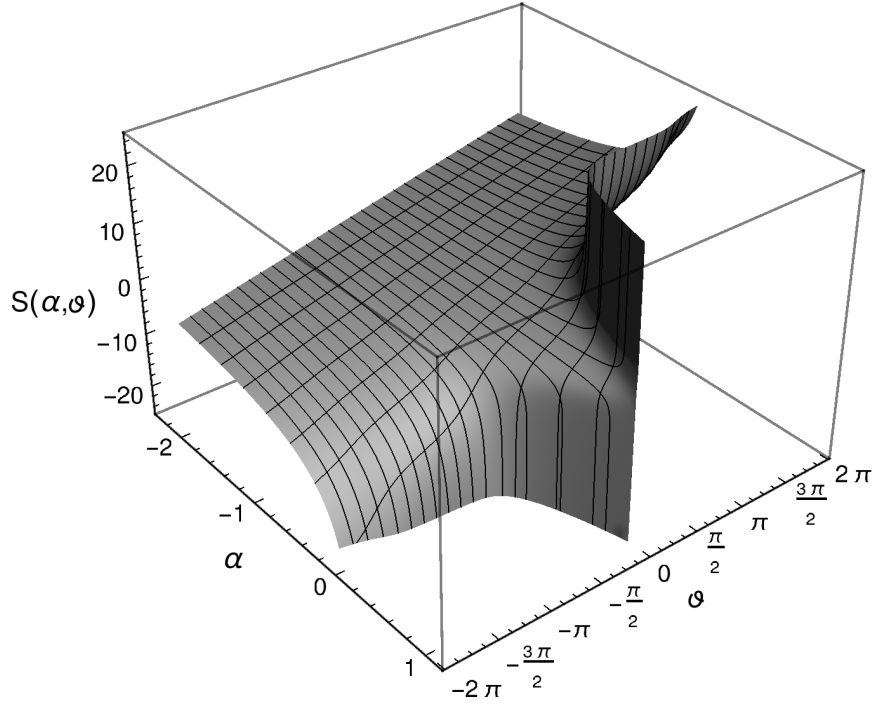}
  \caption{The function $S$ is plotted using formulas~(\ref{S0}), (\ref{S=}), (\ref{S<}), and~(\ref{Sshift}).}
  \label{f_S}
\end{figure}

The proof of Lemma~\ref{l_S} relies on the next auxiliary statement.

\begin{lemma}\label{l_realanalyt}
Let $x_0,\xi_0\in\R$, $a,b>0$, and $f$ be a real-analytic function on the rectangle
\[
R_{a,b}(x_0,\xi_0)=\{(x,\xi)\in\R^2: |x-x_0|<a \mbox{ and }|\xi-\xi_0|<b\}.
\]
Then $(x,\xi)\to \int_{\xi_0}^\xi f(x,\xi')\,d\xi'$ is a real-analytic function on $R_{a,b}(x_0,\xi_0)$.
\end{lemma}
\begin{proof}
Since $f$ is real-analytic on $R_{a,b}(x_0,\xi_0)$, there are an open set $O\subset \C^2$ and a holomorphic function $\tilde f$ on $O$ such that $R_{a,b}(x_0,\xi_0)\subset O$ and $f$ is the restriction of $\tilde f$ to $R_{a,b}(x_0,\xi_0)$. Let $F$ denote the function $(x,\xi)\to \int_{\xi_0}^\xi f(x,\xi')\,d\xi'$ on $R_{a,b}(x_0,\xi_0)$.
Fix $0<a'<a$ and $0<b'<b$. There exist open subsets $O'$ and $O''$ of $\C$ such that $R_{a',b'}(x_0,\xi_0)\subset O'\times O''\subset O$. Moreover, we can assume that $O''$ is convex. We define the function $\tilde F$ on $O'\times O''$ by the formula
\[
\tilde F(z,\zeta) = \int_0^1 (\zeta-\xi_0)\tilde f(z,\xi_0+(\zeta-\xi_0)t)\,dt,\quad z\in O',\zeta\in O''.
\]
Clearly, $\tilde F$ is holomorphic on $O'\times O''$ and coincides with $F$ on $R_{a',b'}(x_0,\xi_0)$. This means that $F$ is real-analytic on $R_{a',b'}(x_0,\xi_0)$. Since $a'$ and $b'$ can be chosen arbitrarily close to $a$ and $b$, we conclude that $F$ is real-analytic on $R_{a,b}(x_0,\xi_0)$.
\end{proof}

\begin{proof}[Proof of Lemma~$\ref{l_S}$]
We define the function $S$ on $W$ by the formula
\begin{equation}\label{Sdef}
S(\alpha,\vartheta) = \pi\Sinc(\pi^2\alpha)\int_0^\vartheta \frac{d\vartheta'}{\Cos^2(\pi^2\alpha/4)-\sin^2\vartheta'},\quad (\alpha,\vartheta)\in W.
\end{equation}
Since $S_0(\alpha,0) = 0$ for all $0<\alpha<1$ by~(\ref{S0}), it follows from~(\ref{Sder}) and~(\ref{Sdef}) that $S_0 = S|_{W_0}$. For every $(\alpha,\vartheta)\in W$, there is $\varepsilon>0$ such that $W$ contains the rectangle
\[
\{(\alpha',\vartheta')\in\R^2: |\alpha'-\alpha|<\varepsilon \mbox{ and } |\vartheta'| < |\vartheta|+\varepsilon\}.
\]
It therefore follows from Lemma~\ref{l_realanalyt} and~(\ref{Sdef}) that $S$ is real-analytic in a neighborhood of every point of $W$. This means that $S$ is real-analytic on $W$. Given $\alpha<0$ and $-\pi/2<\vartheta<\pi/2$, we let $A(\alpha,\vartheta)$ denote the right-hand side of~(\ref{S<}) and set $\sigma_\alpha = \sqrt{|\alpha|}$. Using~(\ref{Cos4}), we find for every $\alpha<0$ that
\[
\partial_\vartheta A(\alpha,\vartheta) = \frac{\sh(\pi\sigma_\alpha)}{\sigma_\alpha(\ch^2(\pi\sigma_\alpha/2) - \sin^2\vartheta)} = \frac{\pi\Sinc(\pi^2\alpha)}{\Cos^2(\pi^2\alpha/4)-\sin^2\vartheta},\quad |\vartheta|<\pi/2.
\]
In view of~(\ref{Sdef}) and the equality $A(\alpha,0) = 0$, this implies that $S(\alpha,\vartheta) = A(\alpha,\vartheta)$ for all $\alpha<0$ and $-\pi/2<\vartheta<\pi/2$, i.e., (\ref{S<}) holds. Formula~(\ref{S=}) follows immediately from~(\ref{Sdef}) for $\alpha=0$. Since $\vartheta\to (\Cos^2(\pi^2\alpha/4)-\sin^2\vartheta)^{-1}$ is a continuous $\pi$-periodic function on $\R$ for every $\alpha<0$, it follows from~(\ref{Sdef}) that
\begin{multline}\nonumber
S(\alpha,\vartheta+\pi) = S(\alpha,\vartheta) + \pi\Sinc(\pi^2\alpha)\int_{-\pi/2}^{\pi/2} \frac{d\vartheta'}{\Cos^2(\pi^2\alpha/4)-\sin^2\vartheta'} =\\= S(\alpha,\vartheta) + S(\alpha,\pi/2) - S(\alpha,-\pi/2)
\end{multline}
for all $\alpha<0$ and $\vartheta\in\R$. This implies~(\ref{Sshift}) because $S(\alpha,\pm\pi/2) = \pm\pi/\sqrt{|\alpha|}$ by~(\ref{S<}) and the continuity of $S$. The uniqueness of $S$ follows from the uniqueness theorem for holomorphic functions.
\end{proof}

\begin{theorem}\label{t_N}
Let the function $S$ on $W$ be as in Lemma~$\ref{l_S}$. For every $\alpha<1$ and $\vartheta\in \R$, $N_{\alpha,\vartheta}$ is equal to the set
\begin{multline}\label{eq_N}
\{s\in\R: s = S(\alpha,\vartheta+\pi k)\mbox{ for some }k\in\Z \mbox{ such that } (\alpha,\vartheta+\pi k)\in W\}.
\end{multline}
\end{theorem}
\begin{proof}
Given $\alpha<1$ and $\vartheta\in \R$, we let $\tilde N_{\alpha,\vartheta}$ denote the set~(\ref{eq_N}). We have to prove that
\begin{equation}\label{N=N}
N_{\alpha,\vartheta} = \tilde N_{\alpha,\vartheta}
\end{equation}
for all $\alpha<1$ and $\vartheta\in \R$. By~(\ref{Sigma}) and~(\ref{S0N}), the set $\Sigma_{\alpha,\vartheta}$ contains $-e^{S_0(\alpha,\vartheta)}$ for every $(\alpha,\vartheta)\in W_0$. Since $S$ coincides with $S_0$ on $W_0$, it follows from~(\ref{Nalphatheta}) that
\begin{equation}\label{Rzero}
R(\alpha,\vartheta,-e^{S(\alpha,\vartheta)}) = 0
\end{equation}
for all $(\alpha,\vartheta)\in W_0$. By Lemma~\ref{l_S}, the left-hand side of~(\ref{Rzero}) is a real-analytic function of $(\alpha,\vartheta)$ on $W$. In view of the uniqueness theorem for holomorphic functions, this implies that (\ref{Rzero}) remains valid for all $(\alpha,\vartheta)\in W$. Let $\alpha<1$, $\vartheta\in\R$, and $s\in \tilde N_{\alpha,\vartheta}$. Then there is $k\in \Z$ such that $(\alpha,\vartheta+\pi k)\in W$ and $s = S(\alpha,\vartheta+\pi k)$. By~(\ref{Rzero}), we have $R(\alpha,\vartheta+\pi k,-e^s)=0$. By~(\ref{periodR}), it follows that $R(\alpha,\vartheta,-e^s)=0$, i.e., $s\in N_{\alpha,\vartheta}$. We therefore have the inclusion
\begin{equation}\label{Nsubset}
\tilde N_{\alpha,\vartheta}\subset N_{\alpha,\vartheta},\quad \alpha<1,\,\vartheta\in\R.
\end{equation}
If $(\alpha,\vartheta)\in Q_0$, then $N_{\alpha,\vartheta}=\varnothing$ by Lemma~\ref{l_empty} and~(\ref{Sigma}) and, therefore, (\ref{Nsubset}) implies~(\ref{N=N}). By~(\ref{Q0Q1}), it remains to prove~(\ref{N=N}) for $(\alpha,\vartheta)\in Q_{\geq1}$. In this case, there is at least one $k\in\Z$ such that $(\alpha,\vartheta+\pi k)\in W$. We hence have
\begin{equation}\label{Nempty}
\tilde N_{\alpha,\vartheta} \neq \varnothing,\quad (\alpha,\vartheta)\in Q_{\geq1}.
\end{equation}
We now prove~(\ref{N=N}) for $(\alpha,\vartheta)\in Q_{\geq1}$ by separately considering the cases $\alpha>0$, $\alpha=0$, and $\alpha<0$.
\par\medskip\noindent 1. Let $\alpha>0$, $\kappa=\sqrt{\alpha}$, and $E_1, E_2\in \Sigma_{\alpha,\vartheta}$. By~(\ref{coscos}), we have $\cos(\vartheta\pm\vartheta_\kappa)\neq 0$ and it follows from~(\ref{eig_eq}) that $|E_1/E_2|^\kappa = 1$ and, hence, $E_1=E_2$. This means that $\Sigma_{\alpha,\vartheta}$ and, consequently, $N_{\alpha,\vartheta}$ contain at most one element. In view of~(\ref{Nsubset}) and~(\ref{Nempty}), this implies~(\ref{N=N}).
\par\medskip\noindent 2. Let $\alpha = 0$ and $E_1, E_2\in \Sigma_{0,\vartheta}$. By~(\ref{R0}), we have
\[
-\ln|E_{1,2}|\cos\vartheta + \pi\sin\vartheta = 0.
\]
By~(\ref{Q1'}), the condition $(0,\vartheta)\in Q_{\geq1}$ ensures that $\cos\vartheta\neq 0$. It follows that $\ln|E_1/E_2|=0$ and, hence, $E_1=E_2$. This means that $\Sigma_{0,\vartheta}$ and, consequently, $N_{0,\vartheta}$ contain at most one element. In view of~(\ref{Nsubset}) and~(\ref{Nempty}), this implies~(\ref{N=N}).
\par\medskip\noindent 3. Let $\alpha < 0$ and $s\in N_{\alpha,\vartheta}$. Then we have $(\alpha,\vartheta)\in W$ and, hence, $s'=S(\alpha,\vartheta)$ is an element of $\tilde N_{\alpha,\vartheta}$. Let $E = -e^s$ and $E' = -e^{s'}$. Since $s'\in N_{\alpha,\vartheta}$ by~(\ref{Nsubset}), we have $E,E'\in\Sigma_{\alpha,\vartheta}$. Let $\kappa=i\sqrt{|\alpha|}$. By~(\ref{coscos}), we have $\cos(\vartheta\pm\vartheta_\kappa)\neq 0$ and it follows from~(\ref{eig_eq}) that $|E/E'|^\kappa = 1$. This implies that $s = s'+2\pi k/\sqrt{|\alpha|}$ for some $k\in\Z$. By Lemma~\ref{l_S}, we conclude that $s = S(\alpha,\vartheta+\pi k)$ and, therefore, $s\in \tilde N_{\alpha,\vartheta}$. This means that $N_{\alpha,\vartheta}\subset \tilde N_{\alpha,\vartheta}$, whence (\ref{N=N}) follows by~(\ref{Nsubset}).
\end{proof}

By~(\ref{Q0}), (\ref{Q11}), (\ref{Qi}), (\ref{W}), and Theorem~\ref{t_N}, the set $N_{\alpha,\vartheta}$ is empty for $(\alpha,\vartheta)\in Q_0$, contains precisely one element for $(\alpha,\vartheta)\in Q_1$, and is countably infinite for $(\alpha,\vartheta)\in Q_\infty$. In view of~(\ref{Sigma}), the same is true for $\Sigma_{\alpha,\vartheta}$ (the emptiness of $\Sigma_{\alpha,\vartheta}$ for $(\alpha,\vartheta)\in Q_0$ also follows from Theorem~\ref{t_measure}). Corollary~\ref{cor_measure} therefore implies that $h_{\alpha,\vartheta}$ has no eigenvalues for $(\alpha,\vartheta)\in Q_0$, has one eigenvalue for $(\alpha,\vartheta)\in Q_1$, and has infinitely many eigenvalues for $(\alpha,\vartheta)\in Q_\infty$, in agreement with what was claimed in Sec.~\ref{intro}.

\begin{figure}
  \includegraphics[width=\linewidth]{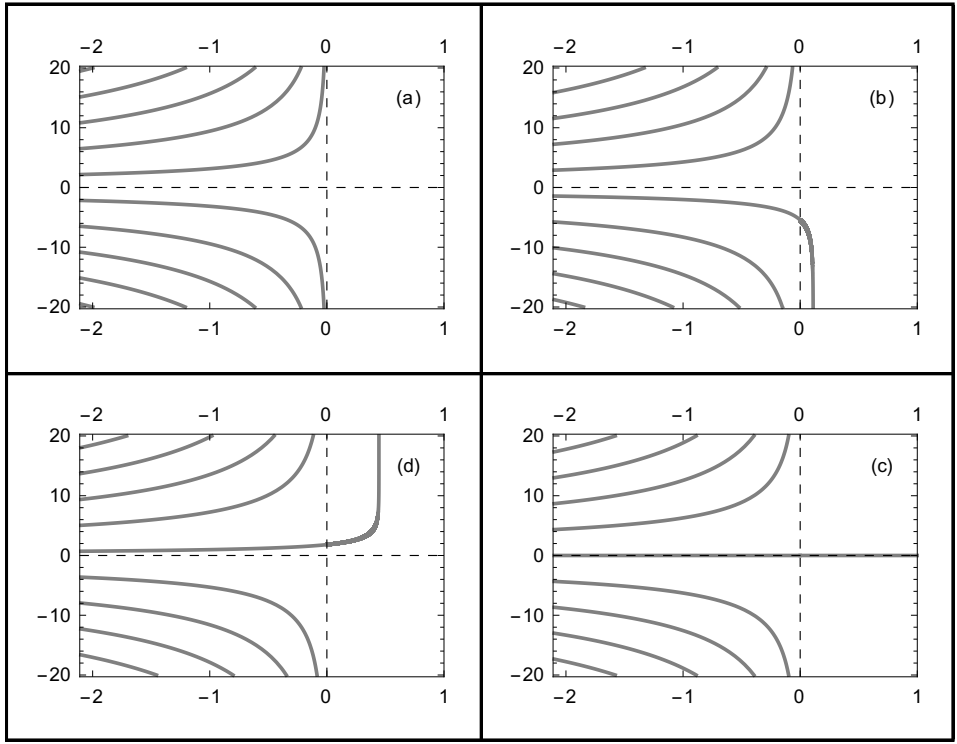}
  \caption{Plots (a), (b), (c), and (d) represent the set $N_\vartheta$ for $\vartheta = \pm\pi/2$, $-\pi/3$, $0$, and $\pi/6$ respectively. The horizontal and vertical axes correspond to the variables $\alpha$ and $s$.}
  \label{f1}
\end{figure}

We now obtain a graphical representation for the sets $N_{\alpha,\vartheta}$. Given $\vartheta\in\R$, we let $N_\vartheta$ denote the subset of $(\alpha,s)$-plane whose sections by the lines $\alpha = \mathrm{const}$ are precisely the sets $N_{\alpha,\vartheta}$,
\[
N_\vartheta = \{(\alpha,s)\in\R^2: \alpha < 1 \mbox{ and } s\in N_{\alpha,\vartheta}\}.
\]
Further, for every $\vartheta\in\R$, we let $S_\vartheta$ denote the function $\alpha\to S(\alpha,\vartheta)$ defined on the domain $D_{S_\vartheta} = \{\alpha\in\R: (\alpha,\vartheta)\in W\}$ and set
\[
G_\vartheta = \mbox{graph of }S_\vartheta = \{(\alpha,s)\in \R^2: (\alpha,\vartheta)\in W \mbox{ and }s = S(\alpha,\vartheta)\}.
\]
By Theorem~\ref{t_N}, we have
\begin{equation}\label{bigcupG}
N_\vartheta = \bigcup_{k\in\Z} G_{\vartheta+\pi k}.
\end{equation}
It follows from~(\ref{Q11}), (\ref{Qi}), and~(\ref{W}) that
\begin{equation}\nonumber
D_{S_\vartheta} =
\left\{
\begin{matrix}
(-\infty,\omega(\vartheta)),& |\vartheta|\leq \pi/2,\\
\R_-,& |\vartheta|> \pi/2,
\end{matrix}
\right.
\end{equation}
where $\omega$ is given by~(\ref{omega}).
If $|\vartheta|<\pi/2$, then it follows from~(\ref{S0}) and Lemma~\ref{l_S} that
\begin{equation}\label{Stheta<}
S_\vartheta(\alpha) =
\left\{
\begin{matrix}
\frac{1}{\sqrt{\alpha}} \ln \frac{\cos(\vartheta-\pi\sqrt{\alpha}/2)}{\cos(\vartheta+\pi\sqrt{\alpha}/2)}, & 0<\alpha<\omega(\vartheta),\\
\pi\tg\vartheta,& \alpha=0,\\
\frac{2}{\sqrt{|\alpha|}}\arctg\left(\tg\vartheta\thyp\frac{\pi\sqrt{|\alpha|}}{2}\right),& \alpha<0.
\end{matrix}
\right.
\end{equation}
Since $S$ is continuous on $W$, we can calculate $S_{\pm\pi/2}(\alpha)$ for $\alpha<0$ by passing to the limits $\vartheta\uparrow \pi/2$ and $\vartheta\downarrow -\pi/2$ in $S_\vartheta(\alpha)$. In view of~(\ref{S<}), we obtain
\begin{equation}
S_{\pm\pi/2}(\alpha) = \pm\frac{\pi}{\sqrt{|\alpha|}},\quad \alpha<0.\nonumber
\end{equation}
By Lemma~\ref{l_S}, we have
\begin{align}
& D_{S_{\vartheta+\pi k}} = \R_-,\nonumber\\
& S_{\vartheta+\pi k}(\alpha) = S_\vartheta(\alpha)+\frac{2\pi k}{\sqrt{|\alpha|}},\quad \alpha<0,\label{Stheta>}
\end{align}
for every $\vartheta\in [-\pi/2,\pi/2]$ and every nonzero $k\in\Z$.
Formulas~(\ref{bigcupG}), (\ref{Stheta<}), and~(\ref{Stheta>}) allow us to draw the set $N_\vartheta$,  for every $\vartheta\in\R$. In Fig.~\ref{f1}, this set, which represents the $\alpha$-dependence of eigenvalues of $h_{\alpha,\vartheta}$ in the logarithmic scale, is shown for $\vartheta = \pm\pi/2$, $-\pi/3$, $0$, and $\pi/6$. For $\vartheta \neq \pi/2+\pi k$, where $k\in\Z$, there is precisely one eigenvalue that crosses the line $\alpha=0$ in an analytic way, while all other eigenvalues die away at zero or minus infinity as $\alpha\uparrow 0$. If $\vartheta = \pi/2+\pi k$ for some $k\in\Z$, then there are no eigenvalues for $\alpha\geq 0$.

\subsection{Continuous part of \texorpdfstring{$\mathcal V_{\alpha,\vartheta}$}{V}.}

We now consider the absolutely continuous part of the spectral measure $\mathcal V_{\alpha,\vartheta}$. By Theorem~\ref{t_measure}, its density $t_{\alpha,\vartheta}$ is given by~(\ref{density}). Let the function $T$ on $\C\times\C\times \C_{3\pi/2}$ be defined by the formula
\begin{equation}\label{T}
T(\alpha,\vartheta,z) = 2 R(\alpha,\vartheta,z)\overline{R(\bar\alpha,\bar\vartheta,z)},\quad \alpha,\vartheta\in\C,\,z\in\C_{3\pi/2}.
\end{equation}
Clearly, $(\alpha,\vartheta)\to T(\alpha,\vartheta,z)$ is a holomorphic function on $\C\times\C$ for every $z\in \C_{3\pi/2}$.
If $\alpha$ and $\vartheta$ are real, then we have
\begin{equation}\label{TR}
T(\alpha,\vartheta,z) = 2 |R(\alpha,\vartheta,z)|^2,\quad z\in\C_{3\pi/2},
\end{equation}
and it follows from~(\ref{density}) that
\begin{equation}\label{Tt}
t_{\alpha,\vartheta}(E) = T(\alpha,\vartheta,E)^{-1},\quad E>0,
\end{equation}
for every $\alpha<1$ and $\vartheta\in\R$. We shall explicitly express $T$ in terms of the functions $\Sinc$ and $\Cos$. In view of~(\ref{Tt}), this will also give us a formula for the density $t_{\alpha,\vartheta}$.

\begin{figure}
  \includegraphics[width=\linewidth]{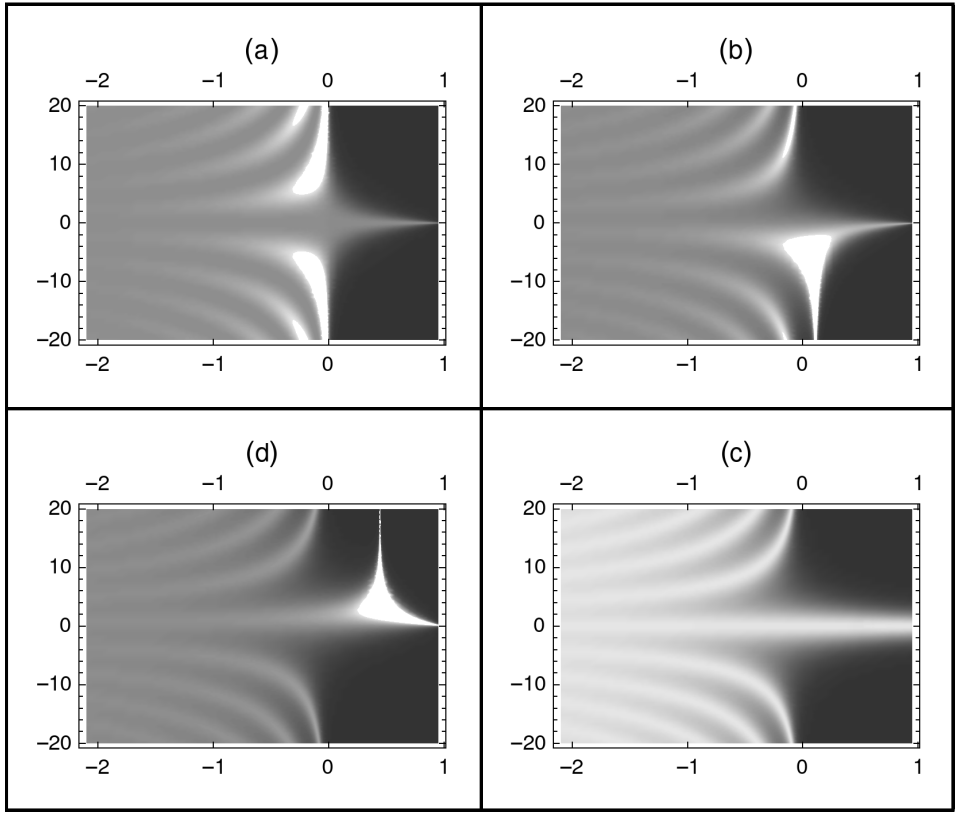}
  \caption{Plots (a), (b), (c), and (d) represent the function $\mathfrak  t_\vartheta(\alpha,s)$ for $\vartheta = \pm\pi/2$, $-\pi/3$, $0$, and $\pi/6$ respectively. The horizontal and vertical axes correspond to the variables $\alpha$ and $s$. The value of $\mathfrak  t_\vartheta$ is encoded in the brightness of the plot: brighter regions correspond to greater values of the function.}
  \label{f2}
\end{figure}

By~(\ref{R}) and~(\ref{T}), we have
\begin{multline}\label{T'}
T(\kappa^2,\vartheta,Ee^{i\phi}) =  \frac{2}{\kappa^2}\left(E^{-\kappa}\cos^2(\vartheta-\vartheta_\kappa) -\right.\\\left. - 2 \cos((\pi-\phi)\kappa) \cos(\vartheta-\vartheta_\kappa)\cos(\vartheta+\vartheta_\kappa) + E^{\kappa}\cos^2(\vartheta+\vartheta_\kappa)\right)
\end{multline}
for all $\kappa\in\C\setminus\{0\}$, $\vartheta\in\C$, $E>0$, and $-\pi/2<\phi<3\pi/2$.
In view of the equality
\[
\ln E \Sinc\left(-\kappa^2\ln^2 E\right) = \frac{1}{\kappa}\sh(\kappa\ln E) = \frac{E^{\kappa}-E^{-\kappa}}{2\kappa},\quad \kappa\in\C\setminus\{0\},\,E>0,
\]
which follows from~(\ref{Cos2}), and the trigonometric identities
\begin{align}
&2\cos^2(\vartheta\pm\vartheta_\kappa) = 1 + \cos2\vartheta \cos\pi\kappa \mp \sin2\vartheta\sin\pi\kappa,\nonumber\\
&2\cos(\vartheta+\vartheta_\kappa)\cos(\vartheta-\vartheta_\kappa) = \cos2\vartheta + \cos\pi\kappa,\nonumber
\end{align}
which hold for all $\vartheta,\kappa\in\C$, we derive from~(\ref{T'}) that
\begin{multline}\label{T''}
T(\kappa^2,\vartheta,Ee^{i\phi}) = \ln^2E \Sinc^2\left(-\frac{\kappa^2}{4}\ln^2 E\right)(1+\cos2\vartheta\cos\pi\kappa) -\\
- 2\pi \ln E\Sinc\left(-\kappa^2\ln^2 E\right)\sinc\pi\kappa \sin 2\vartheta
+ 2\frac{\cos\pi\kappa-\cos(\pi-\phi)\kappa}{\kappa^2}\cos 2\vartheta +\\+  2\frac{1-\cos\pi\kappa\cos(\pi-\phi)\kappa}{\kappa^2}
\end{multline}
for every $\kappa\in \C\setminus\{0\}$, $\vartheta\in\C$, $E>0$, and $-\pi/2<\phi<3\pi/2$.
Let the functions $\tau$ and $\mu$ on $\C\times\C$ be defined by the formulas
\begin{align}\label{tau}
&\tau(\alpha,\phi) = (\pi-\phi)^2\Sinc^2\left(\frac{(\pi-\phi)^2\alpha}{4}\right) - \pi^2\Sinc^2\left(\frac{\pi^2\alpha}{4}\right),\\
&\mu(\alpha,\phi) = 2\pi^2\Sinc^2(\pi^2\alpha) + \Cos(\pi^2\alpha)\tau(\alpha,\phi) \label{mu}
\end{align}
for every $\alpha,\phi\in\C$. Performing elementary trigonometric transformations, we obtain
\begin{align}
&\tau(\kappa^2,\phi) = \frac{4}{\kappa^2}\left(\sin^2\frac{(\pi-\phi)\kappa}{2} - \sin^2 \frac{\pi\kappa}{2}\right) = 2\frac{\cos\pi\kappa - \cos(\pi-\phi)\kappa}{\kappa^2},\label{tau'}\\
&\mu(\kappa^2,\phi) = 2\frac{\sin^2\pi\kappa}{\kappa^2} + \tau(\kappa^2,\phi)\cos\pi\kappa = 2 \frac{1-\cos\pi\kappa\,\cos(\pi-\phi)\kappa}{\kappa^2}\label{mu'}
\end{align}
for all $\kappa\in\C\setminus\{0\}$ and $\phi\in\C$. In view of these formulas, (\ref{T''}) implies that
\begin{multline}\label{T'''}
T(\alpha,\vartheta,Ee^{i\phi}) = \ln^2E \Sinc^2\left(-\frac{\alpha}{4}\ln^2 E\right)(1+\cos2\vartheta\Cos(\pi^2\alpha)) -\\-
2\pi \ln E\Sinc\left(-\alpha\ln^2 E\right)\Sinc(\pi^2\alpha) \sin 2\vartheta + \tau(\alpha,\phi)\cos 2\vartheta +
\mu(\alpha,\phi)
\end{multline}
for all $\alpha\in \C\setminus\{0\}$, $\vartheta\in\C$, $E>0$, and $-\pi/2<\phi<3\pi/2$. By continuity, this equality remains valid for $\alpha=0$. Thus, (\ref{T'''}) holds for all $\alpha,\vartheta\in \C$, $E>0$, and $-\pi/2<\phi<3\pi/2$.

Formulas~(\ref{Tt}) and~(\ref{T'''}) can be illustrated by drawing the graphs of the density $t_{\alpha,\vartheta}(E)$ as a function of $\alpha$ and $E$ for various values of $\vartheta$. For this, it is convenient to use the logarithmic scale for the energy variable and multiply the density $t_{\alpha,\vartheta}$ by the factor $\mu(\alpha,0) = 2\pi^2\Sinc^2(\pi^2\alpha)$. More precisely, given $\vartheta\in\R$, we define the function $\mathfrak  t_\vartheta$ on $(-\infty,1)\times \R$ by setting
\begin{equation}\label{ttt}
\mathfrak  t_\vartheta(\alpha,s) = 2\pi^2\Sinc^2(\pi^2\alpha)t_{\alpha,\vartheta}(e^s),\quad s\in\R,\,\alpha<1.
\end{equation}
By~(\ref{Tt}), (\ref{mu}), and~(\ref{T'''}), we have
\begin{multline}\nonumber
\mathfrak  t_\vartheta(\alpha,s)^{-1} = 1 + \frac{s^2\Sinc^2\left(-\alpha s^2/4\right)}{2\pi^2\Sinc^2(\pi^2\alpha)} (1+\cos2\vartheta\Cos(\pi^2\alpha)) - \frac{s\Sinc\left(-\alpha s^2\right)}{\pi\Sinc(\pi^2\alpha)} \sin 2\vartheta
\end{multline}
for every $\alpha<1$ and $\vartheta,s\in\R$. In Fig.~\ref{f2}, the function $\mathfrak  t_\vartheta$ is plotted using this formula for $\vartheta = \pm\pi/2$, $-\pi/3$, $0$, and $\pi/6$.

\begin{figure}
  \includegraphics[scale=.75]{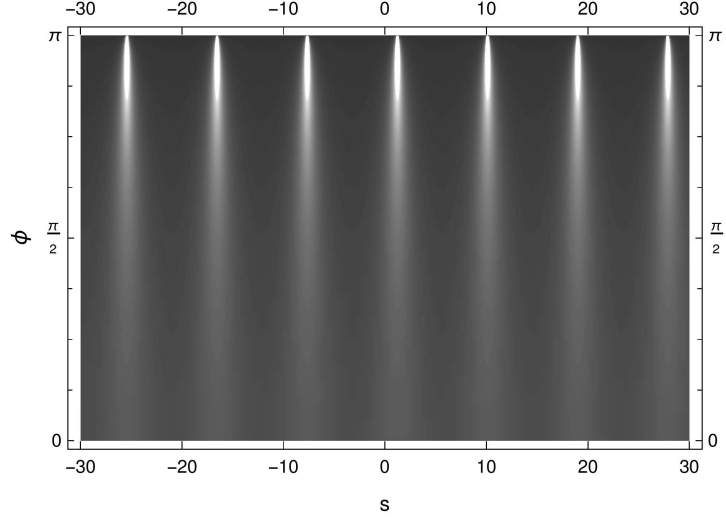}
  \caption{The function $J_{\alpha,\vartheta}(s,\phi)$ is plotted for $\alpha=-1/2$ and $\vartheta = \pi/6$. As in Fig.~\ref{f2}, brighter regions correspond to greater values of the function.}
  \label{f3}
\end{figure}

The comparison between Figs.~\ref{f1} and~\ref{f2} shows that the $\alpha$-dependence of eigenvalues and the density of the continuous part of $\mathcal V_{\alpha,\vartheta}$ follows the same pattern. This phenomenon can be easily understood if we recall that the point and continuous parts of $\mathcal V_{\alpha,\vartheta}$ are both determined by the imaginary part of $\mathscr M_{\alpha,\vartheta}$ via~(\ref{nualphavartheta}) and consider its behaviour in the upper complex half-plane.
As in the case of the density $t_{\alpha,\vartheta}$, to facilitate the visualization of $\Im\mathscr M_{\alpha,\vartheta}(z)$, we multiply it by $2\pi^2\Sinc^2(\pi^2\alpha)$ and pass to the logarithmic scale for $|z|$. Therefore, for every $\alpha<1$ and $\vartheta\in\R$, we introduce the function $J_{\alpha,\vartheta}$ on $\R\times [0,\pi)$ that is defined by the formula
\begin{multline}\nonumber
J_{\alpha,\vartheta}(s,\phi) = 2\pi^2\Sinc^2(\pi^2\alpha)\,\Im\mathscr M_{\alpha,\vartheta}(e^{s+i\phi}) =\\=
2\pi(\pi-\phi)\Sinc((\pi-\phi)^2\alpha)\Sinc(\pi^2\alpha)T(\alpha,\vartheta,e^{s+i\phi})^{-1},\quad s\in\R,\, 0\leq\phi<\pi,
\end{multline}
where the second equality follows from~(\ref{subset}), (\ref{ImM}), and~(\ref{TR}).
In view of~(\ref{Tt}) and~(\ref{ttt}), we have $J_{\alpha,\vartheta}(s,0) = \mathfrak t_\vartheta(\alpha,s)$ for every $\alpha<1$ and $\vartheta,s\in\R$.
Using~(\ref{T'''}), we can explicitly express $J_{\alpha,\vartheta}$ in terms of $\Sinc$ and $\Cos$. In Fig.~\ref{f3}, the function $J_{\alpha,\vartheta}$ is plotted for $\alpha=-1/2$ and $\vartheta = \pi/6$. We see that the graph of $J_{\alpha,\vartheta}$ contains equidistant vertical ridges that connect the points of $N_{\alpha,\vartheta}$ at their upper ends with the maxima of the function $s\to\mathfrak t_\vartheta(\alpha,s)$ at their lower ends. This means that the graph of $\Im\mathscr M_{\alpha,\vartheta}$ contains ridges along logarithmically equidistant semicircles in the upper half-plane that connect eigenvalues of $h_{\alpha,\vartheta}$ on the negative half-axis with the maxima of the density $t_{\alpha,\vartheta}$ on the positive half-axis. The values $\alpha=-1/2$ and $\vartheta = \pi/6$ chosen for Fig.~\ref{f3} play no special role: the functions $J_{\alpha,\vartheta}$ and $\Im\mathscr M_{\alpha,\vartheta}$ behave in the same way for every $\alpha<0$ and $\vartheta\in\R$.

\section{Self-adjoint extensions}
\label{s_one}

In this section, we recall basic facts concerning self-adjoint extensions of one-dimensional Schr\"odinger operators and then apply the general theory to proving Theorem~\ref{t_sa}. We refer the reader to~\cite{Naimark,Teschl,Weidmann} for a detailed treatment of one-dimensional Schr\"odinger operators.

\subsection{General theory}
\label{s_one1}
As in Sec.~\ref{intro}, let $\lambda_+$ be the restriction to $\R_+$ of the Lebesgue measure $\lambda$ on $\R$ and $\mathcal D$ be the space of all complex continuously differentiable functions on $\R_+$ whose derivative is absolutely continuous on $\R_+$.
Let $q$ be a complex locally integrable function on $\R_+$. Given $z\in\C$, we let $l_{q,z}$ denote the linear operator from $\mathcal D$ to the space of complex $\lambda_+$-equivalence classes such that
\begin{equation}\label{lqz}
(l_{q,z} f)(r) = -f''(r)+ q(r) f(r) - z f(r)
\end{equation}
for $\lambda$-a.e. $r\in \R_+$ and set
\begin{equation}\label{lq0}
l_q = l_{q,0}.
\end{equation}
For every $f\in\mathcal D$ and $z\in\C$, we have $l_{q,z}f = l_q f - z[f]$, where, as in Sec.~\ref{intro}, $[f]=[f]_{\lambda_+}$ denotes the $\lambda_+$-equivalence class of $f$.
For every $a>0$ and all complex numbers $z$, $\zeta_1$, and $\zeta_2$, there is a unique solution $f$ of the equation $l_{q,z}f=0$ such that $f(a)=\zeta_1$ and $f'(a)=\zeta_2$. This implies that solutions of $l_{q,z}f=0$ constitute a two-dimensional subspace of $\mathcal D$. If $f,g\in \mathcal D$ are such that $W_r(f,g)$ has a limit as $r\downarrow 0$, then we set
\begin{equation}\label{Wdown}
W^\downarrow(f,g) = \lim_{r\downarrow 0} W_r(f,g).
\end{equation}
Similarly, if $f,g\in \mathcal D$ are such that $W_r(f,g)$ has a limit as $r\uparrow \infty$, then we set
\begin{equation}\label{Wup}
W^\uparrow(f,g) = \lim_{r\uparrow \infty} W_r(f,g).
\end{equation}

In the rest of this subsection, we assume that $q$ is real. Let
\begin{equation}\label{D_q}
\mathcal D_q = \{f\in \mathcal D : f \mbox{ and } l_q f \mbox{ are both square-integrable on } \R_+\}.
\end{equation}
A $\lambda_+$-measurable complex function $f$ is said to be left or right square-integrable on $\R_+$ if respectively $\int_0^a |f(r)|^2\,dr <\infty$ or $\int_a^\infty
|f(r)|^2\,dx <\infty$ for any $a>0$. The subspace of $\mathcal D$ consisting of left or right square-integrable on $\R_+$ functions $f$ such that $l_qf$ is also respectively left or right square-integrable on $\R_+$ is denoted by $\mathcal D_q^\downarrow$ or $\mathcal D_q^\uparrow$. We obviously have $\mathcal D_q = \mathcal D_q^\downarrow\cap \mathcal D_q^\uparrow$.
It follows from~(\ref{lqz}) by integrating by parts that
\[
\int_a^b ((l_{q,z} f)(r)g(r) - f(r) (l_{q,z} g)(r))\,dr = W_b(f,g) - W_a(f,g)
\]
for every $f,g\in \mathcal D$, $z\in\C$, and $a,b>0$. This implies the existence of limits in the right-hand sides of~(\ref{Wdown}) and~(\ref{Wup}) respectively for every $f,g\in \mathcal D_q^\downarrow$ and $f,g\in \mathcal D_q^\uparrow$. Hence, $W^\downarrow(f,g)$ is well-defined for every $f,g\in \mathcal D_q^\downarrow$ and $W^\uparrow(f,g)$ is well-defined for every $f,g\in \mathcal D_q^\uparrow$.
Moreover, it follows that
\begin{equation}\label{anti}
\langle l_q f, [g]\rangle - \langle [f], l_q g\rangle = W^\uparrow(\bar f,g) - W^\downarrow(\bar f,g)
\end{equation}
for any $f,g\in \mathcal D_q$, where $\langle\cdot,\cdot\rangle$ is the scalar product in $L_2(\R_+)$.

For any linear subspace $Z$ of $\mathcal D_q$, let $L_q(Z)$ be the linear operator in $L_2(\R_+)$ defined by the relations
\begin{equation}\label{L_q(Z)}
\begin{split}
& D_{L_q(Z)} = \{[f]: f\in Z\},\\
& L_q(Z) [f] = l_q f,\quad f\in Z.
\end{split}
\end{equation}
We define the minimal operator $L_q$ by setting
\begin{equation}\label{L_q}
L_q = L_q(\mathcal D_q^0),
\end{equation}
where
\begin{equation}\label{D^0_q}
\mathcal D_q^0 = \{f\in \mathcal D_q: W^\downarrow(f,g)=W^\uparrow(f,g)=0\mbox{ for every }g\in \mathcal D_q\}.
\end{equation}
By~(\ref{anti}), the operator $L_q(Z)$ is symmetric if and only if $W^\downarrow(\bar f,g) = W^\uparrow(\bar f,g)$ for any $f,g\in Z$. In particular, $L_q$ is a symmetric operator.
Moreover, $L_q$ is closed and densely defined, and its adjoint $L_q^*$ is given by
\begin{equation}\label{lq*}
L_q^* = L_q(\mathcal D_q)
\end{equation}
(see~Lemma~9.4 in~\cite{Teschl}).

If $W^\downarrow(f,g)=0$ for any $f,g\in \mathcal D_q^\downarrow$, then $q$ is said to be in the limit point case (l.p.c.) on the left. Otherwise $q$ is said to be in the limit circle case (l.c.c.) on the left. Similarly, $q$ is said to be in the l.p.c. on the right if $W^\uparrow(f,g)=0$ for any $f,g\in \mathcal D_q^\uparrow$ and to be in the l.c.c. on the right otherwise. According to the well-known Weyl alternative (see, e.g., \cite{Teschl}, Theorem~9.9), $q$ is in the l.c.c. on the left if and only if all solutions of $l_{q}f = 0$ are left square-integrable on $\R_+$ (and, hence, belong to $\mathcal D_q^\downarrow$).

If $q$ is in the l.p.c. both on the left and on the right, then (\ref{lq*}) implies that $L_q^*$ is symmetric and, therefore, $L_q$ is self-adjoint.

If $q$ is in the l.c.c. on the left and in the l.p.c. on the right, then
$L_q$ has deficiency indices $(1,1)$ and the self-adjoint extensions of $L_q$ are precisely the operators (see~\cite{Weidmann}, Theorem~5.8)
\begin{equation}\label{Lqf}
L_q^f = L_q(Z_q^f),
\end{equation}
where $f$ is a nontrivial real solution of $l_q f = 0$ and the subspace $Z_q^f$ of $\mathcal D_q$ is given by
\begin{equation}\label{Zqf}
Z_q^f = \{g\in \mathcal D_q: W^\downarrow(f,g) = 0\}.
\end{equation}
The operator $L_q^f$ determines $f$ uniquely up to a nonzero real coefficient.

If $q$ is locally square-integrable on $\R_+$, then formulas~(\ref{L_q}) and~(\ref{D^0_q}) imply that $C_0^\infty(\R_+)$ is contained in $\mathcal D_q^0$ and $L_q$ is an extension of $L_q(C_0^\infty(\R_+))$.

\begin{lemma}\label{ll}
Let $q$ be a real locally square-integrable function on $\R_+$. Then $L_q$ is the closure of $L_q(C_0^\infty(\R_+))$.
\end{lemma}
\begin{proof}
See Lemma~17 in~\cite{Smirnov2016}.
\end{proof}

\subsection{The case of the inverse-square potential}
\label{s_one2}
By~(\ref{lalphaz}) and~(\ref{lqz}), we have
\begin{equation}
\mathscr L_{\alpha,z} = l_{q_\alpha,z},\quad \alpha,z\in\C,\label{eq1}
\end{equation}
where the function $q_\alpha$ on $\R_+$ is given by~(\ref{qalpha}). In view of~(\ref{lalpha0}) and~(\ref{lq0}), this implies that
\begin{equation}\label{eq2}
\mathscr L_\alpha = l_{q_\alpha},\quad \alpha\in\C.
\end{equation}
By~(\ref{lalpha}) and~(\ref{eq1}), we obtain
\begin{equation}\label{lqalpha}
l_{q_\alpha,z} \mathcal U^\alpha_\vartheta(z) = 0,\quad \alpha,\vartheta,z\in\C.
\end{equation}

If $\alpha$ is real, then $q_\alpha$ is real. It follows from~(\ref{Delta}), (\ref{H_alpha(Z)}), (\ref{D_q}), (\ref{L_q(Z)}), and~(\ref{eq2}) that
\begin{align}
&\Delta_\alpha = \mathcal D_{q_\alpha},\label{eq3}\\
&H_\alpha(Z) = L_{q_\alpha}(Z)\label{eq4}
\end{align}
for every $\alpha\in\R$ and every linear subspace $Z$ of $\Delta_\alpha$. Hence, (\ref{checkh}), (\ref{hkappadef}), and Lemma~\ref{ll} imply that
\begin{equation}\label{halpha}
h_\alpha = L_{q_\alpha},\quad \alpha\in\R.
\end{equation}
By~(\ref{lq*}), (\ref{eq3}), (\ref{eq4}), and~(\ref{halpha}), equality~(\ref{halpha*}) holds for all real $\alpha$.

If $\alpha=\kappa^2$ for $\kappa\in\C$, then the equation $l_{q_\alpha} f = 0$ has linearly independent solutions $r^{1/2\pm\kappa}$ for $\kappa\neq 0$ and $r^{1/2}$ and $r^{1/2}\ln r$ for $\kappa=0$. It follows that
\begin{itemize}
\item[(i)] $q_\alpha$ is in the l.p.c. both on the left and on the right for $\alpha\geq 1$ and
\item[(ii)]  $q_\alpha$ is in the l.p.c. on the right and in the l.c.c. on the left for $\alpha<1$.
\end{itemize}

In view of~(\ref{lq0}) and~(\ref{lqalpha}), we have $l_{q_\alpha}\mathcal U^\alpha_\vartheta(0) = 0$ for every $\alpha,\vartheta\in\C$. For $\alpha<1$ and $\vartheta\in\R$, the function $\mathcal U^\alpha_\vartheta(0)$ is real and nontrivial, and it follows from~(\ref{zalphatheta}), (\ref{Zqf}), and~(\ref{eq3}) that $Z_{\alpha,\vartheta} = Z_{q_\alpha}^{\mathcal U^\alpha_\vartheta(0)}$. By~(\ref{halphatheta}), (\ref{Lqf}), and~(\ref{eq4}), we conclude that
\begin{equation}\label{halphatheta1}
h_{\alpha,\vartheta} = L^{\mathcal U^\alpha_\vartheta(0)}_{q_\alpha},\quad \alpha<1,\,\vartheta\in\R.
\end{equation}

\begin{proof}[Proof of Theorem~$\ref{t_sa}$]
In view of~(\ref{halpha}) and condition~(i), the operator $h_\alpha$ is self-adjoint for $\alpha\geq 1$. Let $\alpha<1$.
By~(\ref{halpha}) and~(\ref{halphatheta1}), $h_{\alpha,\vartheta}$ a self-adjoint extension of $h_\alpha$ for every $\vartheta\in\R$. Conversely, let $H$ be a self-adjoint extension of $h_\alpha$. Then $H=L^f_{q_\alpha}$ for some nontrivial real $f\in \mathcal D$ satisfying $l_{q_\alpha}f=0$. By~(\ref{LAB}), (\ref{lq0}), and~(\ref{eq1}), we have $l_{q_\alpha}\mathcal A^\alpha(0) = l_{q_\alpha}\mathcal B^\alpha(0) = 0$. Since $\mathcal A^\alpha(0)$ and $\mathcal B^\alpha(0)$ are real and linearly independent, it follows from~(\ref{Ualphatheta}) that $f = c\,\mathcal U^\alpha_\vartheta(0)$ for some $c,\vartheta\in\R$ such that $c\neq 0$. In view of~(\ref{halphatheta1}), this means that $H = h_{\alpha,\vartheta}$.

Suppose now that $\vartheta,\vartheta'\in\R$ and $h_{\alpha,\vartheta} = h_{\alpha,\vartheta'}$. By~(\ref{halphatheta1}), we have $\mathcal U^\alpha_\vartheta(0) = c\,\mathcal U^\alpha_{\vartheta'}(0)$ for some nonzero real $c$. In view of~(\ref{Ualphatheta}), this implies that $\sin\vartheta = c\sin\vartheta'$ and $\cos\vartheta = c\cos\vartheta'$ and, therefore, $e^{i\vartheta} = ce^{i\vartheta'}$. It follows that $e^{i(\vartheta-\vartheta')}=c$, whence $c = \pm 1$ and $\vartheta-\vartheta'\in\pi\Z$.
\end{proof}

\section{Eigenfunction expansions}
\label{s_eig}

This section consists of two subsections. In the first one, we briefly describe the construction of eigenfunction expansions of one-dimensional Schr\"odinger operators developed in~\cite{GesztesyZinchenko,KST}. This construction, which is adapted to the case of operators with a simple spectrum and relies on so-called singular Titchmarsh--Weyl $m$-functions, can be viewed as a variant of Kodaira's general approach~\cite{Kodaira} based on matrix-valued measures (see Remark~16 in~\cite{Smirnov2016}). In the second subsection, we prove Theorem~\ref{leig2} by applying the general theory to the case of the inverse-square potential.

\subsection{General theory}
Let $q$ be a real locally integrable function on $\R_+$. We assume that $q$ is in the l.c.c. on the left and in the l.p.c. on the right.

Let $O\subset \C$ be an open set. We say that a map $u\colon O\to \mathcal D$ is a $q$-solution on $O$ if $l_{q,z}u(z) = 0$ for every $z\in O$. A $q$-solution $u$ on $O$ is said to be holomorphic if the functions $z\to u(z|r)$ and $z\to \partial_r u(z|r)$ are holomorphic on $O$ for any $r\in \R_+$. A $q$-solution $u$ on $O$ is said to be nonvanishing if $u(z)\neq 0$ for every $z\in O$. A $q$-solution in $\C$ is said to be real-entire if it is holomorphic on $\C$ and $u(E)$ is real for every $E\in\R$.

Let $u$ be a real-entire $q$-solution. Since $q$ is in l.c.c. on the left, we have $u(z)\in \mathcal D_q^\downarrow$ for every $z\in\C$. Suppose that $u$ is nonvanishing and
\begin{equation}\label{bbb}
W^\downarrow(u(z),u(z')) = 0,\quad z,z'\in \C.
\end{equation}
Let $v$ be a nonvanishing holomorphic $q$-solution on $\C_+$ such that $v(z)$ is right square-integrable for every $z\in\C_+$ (such a $v$ always exists; see Lemma~9.8 in~\cite{Teschl}). If $W(u(z),v(z)) = 0$ for some $z\in \C_+$,\footnote{We recall that $W(f,g)$ denotes the value of the function $r\to W_r(f,g)$ if $f,g\in \mathcal D$ are such that this function is constant (in particular, if $l_{q,z}f=l_{q,z}g=0$ for some $z\in\C$).} then $u(z)$ is proportional to $v(z)$ and, hence, $u(z)\in \mathcal D_q$. In view of~(\ref{Zqf}) and~(\ref{bbb}), this means that $u(z)\in Z_q^{u(0)}$ and, therefore, $[u(z)]$ is an eigenvector of the self-adjoint operator $L_q^{u(0)}$ with the eigenvalue $z$. But this cannot be the case because all eigenvalues of $L_q^{u(0)}$ must be real. It follows that $W(u(z),v(z)) \neq 0$ for all $z\in\C_+$.

Given a nonvanishing real-entire $q$-solution $u$, one can always find another real-entire $q$-solution $\tilde u$ such that $W(u(z),\tilde u(z))\neq 0$ for all $z\in \C$ (see Lemma~2.4 in~\cite{KST}).

Let $u$ and $\tilde u$ be real-entire $q$-solutions such that (\ref{bbb}) holds and $W(u(z),\tilde u(z))\neq 0$ for all $z\in\C$. We define the holomorphic function $\mathcal M^q_{u,\tilde u}$ on $\C_+$ by setting
\begin{equation}\label{mathcalM}
 \mathcal M^q_{u,\tilde u}(z) = \frac{1}{\pi}\frac{W(v(z),\tilde u(z))}{W(v(z),u(z))W(u(z),\tilde u(z))},
\end{equation}
where $v$ is a nonvanishing holomorphic $q$-solution on $\C_+$ such that $v(z)$ is right square-integrable for all $z\in\C_+$ (since $q$ is in the l.p.c. on the right, this definition is independent of the choice of $v$). Following~\cite{KST}, we call such functions singular Titchmarsh-Weyl $m$-functions.

The proof of the next statement can be found in~\cite{KST}.
\begin{proposition}\label{t_eig}
Let a locally integrable real function $q$ on $\R_+$ be in the l.c.c. on the left and in the l.p.c. on the right. Let $u$ be a nonvanishing real-entire $q$ solution such that $(\ref{bbb})$ holds for all $z\in\C$. Then the following statements hold:
\begin{itemize}
\item[$1.$] There exists a unique positive Radon measure $\nu$ on $\R$ (called the spectral measure for $q$ and $u$) such that
\[
\int \varphi(E)\, \Im\mathcal M^q_{u,\tilde u}(E+i\eta)\,dE\to \int \varphi(E)\,d\nu(E)\quad (\eta\downarrow 0)
\]
for every continuous function $\varphi$ on $\R$ with compact support and every real-entire $q$-solution $\tilde u$ such that $W(u(z),\tilde u(z))\neq 0$ for every $z\in\C$.
\item[$2.$] Let $\nu$ be the spectral measure for $q$ and $u$. There is a unique unitary operator $U\colon L_2(\R_+)\to L_2(\R,\nu)$ (called the spectral transformation for $q$ and $u$) such that
\begin{equation}\nonumber
(U\psi)(E) = \int_{\R_+}u(E|r)\psi(r)\,dr,\quad \psi\in L_2^c(\R_+),
\end{equation}
for $\nu$-a.e. $E$.
\item[$3.$] Let $\nu$ and $U$ be the spectral measure and transformation for $q$ and $u$. Then we have
$L_q^{u(0)} = U^{-1}\mathcal T^\nu_\iota U$,
where $\iota$ is the identity function on $\R$.
\end{itemize}
\end{proposition}

In the next subsection, we shall verify that $\mathcal V_{\alpha,\vartheta}$ is actually the spectral measure for $q_\alpha$ and $\mathcal U^\alpha_\vartheta$. This justifies using the same term ``spectral measure'' for $\mathcal V_{\alpha,\vartheta}$ and for the measures described by Proposition~\ref{t_eig}.

\subsection{The case of the inverse-square potential}
For any $\kappa\in\C$, we define the map $\mathfrak v^\kappa\colon \C_{3\pi/2}\to \mathcal D$ by the relation
\begin{equation}\label{vkappa}
\mathfrak v^\kappa(z|r) = \frac{i\pi}{2} e^{i\pi\kappa/2} r^{1/2}H^{(1)}_{\kappa}(r z^{1/2}),\quad r\in\R_+,\, z\in\C_{3\pi/2},
\end{equation}
where $H^{(1)}_\kappa$ is the first Hankel function of order $\kappa$. Because $H^{(1)}_\kappa$ is a solution of the Bessel equation, we have
\begin{equation}\label{eqv}
\mathscr L_{\kappa^2,z}\mathfrak v^\kappa(z) = 0,\quad \kappa\in\C,\,z\in\C_{3\pi/2}.
\end{equation}
It follows from the relation $H^{(1)}_{-\kappa} = e^{i\pi\kappa}H^{(1)}_{\kappa}$ (formula~(9) in~Sec.~7.2.1 in~\cite{Bateman}) that
\begin{equation}\label{-kappa}
\mathfrak v^{-\kappa}(z) = \mathfrak v^\kappa(z),\quad \kappa\in\C,\, z\in\C_{3\pi/2}.
\end{equation}
In view of~(\ref{Lu}) and~(\ref{eqv}), the Wronskian $W_r(\mathfrak v^{\kappa}(z),\mathfrak u^{\pm\kappa}(z))$ does not depend on $r$. To find it explicitly, we can use the expression for the Wronskian of Bessel functions (formula~(29) in~Sec.~7.11 in~\cite{Bateman}),
\begin{equation}\label{wrbess}
W_z(J_\kappa,H^{(1)}_\kappa) = \frac{2i}{\pi z}.
\end{equation}
Taking~(\ref{-kappa}) into account and combining~(\ref{wrbess}) with~(\ref{ukappa}), (\ref{bessel}), and~(\ref{vkappa}), we derive that
\begin{equation}\label{wrsol1}
W(\mathfrak v^\kappa(z),\mathfrak u^{\kappa}(z)) = z^{-\kappa/2}e^{i\pi\kappa/2},\quad W(\mathfrak v^\kappa(z),\mathfrak u^{-\kappa}(z)) = z^{\kappa/2}e^{-i\pi\kappa/2}
\end{equation}
for any $\kappa\in \C$ and $z\in\C_{3\pi/2}$.

For $\alpha\in\C$ and $z\in\C_{3\pi/2}$, let the function $\mathscr V^\alpha(z)$ on $\R_+$ be defined by the relation $\mathscr V^\alpha(z) = \mathfrak v^\kappa(z)$, where $\kappa\in \C$ is such that $\kappa^2 = \alpha$ (by~(\ref{-kappa}), this definition does not depend on the choice of $\kappa$). We therefore have
\begin{equation}\label{Vv}
\mathscr V^{\kappa^2}(z) = \mathfrak v^\kappa(z)
\end{equation}
for every $\kappa\in \C$ and $z\in\C_{3\pi/2}$. By~(\ref{eqv}) and~(\ref{Vv}), we obtain
\begin{equation}\label{lvkappa}
\mathscr L_{\alpha,z}\mathscr V^{\alpha}(z) = 0,\quad z\in\C_{3\pi/2},\,\alpha\in\C.
\end{equation}

Using the well-known asymptotic form of $H^{(1)}_\kappa(\zeta)$ for $\zeta\to\infty$ (see formula~(1) in~Sec.~7.13.1 in~\cite{Bateman}), we find that
\begin{equation}\nonumber
\mathfrak v^\kappa(z|r) \sim 2^{-1}\sqrt{\pi}(i+1)z^{-1/4} e^{iz^{1/2}r},\quad r\to\infty,
\end{equation}
for every $\kappa\in \C$ and $z\in\C_{3\pi/2}$ and, hence,
$\mathfrak v^\kappa(z)$ is right square-integrable for all $\kappa\in \C$ and $z\in\C_+$. In view of~(\ref{Vv}), this implies that $\mathscr V^\alpha(z)$ is right square-integrable for all $\alpha\in \C$ and $z\in\C_+$.

\begin{lemma}\label{l_analyt1}
There is a unique holomorphic function $F$ on $\C\times\C_{3\pi/2}\times\C_\pi$ such that $F(\alpha,z,r) = \mathscr V^\alpha(z|r)$ for every $\alpha\in\C$, $z\in\C_{3\pi/2}$, and $r>0$.
\end{lemma}
\begin{proof}
By~(\ref{vkappa}), there is a holomorphic function $G$ on $\C\times\C_{3\pi/2}\times\C_\pi$ such that $G(\kappa,z,r) = \mathfrak v^\kappa(z|r)$ for every $\kappa\in\C$, $z\in\C_{3\pi/2}$, and $r>0$. It follows from~(\ref{-kappa}) and the uniqueness theorem for holomorphic functions that $G(\kappa,z,\zeta) = G(-\kappa,z,\zeta)$ for all $\kappa\in\C$, $z\in\C_{3\pi/2}$, and $\zeta\in\C_\pi$. The existence of $F$ with the required properties is now ensured by Lemma~\ref{l_root1} and~(\ref{Vv}). The uniqueness of $F$ follows from the uniqueness theorem for holomorphic functions.
\end{proof}

It follows immediately from~(\ref{wronskian}) that the identity
\begin{equation}\label{wrprod}
W_r(f_1 f_2,f_3f_4) = f_1(r)f_3(r)W_r(f_2,f_4) + W_r(f_1,f_3)f_2(r)f_4(r)
\end{equation}
holds for every $f_1,f_2,f_3,f_4\in \mathcal D$ and $r>0$.

\begin{lemma}\label{lWu}
Let $\kappa\in\C$ be such that $|\Re\kappa| < 1$. Then we have
\begin{equation}\label{WW}
W^\downarrow(\mathfrak u^\kappa(z),\mathfrak u^\kappa(z')) = 0,\quad  W^\downarrow(\mathfrak u^\kappa(z),\mathfrak u^{-\kappa}(z')) = -\frac{2}{\pi}\sin \pi\kappa
\end{equation}
for every $z,z'\in\C$.
\end{lemma}
\begin{proof}
By~(\ref{ukappa}) and~(\ref{wrprod}), we have
\begin{equation}
W_r(\mathfrak u^\kappa(z),\mathfrak u^\kappa(z')) = 2r^{2+2\kappa}(z'\mathcal X_\kappa(r^2z)\mathcal X'_\kappa(r^2z')-z\mathcal X'_\kappa(r^2z)\mathcal X_\kappa(r^2z')),\label{W+}
\end{equation}
\begin{multline}
W_r(\mathfrak u^\kappa(z),\mathfrak u^{-\kappa}(z')) = - 2\kappa \mathcal X_\kappa(r^2z)\mathcal X_{-\kappa}(r^2z')+\\+2r^2(z'\mathcal X_\kappa(r^2z)\mathcal X'_{-\kappa}(r^2z')-z\mathcal X'_\kappa(r^2z)\mathcal X_{-\kappa}(r^2z'))
\label{W-}
\end{multline}
for all $z,z'\in\C$ and $r>0$. Since $|\Re\kappa| < 1$, the left equality in~(\ref{WW}) follows from~(\ref{W+}) for every $z,z'\in\C$. Formula~(\ref{Xkappa}) implies that
\[
\mathcal X_\kappa(0)\mathcal X_{-\kappa}(0) = (\Gamma(1+\kappa)\Gamma(1-\kappa))^{-1} = \sinc \pi\kappa.
\]
By~(\ref{W-}), we conclude that the right equality in~(\ref{WW}) holds for all $z,z'\in\C$.
\end{proof}

\begin{lemma}\label{lW0}
Let $\alpha<1$ and $\vartheta\in \R$. Then we have
\begin{equation}\label{wqalpha}
W^\downarrow(\mathcal U^\alpha_\vartheta(z),\mathcal U^\alpha_\vartheta(z')) = 0, \quad z,z'\in\C.
\end{equation}
\end{lemma}
\begin{proof}
Let $\alpha\neq 0$ and $\kappa\in\C$ be such that $\kappa^2 = \alpha$. By~(\ref{wkappa}), we have
\begin{multline}\label{WWW}
\alpha W_r(\mathcal U^\alpha_\vartheta(z),\mathcal U^\alpha_\vartheta(z')) = \\ W_r(u^\kappa(z),u^\kappa(z'))\cos^2(\vartheta-\vartheta_\kappa) + W_r(u^{-\kappa}(z),u^{-\kappa}(z'))\cos^2(\vartheta+\vartheta_\kappa) + \\+(W_r(u^\kappa(z'),u^{-\kappa}(z)) - W_r(u^\kappa(z),u^{-\kappa}(z')))\cos(\vartheta-\vartheta_\kappa)\cos(\vartheta+\vartheta_\kappa)
\end{multline}
for all $z,z'\in\C$ and $r>0$. Since $\alpha<1$, we have $|\Re\kappa| < 1$ and, therefore, (\ref{wqalpha}) follows from~(\ref{WWW}) and Lemma~\ref{lWu}.
It remains to consider the case $\alpha=0$. For every $z\in\C$, we define the smooth functions $c(z)$ and $d(z)$ on $\R$ by the relations
\[
c(z|r) = \mathcal X_0(r^2z),\quad d(z|r) = (\gamma-\ln 2)c(z|r)-\mathcal Y(r^2z),\quad r\in\R.
\]
By~(\ref{ukappa}), (\ref{a0}), (\ref{bkappa}), and~(\ref{Aa}), we have
\[
\mathcal A^0(z|r) = 2r^{1/2}(c(z|r)\ln r + d(z|r)),\quad \mathcal B^0(z|r) = \pi r^{1/2} c(z|r),\quad z\in\C,\,r>0.
\]
Using~(\ref{wrprod}), we obtain
\begin{align}
&\frac{1}{4} W_r(\mathcal A^0(z),\mathcal A^0(z')) =  r\ln^2rW_r(c(z),c(z'))+rW_r(d(z),d(z'))+\nonumber\\
&+r\ln r(W_r(c(z),d(z'))+W_r(d(z),c(z'))) +d(z|r)c(z'|r)- c(z|r)d(z'|r),\nonumber\\
&\frac{1}{2\pi} W_r(\mathcal A^0(z),\mathcal B^0(z')) =  r W_r(d(z),c(z'))+r\ln r W_r(c(z),c(z'))-c(z|r)c(z'|r),\nonumber\\
& W_r(\mathcal B^0(z),\mathcal B^0(z')) = \pi^2 r W_r(c(z),c(z'))\nonumber
\end{align}
for every $r>0$ and $z,z'\in\C$. Since $c(z|0) = 1$ and $d(z|0) = \gamma-\ln 2$ for all $z\in\C$, we find that
\[
W^\downarrow(\mathcal A^0(z),\mathcal B^0(z')) = -2\pi,\quad W^\downarrow(\mathcal A^0(z),\mathcal A^0(z')) = W^\downarrow(\mathcal B^0(z),\mathcal B^0(z')) = 0
\]
for all $z,z'\in\C$. In view of~(\ref{Ualphatheta}), this implies~(\ref{wqalpha}) for $\alpha=0$.
\end{proof}

It follows from~(\ref{lalpha}) and~(\ref{lvkappa}) that $W_r(\mathscr V^\alpha(z),\mathcal U^\alpha_\vartheta(z))$ does not depend on $r$ for every $\alpha,\vartheta\in\C$ and $z\in \C_{3\pi/2}$. It is easy to see that
\begin{equation}\label{xxx}
W(\mathscr V^\alpha(z),\mathcal U^\alpha_\vartheta(z)) = R(\alpha,\vartheta,z),\quad \alpha,\vartheta\in\C,\,z\in \C_{3\pi/2},
\end{equation}
where $R$ is the function defined in Lemma~\ref{lR}. Indeed, it follows from~(\ref{wkappa}), (\ref{wrsol1}), and~(\ref{Vv}) that $W(\mathscr V^{\kappa^2}(z),\mathcal U^{\kappa^2}_\vartheta(z))$ is equal to the right-hand side of~(\ref{R}) for every $\kappa\in \C\setminus\{0\}$, $\vartheta\in\C$, and $z\in\C_{3\pi/2}$. Equality~(\ref{xxx}) therefore holds for all $\alpha\in\C\setminus\{0\}$, $\vartheta\in\C$, and $z\in\C_{3\pi/2}$. By~(\ref{Ualphatheta}) and Lemmas~\ref{l_analyt}, \ref{lR}, and~\ref{l_analyt1}, both sides of~(\ref{xxx}) are holomorphic with respect to $(\alpha,\vartheta,z)$ on $\C\times\C\times\C_{3\pi/2}$. Hence, (\ref{xxx}) remains valid for $\alpha=0$.

By~(\ref{Ualphatheta}) and~(\ref{WABalpha}), we have
\begin{equation}\label{WUalpha}
W(\mathcal U^\alpha_\vartheta(z),\mathcal U^\alpha_{\vartheta+\pi/2}(z)) = -2\pi\Sinc^2(\pi^2\alpha),\quad \alpha,\vartheta,z\in\C.
\end{equation}

\begin{proof}[Proof of Theorem~$\ref{leig2}$]
It follows immediately from~(\ref{Ualphatheta}) and the definition of $\mathcal A^\alpha(z)$ and $\mathcal B^\alpha(z)$ that the functions $z\to \mathcal U^\alpha_\vartheta(z|r)$ and $z\to \partial_r\mathcal U^\alpha_\vartheta(z|r)$ are holomorphic on $\C$ for every $r>0$ and $\alpha,\vartheta\in\C$. Equality~(\ref{lqalpha}) therefore implies that $\mathcal U^\alpha_\vartheta$ is a holomorphic $q_\alpha$-solution on $\C$ for every $\alpha,\vartheta\in\C$. If $\alpha$ and $\vartheta$ are real, then $\mathcal U^\alpha_\vartheta(E)$ is real for real $E$ and, hence, $\mathcal U^\alpha_\vartheta$ is a real-entire $q_\alpha$-solution. By~(\ref{vkappa}) and~(\ref{Vv}), the functions $z\to \mathscr V^\alpha(z|r)$ and $z\to \partial_r\mathscr V^\alpha(z|r)$ are holomorphic on $\C_{3\pi/2}$ for every $r>0$ and $\alpha\in\C$. In view of~(\ref{eq1}) and~(\ref{lvkappa}), it follows that $\mathscr V^\alpha$ is a holomorphic $q_\alpha$-solution on $\C_{3\pi/2}$ for every $\alpha\in\C$. Moreover, $\mathscr V^\alpha$ is nonvanishing for every $\alpha\in\C$ by~(\ref{wrsol1}) and~(\ref{Vv}).

We now fix $\alpha<1$ and $\vartheta\in\R$, set $q = q_\alpha$, and let $u$ and $\tilde u$ denote the real-entire $q$-solutions $\mathcal U^\alpha_\vartheta$ and $\mathcal U^\alpha_{\vartheta+\pi/2}$ respectively. Because $\alpha<1$, (\ref{WUalpha}) implies that
\begin{equation}\label{nonzeroW}
W(u(z),\tilde u(z))\neq 0,\quad z\in\C,
\end{equation}
and, hence, $u$ is nonvanishing. By Lemma~\ref{lW0}, we conclude that $q$ and $u$ satisfy all conditions of Proposition~\ref{t_eig}. Moreover, it follows from~(\ref{nonzeroW}) that there exists a well-defined Titchmarsh--Weyl $m$-function $\mathcal M^{q}_{u,\tilde u}$. Since $\mathscr V^\alpha(z)$ is right square-integrable for every $z\in\C_+$, the latter can be found by substituting $v = \mathscr V^\alpha|_{\C_+}$ in the right-hand side of~(\ref{mathcalM}). Using (\ref{M}), (\ref{xxx}), and~(\ref{WUalpha}), we obtain $\mathcal M^{q}_{u,\tilde u}(z) = \mathscr M_{\alpha,\vartheta}(z)$ for every $z\in\C_+$. In view of~(\ref{nonzeroW}), statement~1 of Proposition~\ref{t_eig} and formula~(\ref{nualphavartheta}) imply that $\mathcal V_{\alpha,\vartheta}$ is precisely the spectral measure for $q$ and $u$. The existence and uniqueness of $U_{\alpha,\vartheta}$ and equality~(\ref{diag}) are therefore ensured by statements~2 and~3 of Proposition~\ref{t_eig} and formula~(\ref{halphatheta1}).
\end{proof}

\section{Smoothness properties of the spectral measure}
\label{s_smooth}

In this section, we shall establish Theorem~\ref{t_smooth}.

Before passing to the proof, we note that the smoothness of $\mathcal V_{\alpha,\vartheta}$ with respect to $\alpha$ and $\vartheta$ is suggested by explicit formulas obtained in Sec.~\ref{s_meas}. For example, let us set $\vartheta=\pi/2$ and examine the $\alpha$-dependence of the point part $\mathcal V^p_{\alpha,\pi/2}$ of the measure $\mathcal V_{\alpha,\pi/2}$. It follows from~(\ref{Sigma}), (\ref{W}), Lemma~\ref{l_S}, and Theorem~\ref{t_N} that $\Sigma_{\alpha,\pi/2}=\varnothing$ for $\alpha\geq 0$ and
\[
\Sigma_{\alpha,\pi/2} = \left\{E\in\R: E = -\exp\left((2k+1)\pi|\alpha|^{-1/2}\right)\mbox{ for some }k\in\Z\right\}
\]
for $\alpha<0$. Let $\varphi\in \mathscr S$. In view of~(\ref{nup}) and~(\ref{Sinc2Cos2}), Theorem~\ref{t_measure} implies that
\begin{equation}\label{Vp}
\int \varphi(E)\,d\mathcal V^p_{\alpha,\pi/2}(E) = \frac{2}{\pi^2\Sinc(\pi^2\alpha)\Sinc^2(\pi^2\alpha/4)}\sum_{k\in\Z} \varphi_{(2k+1)\pi}(\alpha)
\end{equation}
for every $\alpha<1$, where $\varphi_c$, $c\in\R$, denotes the function on $\R$ that is identically zero on $[0,\infty)$ and is given by
\[
\varphi_c(\alpha) =
\frac{1}{|\alpha|}\exp\left(c|\alpha|^{-1/2}\right)\varphi\left(-\exp\left(c|\alpha|^{-1/2}\right)\right)
\]
for $\alpha<0$. It follows immediately from the definition of the space $\mathscr S$ that $\varphi_c$ is an infinitely differentiable function on $\R$ for every nonzero real $c$ (however, in general, $\varphi_c$ is not real-analytic at $\alpha=0$ even for real-analytic $\varphi$). It is also possible to verify directly that
$\sum_{k\in\Z} \varphi_{(2k+1)\pi}$ is infinitely differentiable on $\R$ and, therefore, the left-hand side of~(\ref{Vp}) is infinitely differentiable with respect to $\alpha$ on $(-\infty,1)$.

It seems, however, that a complete proof of Theorem~\ref{t_smooth} (including the continuous part of $\mathcal V_{\alpha,\vartheta}$ and the dependence on both $\alpha$ and $\vartheta$) based on explicit formulas for the spectral measures would be extremely cumbersome. We shall adopt a different approach based on representation~(\ref{nualphavartheta}) of $\mathcal V_{\alpha,\vartheta}$ as a boundary value of the holomorphic function $\mathscr M_{\alpha,\vartheta}$. The idea is to derive the infinite differentiability of $\mathcal V_{\alpha,\vartheta}$ with respect to $\alpha$ and $\vartheta$ from that of $\mathscr M_{\alpha,\vartheta}$. Lemma~\ref{l_diff} below gives a condition under which the differentiability of a holomorphic function on $\C_+$ with respect to some parameters implies the same property for its boundary value. This condition involves certain uniform estimates on the derivatives of this function with respect to the parameters in question. In the case of the function $\mathscr M_{\alpha,\vartheta}$, estimates of this type, which are the most nontrivial part of the proof of Theorem~\ref{t_smooth}, are provided by Proposition~\ref{p_estimate} below. Combining Lemma~\ref{l_diff} and Proposition~\ref{p_estimate}, we shall obtain the infinite differentiability of $\mathcal V_{\alpha,\vartheta}$.

We now give a formal exposition.

For every $\varrho,\sigma\geq 0$, we let $\mathcal H_{\varrho,\sigma}$ denote the Banach space consisting of all holomorphic functions on $\C_+$ with the finite norm
\[
\|f\|_{\varrho,\sigma} = \sup_{z\in\C_+} |f(z)| \mathcal N_{\varrho,\sigma}(z),
\]
where the function $\mathcal N_{\varrho,\sigma}$ on $\C_+$ is given by
\[
\mathcal N_{\varrho,\sigma}(z) = \frac{1}{(1+|z|)^\varrho} \left(\frac{\Im z}{1+|z|}\right)^\sigma.
\]
If $\varrho'\geq \varrho\geq 0$ and $\sigma'\geq\sigma\geq 0$, then $\mathcal H_{\varrho,\sigma}\subset \mathcal H_{\varrho',\sigma'}$ and the inclusion map is continuous.
We define the linear space $\mathcal H$ by setting $\mathcal H = \bigcup_{\varrho,\sigma\geq 0} \mathcal H_{\varrho,\sigma}$.
It is well known (see, e.g., \cite{Vladimirov1966}, Ch.~5, Sec.~26.3) that $\int_{\R} f(E+i\eta)\varphi(E)\,dE$ has a limit as $\eta\downarrow 0$ for every $f\in \mathcal H$ and $\varphi\in \mathscr S$. Given $\varphi\in \mathscr S$, we let $B_\varphi$ denote the map $f\to \lim_{\eta\downarrow 0}\int_{\R} f(E+i\eta)\varphi(E)\,dE$ on $\mathcal H$. The definitions of $\mathscr S$ and $\mathcal H_{\varrho,\sigma}$ imply that $f\to \int_{\R} f(E+i\eta)\varphi(E)\,dE$ is a continuous linear functional on $\mathcal H_{\varrho,\sigma}$ for every $\eta>0$, $\varphi\in \mathscr S$, and $\varrho,\sigma\geq 0$. By the Banach-Steinhaus theorem (see~Theorem~III.9 in~\cite{RS1}), it follows that the restriction of $B_\varphi$ to $\mathcal H_{\varrho,\sigma}$ is a continuous linear functional on $\mathcal H_{\varrho,\sigma}$ for every $\varphi\in \mathscr S$ and $\varrho,\sigma\geq 0$.

\begin{lemma}\label{l_diff}
Let $\varrho,\sigma\geq 0$, $n=1,2,\ldots$, and $\varphi\in \mathscr S$. Let $O$ be an open subset of $\R^n$ and $M$ be a map from $O$ to $\mathcal H_{\varrho,\sigma}$ such that $x\to M(x|z)$ is a continuously differentiable function on $O$ for every $z\in\C_+$. For every $j=1,\ldots,n$ and $x\in O$, let the function $M_j(x)$ on $\C_+$ be defined by the equality $M_j(x|z)=\partial_{x_j}M(x|z)$. Suppose there is $C>0$ such that
\begin{equation}\label{Mjbound}
|M_j(x|z)|\leq C \mathcal N_{\varrho,\sigma}(z)^{-1},\quad j=1,\ldots,n, \,\,x\in O,\, z\in\C_+.
\end{equation}
Then the function $x\to B_\varphi(M(x))$ on $O$ is continuously differentiable and we have $M_j(x)\in \mathcal H_{\varrho,\sigma}$ and $\partial_{x_j}B_\varphi(M(x))= B_\varphi(M_j(x))$ for all $j=1,\ldots,n$ and $x\in O$.
\end{lemma}
\begin{proof}
The proof relies on the following convergence property for sequences of holomorphic functions, which easily follows from the Montel theorem (see, e.g., Theorem~12 in Sec.~5.4.4 in~\cite{Ahlfors}).
\begin{itemize}
  \item[(C)] Let $V\subset\C$ be an open set, $f_1,f_2,\ldots$ be holomorphic functions on $V$, and $f$ be a complex function on $V$. Suppose the functions $f_k$ are uniformly bounded on compact subsets of $V$ and $f_k(z)\to f(z)$ as $k\to \infty$ for every $z\in V$. Then $f$ is holomorphic on $V$ and $f_k\to f$ as $k\to\infty$ uniformly on compact subsets of $V$.
\end{itemize}

Let $e_1,\ldots,e_n$ be the standard basis in $\R^n$. Given $x\in O$ and $j = 1,\ldots,n$, we choose $\delta>0$ such that $x+t e_j\in O$ for every $t\in [0,\delta]$. We define the holomorphic functions $h_1,h_2,\ldots$ on $\C_+$ by setting $h_k = t_k^{-1}(M(x+t_ke_j)-M(x))$, where $t_k = \delta/k$. By~(\ref{Mjbound}) and the mean value theorem, we conclude that $|h_k(z)|\leq C \mathcal N_{\varrho,\sigma}(z)^{-1}$ for every $z\in\C_+$ and $k=1,2,\ldots$ and, therefore, the functions $h_k$ are uniformly bounded on compact subsets of $\C_+$. Since $\lim_{k\to\infty} h_k(z) = M_j(x|z)$ for every $z\in\C_+$, property~(C) implies that $M_j(x)$ is holomorphic on $\C_+$. In view of~(\ref{Mjbound}), it follows that $M_j(x)\in \mathcal H_{\varrho,\sigma}$ and $\|M_j(x)\|_{\varrho,\sigma}\leq C$ for every $x\in O$ and $j=1,\ldots,n$.

We now choose $\varrho'>\varrho$, $\sigma'>\sigma$ and let $x\in O$ and $h=(h_1,\ldots,h_n)\in \R^n$ be such that $x+ t h\in O$ for all $t\in [0,1]$. By the mean value theorem, we have
\begin{multline}\nonumber
\left|M(x+h|z)-M(x|z)-\sum\nolimits_j M_j(x|z)h_j\right| \leq\\ \leq|h|\sup_{0<t<1}\sum\nolimits_j |M_j(x+th|z)-M_j(x|z)|
\end{multline}
for every $z\in\C_+$, where $|h|=\max_{1\leq j\leq n} |h_j|$. This implies that
\[
\left\|M(x+h)-M(x)-\sum\nolimits_j M_j(x)h_j\right\|_{\varrho',\sigma'}\leq |h|\sup_{0<t<1} \sum\nolimits_j \|M_j(x+th)-M_j(x)\|_{\varrho',\sigma'}.
\]
Since $B_\varphi$ is continuous on $\mathcal H_{\varrho',\sigma'}$, our statement will be proved if we show that $M_j$ is a continuous map from $O$ to $\mathcal H_{\varrho',\sigma'}$ for every $j=1,\ldots,n$. To this end, we fix $\varepsilon>0$ and choose a compact subset $K$ of $\C_+$ such that
\[
\sup_{z\in \C_+\setminus K} \mathcal N_{\varrho'-\varrho,\sigma'-\sigma}(z) < \frac{\varepsilon}{2C}.
\]
Then we have
\[
\sup_{z\in \C_+\setminus K}|M_j(x'|z)-M_j(x|z)|\mathcal N_{\varrho',\sigma'}(z)\leq \frac{\varepsilon}{2C}\|M_j(x')-M_j(x)\|_{\varrho,\sigma}  <\varepsilon
\]
for all $x,x'\in O$ and $j=1,\ldots,n$. On the other hand, property~(C) and the continuity of the functions $y\to M_j(y|z)$ on $O$ for every $z\in\C_+$ imply that every $x\in O$ has a neighborhood $O_x\subset O$ such that
\[
\sup_{z\in K}|M_j(x'|z)-M_j(x|z)|\mathcal N_{\varrho',\sigma'}(z) <\varepsilon,\quad x'\in O_x,\,j=1,\ldots,n.
\]
Hence, $\|M_j(x')-M_j(x)\|_{\varrho',\sigma'}<\varepsilon$ for all $x'\in O_x$ and $j=1,\ldots,n$, i.e., $M_j$ is a continuous map from $O$ to $\mathcal H_{\varrho',\sigma'}$ for every $j=1,\ldots,n$.
\end{proof}

Let $R$ be as in Lemma~$\ref{lR}$. We set
\begin{equation}\nonumber
\mathscr O = \{(\alpha,\vartheta,z)\in \C\times\C\times\C_{3\pi/2}: R(\alpha,\vartheta,z)\neq 0\}
\end{equation}
and define the holomorphic function $\mathscr F$ on $\mathscr O$ by the equality
\begin{equation}\label{scrF}
\mathscr F(\alpha,\vartheta,z) = \frac{R(\alpha,\vartheta+\pi/2,z)}{R(\alpha,\vartheta,z)}, \quad (\alpha,\vartheta,z)\in \mathscr O.
\end{equation}
It follows from~(\ref{Oalphatheta}), (\ref{subset}), and~(\ref{M}) that $(\alpha,\vartheta,z)\in \mathscr O$ and
\begin{equation}\label{MtildeM}
\mathscr M_{\alpha,\vartheta}(z) = -\frac{\mathscr F(\alpha,\vartheta,z)}{2\pi^2\Sinc^2(\pi^2\alpha)}
\end{equation}
for every $\alpha<1$, $\vartheta\in \R$, and $z\in \C_+$.

In the sequel, we let $\Z_+$ denote the set of all nonnegative integer numbers.
\begin{proposition}\label{p_estimate}
Let $0\leq a<1$, $b\in\R$, and $k,l\in\Z_+$. Then we have
\[
\left|\partial_\alpha^k\partial_\vartheta^l\mathscr F(\alpha,\vartheta,z)\right|\leq P_{a,b}(k,l)(1+\ln^2|z|)^{2k+l+1} \left(\frac{(1+|z|)^{1+a}}{\Im z}\right)^{k+l+1}
\]
for every $\alpha\in [-b^2,a^2]$, $\vartheta\in\R$, and $z\in \C_+$, where
\[
P_{a,b}(k,l) = \frac{\pi^{2k}}{2}\left(\frac{24\pi\ch(\pi b)}{\sinc^2(\pi a)}\right)^{k+l+1}(k+l)!.
\]
\end{proposition}

Before proceeding with the proof of Proposition~\ref{p_estimate}, we shall derive Theorem~\ref{t_smooth} from Lemma~\ref{l_diff} and Proposition~\ref{p_estimate}.

\begin{proof}[Proof of Theorem~$\ref{t_smooth}$]
For every $k,l\in\Z_+$ and $(\alpha,\vartheta)\in Q$, we define the holomorphic function $M_{k,l}(\alpha,\vartheta)$ on $\C_+$ by the equality $M_{k,l}(\alpha,\vartheta|z)=\partial^k_\alpha\partial^l_\vartheta\mathscr F(\alpha,\vartheta,z)$, $z\in\C_+$. Clearly, the function $(\alpha,\vartheta)\to M_{k,l}(\alpha,\vartheta|z)$ on $Q$ is infinitely differentiable for every $z\in\C_+$ and we have
\begin{equation}\label{k+1,l+1}
\partial_\alpha M_{k,l}(\alpha,\vartheta|z) = M_{k+1,l}(\alpha,\vartheta|z),\quad \partial_\vartheta M_{k,l}(\alpha,\vartheta|z) = M_{k,l+1}(\alpha,\vartheta|z)
\end{equation}
for every $(\alpha,\vartheta)\in Q$, $z\in\C_+$, and $k,l\in\Z_+$.
Given $0\leq a<1$ and $b\in\R$, we set $O_{a,b} = \{(\alpha,\vartheta)\in Q: -b^2<\alpha<a^2\}$. By Proposition~\ref{p_estimate}, we have
\begin{equation}\label{Mklest}
|M_{k,l}(\alpha,\vartheta|z)| \leq P_{a,b}(k,l)\sup_{E>0} \frac{E(1+\ln^2E)^{2k+l+1}}{(1+E)^2} \mathcal N_{k+l+2,k+l+2}(z)^{-1}
\end{equation}
for every $(\alpha,\vartheta)\in O_{a,b}$, $z\in\C_+$, and $k,l\in\Z_+$. Since
\begin{equation}\label{union}
Q = \bigcup_{a\in[0,1),\,b\in\R} O_{a,b},
\end{equation}
this implies that $M_{k,l}(\alpha,\vartheta)\in \mathcal H$ for every $(\alpha,\vartheta)\in Q$ and $k,l\in\Z_+$. We now fix $\varphi\in \mathscr S$. Given $k,l\in\Z_+$, we let $F_{k,l}$ denote the function $(\alpha,\vartheta)\to B_\varphi(M_{k,l}(\alpha,\vartheta))$ on $Q$. Let $0\leq a<1$, $b\in\R$, and $k,l\in\Z_+$. In view of~(\ref{k+1,l+1}) and~(\ref{Mklest}), we can apply Lemma~\ref{l_diff} to $O=O_{a,b}$, $M = M_{k,l}|_O$, and $\varrho = \sigma = k+l+3$. As a result, we conclude that $F_{k,l}$ is continuously differentiable on $O_{a,b}$ and
\begin{equation}\label{Fkl}
\partial_\alpha F_{k,l}(\alpha,\vartheta) = F_{k+1,l}(\alpha,\vartheta),\quad \partial_\vartheta F_{k,l}(\alpha,\vartheta) = F_{k,l+1}(\alpha,\vartheta)
\end{equation}
for all $(\alpha,\vartheta)\in O_{a,b}$. By~(\ref{union}), it follows that $F_{k,l}$ is continuously differentiable on $Q$ and equalities~(\ref{Fkl}) hold for all $(\alpha,\vartheta)\in Q$ and $k,l\in\Z_+$. We now use induction on $n$ to prove the following statement
\begin{itemize}
\item[($S_n$)] $F_{0,0}$ is $n$ times differentiable on $Q$ and $\partial_\alpha^k\partial_\vartheta^l F_{0,0}(\alpha,\vartheta) = F_{k,l}(\alpha,\vartheta)$ for every $(\alpha,\vartheta)\in Q$ and $k,l\in\Z_+$ such that $k+l=n$.
\end{itemize}
The statement~($S_0$) trivially holds because every function on $Q$ is $0$ times differentiable. We now suppose $n\geq 1$ and derive~($S_n$) from~($S_{n-1}$). Let $k,l\in\Z_+$ be such that $k+l=n-1$. Since $F_{k,l}$ is differentiable, ($S_{n-1}$) implies that the function $(\alpha,\vartheta)\to \partial_\alpha^k\partial_\vartheta^l F_{0,0}(\alpha,\vartheta)$ on $Q$ is differentiable. This means that $F_{0,0}$ is $n$ times differentiable. Suppose now $k,l\in\Z_+$ are such that $k+l=n$. Then we have either $k>0$ or $l>0$. Hence, we can represent $\partial_\alpha^k\partial_\vartheta^l F_{0,0}(\alpha,\vartheta)$ either as $\partial_\alpha(\partial_\alpha^{k-1}\partial_\vartheta^l F_{0,0}(\alpha,\vartheta))$ or as $\partial_\vartheta(\partial_\alpha^{k}\partial_\vartheta^{l-1} F_{0,0}(\alpha,\vartheta))$. In both cases, it follows from~($S_{n-1}$) and~(\ref{Fkl}) that $\partial_\alpha^k\partial_\vartheta^l F_{0,0}(\alpha,\vartheta)=F_{k,l}(\alpha,\vartheta)$ for all $(\alpha,\vartheta)\in Q$. This completes the derivation of~($S_n$) from~($S_{n-1}$). By induction, we conclude that ($S_n$) holds for all $n\in\Z_+$ and, therefore, the function $F_{0,0}$ is infinitely differentiable. Given $(\alpha,\vartheta)\in Q$, we set $G_{\alpha,\vartheta}=\mathscr M_{\alpha,\vartheta}|_{\C_+}$. By~(\ref{MtildeM}), we have $M_{0,0}(\alpha,\vartheta) = -2\pi^2\Sinc(\pi^2\alpha) G_{\alpha,\vartheta}$ for every $(\alpha,\vartheta)\in Q$. Since $F_{0,0}(\alpha,\vartheta) = B_\varphi(M_{0,0}(\alpha,\vartheta))$, we conclude that the function $(\alpha,\vartheta)\to B_\varphi(G_{\alpha,\vartheta})$ on $Q$ is infinitely differentiable for every $\varphi\in\mathscr S$. To complete the proof, it remains to note that
\[
\int \varphi(E)\,d\mathcal V_{\alpha,\vartheta}(E) = \frac{B_\varphi(G_{\alpha,\vartheta}) - \overline{B_{\bar\varphi}(G_{\alpha,\vartheta})}}{2i}
\]
for every $(\alpha,\vartheta)\in Q$ and $\varphi\in \mathscr S$ by~(\ref{nualphavartheta}).
\end{proof}

The rest of this section is devoted to the proof of Proposition~\ref{p_estimate}.

It follows from~(\ref{tau'}) and~(\ref{mu'}) that
\[
\mu(\kappa^2,\phi)^2 - \tau(\kappa^2,\phi)^2 = 4\pi^2(\pi-\phi)^2\sinc^2\pi\kappa\,\sinc^2(\pi-\phi)\kappa
\]
for every $\kappa\in\C\setminus\{0\}$ and $\phi\in\C$. By continuity, this equality remains valid for $\kappa = 0$. Hence, we have
\begin{equation}\label{mutau}
\mu(\alpha,\phi)^2 - \tau(\alpha,\phi)^2 = 4\pi^2(\pi-\phi)^2\Sinc^2(\pi^2\alpha)\Sinc^2((\pi-\phi)^2\alpha),\quad \alpha,\phi\in\C.
\end{equation}

Let the function $\Phi$ on $\C\times\R_+$ be defined by the relation
\begin{equation}\label{Phi}
\Phi(\alpha,E) = \ln E \Sinc\left(-\frac{\alpha}{4}\ln^2 E\right),\quad \alpha\in\C,\,E>0.
\end{equation}
In view of~(\ref{Sinc4}), rewriting equality~(\ref{T'''}) in terms of $\Phi$ yields
\begin{multline}\label{T1}
T(\alpha,\vartheta,Ee^{i\phi}) = \mu(\alpha,\phi)+\Phi(\alpha,E)^2 + \\ + \left(\Phi(\alpha,E)^2\Cos(\pi^2\alpha)+\tau(\alpha,\phi)\right)\cos2\vartheta - \\ - 2\pi \Phi(\alpha,E)\Cos\left(-\frac{\alpha}{4}\ln^2E\right)\Sinc(\pi^2\alpha) \sin 2\vartheta
\end{multline}
for all $\alpha,\vartheta\in \C$, $E>0$, and $-\pi/2<\phi<3\pi/2$.

By~(\ref{tau}), $\tau(\alpha,\phi)$ is real for real $\alpha$ and $\phi$.
We also observe that
\begin{equation}\label{mu>0}
\mu(\alpha,\phi) \geq 0,\quad \alpha,\phi\in\R.
\end{equation}
Indeed, (\ref{mu'}) implies that this is true for all $\alpha\in \R\setminus\{0\}$ and $\phi\in\R$. By continuity, (\ref{mu>0}) remains valid for $\alpha=0$.

The next lemma is the key part of the proof of Proposition~\ref{p_estimate}.

\begin{lemma}\label{l_R}
Let $\alpha<1$, $\vartheta\in\R$, $E>0$, and $0\leq \phi<\pi$. Then we have
\[
\frac{1}{|R(\alpha,\vartheta,Ee^{i\phi})|} \leq \frac{(\Phi(\alpha,E)^2 + \mu(\alpha,\phi))^{1/2}}{\pi(\pi-\phi)\Sinc(\pi^2\alpha)\Sinc((\pi-\phi)^2\alpha)}.
\]
\end{lemma}
\begin{proof}
Let
\begin{align}
& G = \Phi(\alpha,E)^2 + \mu(\alpha,\phi),\label{G}\\
& H = \sqrt{\tau(\alpha,\phi)^2+\Phi(\alpha,E)^4+2\mu(\alpha,\phi)\Phi(\alpha,E)^2}. \label{H}
\end{align}
By~(\ref{mu>0}), $H$ is well-defined and both $G$ and $H$ are nonnegative. Using~(\ref{mu}) and the identity
\[
\Cos^2\left(-\frac{\alpha}{4}\ln^2E\right) = 1 + \frac{\alpha}{4}\Phi(\alpha,E)^2,
\]
which follows from~(\ref{Phi}) and~(\ref{Sinc2Cos2}), we find that
\begin{equation}\label{Hcd}
H = \sqrt{c^2 + d^2},
\end{equation}
where
\[
c = \Phi(\alpha,E)^2\Cos(\pi^2\alpha) + \tau(\alpha,\phi),\quad d = -2\pi\Phi(\alpha,E)\Cos\left(-\frac{\alpha}{4}\ln^2E\right)\Sinc(\pi^2\alpha).
\]
Since $|c\cos2\vartheta + d\sin2\vartheta|\leq \sqrt{c^2+d^2}$ by the Cauchy--Bunyakovsky inequality, (\ref{T1}) and~(\ref{Hcd}) imply that
\begin{equation}\label{ineqT}
T(\alpha,\vartheta,Ee^{i\phi}) = G + c\cos2\vartheta + d\sin2\vartheta \geq G - H.
\end{equation}
By~(\ref{G}) and~(\ref{H}), we have
\begin{equation}\label{GH}
G^2 - H^2 = \mu(\alpha,\phi)^2 - \tau(\alpha,\phi)^2.
\end{equation}
In view of~(\ref{mutau}), it follows that the left-hand side of~(\ref{GH}) is nonnegative. We therefore have $H\leq G$, whence $2G\geq G+H$. Multiplying this inequality with~(\ref{ineqT}) and using~(\ref{mutau}) and~(\ref{GH}), we obtain
\begin{equation}\label{2GT}
2G T(\alpha,\vartheta,Ee^{i\phi}) \geq 4\pi^2(\pi-\phi)^2\Sinc^2(\pi^2\alpha)\Sinc^2((\pi-\phi)^2\alpha).
\end{equation}
Since $\alpha<1$ and $0\leq \phi<\pi$, the right-hand side of~(\ref{2GT}) is strictly positive. Hence, $T(\alpha,\vartheta,Ee^{i\phi})>0$, and (\ref{G}) and~(\ref{2GT}) imply that
\[
\frac{1}{T(\alpha,\vartheta,Ee^{i\phi})} \leq \frac{\Phi(\alpha,E)^2 + \mu(\alpha,\phi)}{2\pi^2(\pi-\phi)^2\Sinc^2(\pi^2\alpha)\Sinc^2((\pi-\phi)^2\alpha)}.
\]
The required estimate is now ensured by~(\ref{TR}) because $\Sinc(\pi^2\alpha)$ and $\Sinc((\pi-\phi)^2\alpha)$ are strictly positive by Lemma~\ref{l_Sinc}.
\end{proof}

For every $\alpha\in \R$ and $\phi\in [0,\pi]$, we have
\begin{equation}\label{ineqmu}
\mu(\alpha,\phi) \leq 2\pi^2\Sinc^2(\pi^2\alpha) + \pi^2.
\end{equation}
Indeed, let $\phi\in [0,\pi]$. If $\alpha<0$, then $\pi^2\alpha\leq(\pi-\phi)^2\alpha$ and, therefore, $0<\Sinc((\pi-\phi)^2\alpha/4) \leq \Sinc(\pi^2\alpha/4)$ by Lemma~\ref{l_Sinc}. In view of~(\ref{tau}), this implies that $\tau(\alpha,\phi)\leq 0$ and, hence, (\ref{ineqmu}) follows from~(\ref{Cos4}) and~(\ref{mu}). If $\alpha\geq 0$, then~(\ref{tau}) ensures that $|\tau(\alpha,\phi)|\leq \pi^2$ (note that $|\Sinc\xi|\leq 1$ for $\xi\geq 0$ by Lemma~\ref{l_Sinc}). Since $|\Cos(\pi^2\alpha)|\leq 1$ by~(\ref{Cos3}), it follows from~(\ref{mu}) that (\ref{ineqmu}) is again satisfied. This completes the proof of~(\ref{ineqmu}).

\begin{lemma}\label{l_R1}
Let $-1<a<1$. Then we have
\[
\frac{1}{|R(\alpha,\vartheta,Ee^{i\phi})|} \leq \frac{(|\ln E| + 3\pi)(E^{a/2} +E^{-a/2})}{2\pi(\pi-\phi)\sinc^2(\pi a)}
\]
for all $\alpha\leq a^2$, $\vartheta\in\R$, $E>0$, and $0\leq \phi<\pi$.
\end{lemma}
\begin{proof}
It follows from~(\ref{Phi}) and Lemma~\ref{l_Sinc} that
\[
|\Phi(\alpha,E)|\leq |\ln E|\ch\left(\frac{a}{2}\ln E\right) = \frac{|\ln E|}{2}(E^{a/2} +E^{-a/2})
\]
for all $\alpha\leq a^2$ and $E>0$. By Lemma~\ref{l_Sinc}, we have $\Sinc(\pi^2\alpha)\geq \Sinc(\pi^2 a^2) = \sinc(\pi a)$ for every $\alpha\leq a^2$. Since $0<\sinc(\pi a) \leq 1$, inequality~(\ref{ineqmu}) and the above estimate imply that
\begin{multline}\nonumber
\frac{(\Phi(\alpha,E)^2+\mu(\alpha,\phi))^{1/2}}{\Sinc(\pi^2\alpha)}\leq \frac{\sqrt{2}\pi\Sinc(\pi^2\alpha)+\pi+|\Phi(\alpha,E)|}{\Sinc(\pi^2\alpha)} = \\= \sqrt{2}\pi + \frac{|\Phi(\alpha,E)|+\pi}{\Sinc(\pi^2 \alpha)} \leq \frac{3\pi+|\Phi(\alpha,E)|}{\sinc(\pi a)}\leq \frac{3\pi+|\ln E|}{2\sinc(\pi a)}(E^{a/2} +E^{-a/2})
\end{multline}
for every $\alpha\leq a^2$, $E>0$, and $0\leq \phi\leq\pi$. The required inequality now follows from Lemma~\ref{l_R} because $\Sinc((\pi-\phi)^2\alpha)\geq \sinc(\pi a)$ for all $\alpha\leq a^2$ and $0\leq \phi\leq 2\pi$ by Lemma~\ref{l_Sinc}.
\end{proof}

\begin{lemma}\label{l_dCos}
Let $a,b$ be real numbers, $A,B,C\geq 0$, and $n,k_1,k_2,k_3\in\Z_+$. Then
\begin{multline}\label{dCos}
\left|\partial^n_\alpha\left(\Cos^{(k_1)}(A^2\alpha)\Cos^{(k_2)}(B^2\alpha)\Cos^{(k_3)}(-C^2\alpha)\right)\right|\leq \\ \leq \frac{n!(A+B+C)^{2n}}{(2n)!2^{k_1+k_2+k_3}}\ch(Ab)\ch(Bb)\ch(Ca)
\end{multline}
for every $\alpha\in [-b^2,a^2]$.
\end{lemma}
\begin{proof}
For every $u\in\R$ and $n\in\Z_+$, we have the inequality (see formula~(12) in~\cite{Smirnov2018})
\[
|\Cos^{(n)}(\xi)|\leq \frac{n!}{(2n)!}\ch u,\quad \xi\geq -u^2.
\]
Using this estimate and the standard formula for the $n$th derivative of a product of functions, we find that the left-hand side of~(\ref{dCos}) is bounded by
\[
\ch(Ab)\ch(Bb)\ch(Ca)\sum\frac{n!}{n_1!n_2!n_3!}\frac{(n_1+k_1)!A^{2n_1}}{(2n_1+2k_1)!}\frac{(n_2+k_2)!B^{2n_2}}{(2n_2+2k_2)!} \frac{(n_3+k_3)!C^{2n_3}}{(2n_3+2k_3)!}
\]
for every $\alpha\in [-b^2,a^2]$, where the sum is taken over all $n_1,n_2,n_3\in\Z_+$ such that $n_1+n_2+n_3=n$. Since
\[
\frac{(n+k)!}{(2n+2k)!}\leq \frac{1}{2^k}\frac{n!}{(2n)!}
\]
for all $n,k\in\Z_+$, the sum in the above expression does not exceed
\[
\frac{1}{2^{k_1+k_2+k_3}}\frac{n!}{(2n)!}\sum_{n_1+n_2+n_3=n} \frac{(2n)!}{(2n_1)!(2n_2)!(2n_3)!}A^{2n_1}B^{2n_2}C^{2n_3},
\]
whence the required estimate follows immediately.
\end{proof}

\begin{lemma}\label{l_dR}
Let $a,b\in\R$. Then we have
\[
|\partial^n_\alpha R(\alpha,\vartheta,Ee^{i\phi})| \leq \frac{n!}{(2n)!} (\pi+|\ln E|)^{2n+1}\ch(\pi b)(E^{a/2} +E^{-a/2})
\]
for all $\alpha\in [-b^2,a^2]$, $\vartheta\in\R$, $E>0$, $\phi\in [0,\pi]$, and $n\in \Z_+$.
\end{lemma}
\begin{proof}
By elementary trigonometric transformations, we derive from~(\ref{R}) that
\begin{multline}\nonumber
R(\kappa^2,\vartheta,Ee^{i\phi}) = \left(ie^{-i\vartheta}\frac{\sin\vartheta_\kappa}{\kappa}\cos\frac{(\pi-\phi)\kappa}{2}- \frac{i}{\kappa}\cos\vartheta\sin\frac{\phi\kappa}{2}\right)(E^{\kappa/2}+E^{-\kappa/2}) -\\- \left(e^{-i\vartheta}\cos\vartheta_\kappa\cos\frac{(\pi-\phi)\kappa}{2}+i\sin\vartheta\cos\frac{\phi\kappa}{2}\right)
\frac{E^{\kappa/2}-E^{-\kappa/2}}{\kappa}
\end{multline}
for every $\kappa\in\C\setminus\{0\}$, $\vartheta\in\C$, $E>0$, and $-\pi/2<\phi<3\pi/2$.
We now fix $E>0$ and $\phi\in [0,\pi]$ and set $A = \pi/2$, $\tilde A = \phi/2$, $B = (\pi-\phi)/2$, and $C = |\ln E|/2$.
In view of the equalities
\[
\frac{E^{\kappa/2}-E^{-\kappa/2}}{\kappa} = \ln E\Sinc\left(-\frac{\kappa^2}{4}\ln^2E\right),\quad \frac{E^{\kappa/2}+E^{-\kappa/2}}{2} = \Cos\left(-\frac{\kappa^2}{4}\ln^2E\right),
\]
which hold for all $\kappa\in\C\setminus\{0\}$ by~(\ref{Cos2}), we conclude that
\begin{multline}\nonumber
R(\alpha,\vartheta,Ee^{i\phi}) = i\left(\pi e^{-i\vartheta}\Sinc(A^2\alpha)\Cos(B^2\alpha) -\phi\cos\vartheta\Sinc(\tilde A^2\alpha)\right)\Cos(-C^2\alpha) -\\- \ln E\left(e^{-i\vartheta}\Cos(A^2\alpha)\Cos(B^2\alpha)+i\sin\vartheta\Cos(\tilde A^2\alpha)\right)
\Sinc(-C^2\alpha)
\end{multline}
for all $\alpha\in\C\setminus\{0\}$ and $\vartheta\in\C$. By continuity, this equality remains valid for $\alpha=0$. By Lemma~\ref{l_dCos} and~(\ref{Cos'}), it follows that
\begin{multline}\nonumber
|\partial^n_\alpha R(\alpha,\vartheta,Ee^{i\phi})| \leq \\ \leq \frac{n!}{(2n)!}((A+B+C)^{2n}\ch(Ab)\ch(Bb) + (\tilde A + C)^{2n}\ch(\tilde Ab))(|\ln E|+\pi)\ch(Ca)
\end{multline}
for every $\alpha\in[-b^2,a^2]$ and $\vartheta\in\R$. This implies the required estimate because $A$, $\tilde A$, and $B$ do not exceed $\pi/2$ by the condition $\phi\in [0,\pi]$ and $\ch^2(\pi b/2)\leq \ch(\pi b)$.
\end{proof}

Given $k\in\Z_+$, let $\Lambda_k=\Z_+^{\{0,\ldots,k\}}$, i.e., $\Lambda_k$ is the set of all maps from $\{0,\ldots,k\}$ to $\Z_+$. For every $k,l\in\Z_+$, we set
\begin{multline}\label{skl}
\mathcal Q_{k,l} =\\= \left\{(s,t)\in\Lambda_k\times\Lambda_k : \sum_{j=0}^k j(s(j)+t(j)) = k,\,\,\sum_{j=0}^k (s(j)+t(j)) = k+l+1\right\}.
\end{multline}

\begin{lemma}\label{l_comb}
For every $k,l\in\Z_+$, there is a map $C\colon \mathcal Q_{k,l}\to\Z$ such that
\begin{equation}\label{sumC}
\sum_{(s,t)\in \mathcal Q_{k,l}}|C(s,t)| \leq 2^{k+l}(k+l)!
\end{equation}
and
\begin{multline}\label{dkdl}
\partial_\alpha^k\partial_\vartheta^l\mathscr F(\alpha,\vartheta,z) =\\= \frac{1}{R(\alpha,\vartheta,z)^{k+l+1}}  \sum_{(s,t)\in \mathcal Q_{k,l}} C(s,t)\prod_{j=0}^k (\partial_\alpha^j R(\alpha,\vartheta,z))^{s(j)} (\partial_\alpha^j R(\alpha,\vartheta+\pi/2,z))^{t(j)}
\end{multline}
for every $(\alpha,\vartheta,z)\in \mathscr O$ (we assume that $\zeta^0=1$ for every $\zeta\in\C$).
\end{lemma}

The proof of Lemma~\ref{l_comb}, which is of purely combinatorial nature, is elementary but rather lengthy. It is given in Appendix~\ref{app}.

\begin{proof}[Proof of Proposition~$\ref{p_estimate}$]
We fix $z\in \C_+$ and set $E=|z|$. Then there is a unique $\phi\in (0,\pi)$ such that $z=Ee^{i\phi}$. Let the map $C\colon \mathcal Q_{k,l}\to\Z$ be as in Lemma~\ref{l_comb}. In view of Lemma~\ref{l_dR}, the sum in the right-hand side of~(\ref{dkdl}) is bounded above by
\[
\sum_{(s,t)\in \mathcal Q_{k,l}} |C(s,t)|\prod_{j=0}^k \left(\frac{j!}{(2j)!}(\pi+|\ln E|)^{2j+1}\ch(\pi b)(E^{a/2}+E^{-a/2})\right)^{s(j)+t(j)}
\]
for every $\alpha\in [-b^2,a^2]$ and $\vartheta\in\R$. As $j!/(2j)!\leq 2^{-j}$ for every $j\in\Z_+$, it follows from~(\ref{skl}) and~(\ref{sumC}) that this expression does not exceed
\[
2^l (k+l)!(\pi+|\ln E|)^{3k+l+1}\left(\ch(\pi b)(E^{a/2}+E^{-a/2})\right)^{k+l+1}.
\]
Since
\[
E(E^{a/2}+E^{-a/2})^2 \leq 4(1+E)^{1+a},\quad (1+|\ln E|)^2\leq 2(1+\ln^2E),
\]
Lemma~\ref{l_R1} and~(\ref{dkdl}) therefore imply that
\[
\left|\partial_\alpha^k\partial_\vartheta^l\mathscr F(\alpha,\vartheta,z)\right|\leq P_{a,b}(k,l)(1+\ln^2E)^{2k+l+1} \left(\frac{(1+E)^{1+a}}{E(\pi-\phi)}\right)^{k+l+1}
\]
for every $\alpha\in [-b^2,a^2]$ and $\vartheta\in\R$. Hence, the required estimate follows because $E(\pi-\phi)\geq E\sin(\pi-\phi)=E\sin\phi = \Im z$.
\end{proof}

\appendix

\section{Even holomorphic functions}
\label{app_even}

Given $r>0$, we let $\mathbb D_r$ denote the disc of radius $r$ in the complex plane centred at the origin: $\mathbb D_r=\{z\in\C: |z|<r\}$.

\begin{lemma}\label{l_open}
\hfill
\begin{itemize}
\item[1.] The map $z\to z^2$ from $\C$ to itself is open.
\item[2.] Let $X$ be a topological space. The map $(z,x)\to (z^2,x)$ from $\C\times X$ to itself is open.
\end{itemize}
\end{lemma}
\begin{proof}
1. Let $O$ be an open subset of $\C$ and $\tilde O$ be its image under the map $z\to z^2$. Let $\zeta\in \tilde O\setminus\{0\}$. Then there is a holomorphic function $g$ defined on a neighborhood of $\zeta$ and such that $g(\zeta)\in O$ and $g(\zeta')^2 = \zeta'$ for every $\zeta'\in D_{g}$ (i.e., $g$ is a holomorphic branch of the square root in a neighborhood of $\zeta$ whose value at $\zeta$ belongs to $O$). By continuity of $g$, there is a neighborhood $V$ of $\zeta$ such that $g(\zeta')\in O$ for all $\zeta'\in V$ and, hence, $V\subset \tilde O$. This means that $\zeta$ is an interior point of $\tilde O$ for every nonzero $\zeta\in \tilde O$. If $0\in \tilde O$, then $0\in O$ and $\mathbb D_r\subset O$ for some $r>0$. It follows that $\mathbb D_{r^2}\subset \tilde O$ and, therefore, $0$ is an interior point of $\tilde O$. This implies that $\tilde O$ is open.
\par\medskip\noindent
2. The assertion follows immediately from statement~1 and the definition of product topology.
\end{proof}

\begin{lemma}\label{l_root}
Let $O\subset \C$ be such that $-z\in O$ for every $z\in O$. Let $\tilde O = \{\zeta\in\C: \zeta = z^2 \mbox{ for some } z\in O\}$ and $f$ be a map on $O$ such that $f(-z) = f(z)$ for every $z\in O$. Then there is a unique map $\tilde f$ on $\tilde O$ such that $f(z) = \tilde f(z^2)$ for all $z\in O$. If $O$ is open and $f$ is a holomorphic function on $O$, then $\tilde O$ is open and $\tilde f$ is a holomorphic function on $\tilde O$.
\end{lemma}
\begin{proof}
The uniqueness of $\tilde f$ is obvious. To prove the existence, we choose $w_\zeta\in \C$ such that $\zeta = w_\zeta^2$ for every $\zeta\in\C$. Clearly, $w_\zeta\in O$ for every $\zeta\in\tilde O$, and we can define $\tilde f$ as the map on $\tilde O$ taking $\zeta\in\tilde O$ to $f(w_\zeta)$. Then we have $f(z) = f(w_{z^2}) = \tilde f(z^2)$ for every $z\in O$ because $z=\pm w_{z^2}$.

Let $O$ be open and $f$ be holomorphic. By statement~1 of Lemma~\ref{l_open}, $\tilde O$ is open. Let $\zeta\in \tilde O\setminus\{0\}$. As in the proof of statement~1 of Lemma~\ref{l_open}, we choose a holomorphic branch of the square root $g$ and a neighborhood $V$ of $\zeta$ such that $V\subset D_g$ and $g(V)\subset O$. Then $V\subset \tilde O$ and $f(g(\zeta')) = \tilde f(g(\zeta')^2) = \tilde f(\zeta')$ for every $\zeta'\in V$. This means that $\tilde f$ coincides on $V$ with the composition of holomorphic functions $f$ and $g$ and, hence, is holomorphic on $V$. This implies that $\tilde f$ is holomorphic on $\tilde O\setminus\{0\}$. If $0\in\tilde O$, then $0\in O$ and there is $r>0$ such that $\mathbb D_r\subset O$. For $k = 0,1,\ldots$, let $a_k = f^{(k)}(0)/k!$ be the Taylor coefficients of $f(z)$ at $z=0$. Since $f$ is even, we have $a_k = 0$ for odd $k$ and, hence, $f(z) = \sum_{n=0}^\infty a_{2n} z^{2n}$ for all $z\in\mathbb D_r$. It follows that the series $\sum_{n=0}^\infty a_{2n} \zeta^n$ converges to some $h(\zeta)$ for every $\zeta\in \mathbb D_{r^2}$. Clearly, $h$ is holomorphic on $\mathbb D_{r^2}$ and we have $h(z^2) = f(z) = \tilde f(z^2)$ for every $z\in \mathbb D_r$. This means that $\tilde f$ coincides with $h$ on $\mathbb D_{r^2}$ and, therefore, $\tilde f$ is holomorphic on $\tilde O$.
\end{proof}

\begin{lemma}\label{l_root1}
Let $n = 1,2,\ldots$ and $O\subset\C\times \C^n$ be such that $(-z,u)\in O$ for every $(z,u)\in O$. Let $\tilde O$ be the image of $O$ under the map $(z,u)\to (z^2,u)$ from $\C\times\C^n$ to itself and $f$ be a map on $O$ such that $f(-z,u) = f(z,u)$ for every $(z,u)\in O$. Then there is a unique map $\tilde f$ on $\tilde O$ such that $f(z,u) = \tilde f(z^2,u)$ for all $(z,u)\in O$. If $O$ is open and $f$ is a holomorphic function on $O$, then $\tilde O$ is open and $\tilde f$ is a holomorphic function on $\tilde O$.
\end{lemma}
\begin{proof}
The uniqueness of $\tilde f$ is obvious. To prove the existence, we choose $w_\zeta\in \C$ such that $\zeta = w_\zeta^2$ for every $\zeta\in\C$. Clearly, $(w_\zeta,u)\in O$ for every $(\zeta,u)\in\tilde O$, and we can define $\tilde f$ as the map on $\tilde O$ taking $(\zeta,u)\in\tilde O$ to $f(w_\zeta,u)$. Then we have $f(z,u) = f(w_{z^2},u) = \tilde f(z^2,u)$ for every $(z,u)\in O$ because $z=\pm w_{z^2}$.

Let $O$ be open and $f$ be holomorphic. By statement~2 of Lemma~\ref{l_open}, $\tilde O$ is open. For every $u\in\C^n$, we let $j_1(u)$  denote the holomorphic map $z\to (z,u)$ from $\C$ to $\C\times \C^n$. For every $z\in\C$, we let $j_2(z)$ denote the holomorphic map $u\to (z,u)$ from $\C^n$ to $\C\times \C^n$. In view of the Hartogs theorem, the holomorphy of $\tilde f$ will be proved if we show that $\tilde f\circ j_1(u)$ and $\tilde f\circ j_2(z)$ are holomorphic functions for every $u\in\C^n$ and $z\in \C$. Let $s$ be the map $z\to z^2$ from $\C$ to itself and $t$ be the map $(z,u)\to (z^2,u)$ from $\C\times \C^n$ to itself. Since $f = \tilde f\circ t$ and $t\circ j_1(u) = j_1(u)\circ s$, we have $f\circ j_1(u) = (\tilde f\circ j_1(u))\circ s$ for every $u\in\C^n$. Because $f\circ j_1(u)$ is a holomorphic function, Lemma~\ref{l_root} implies that $\tilde f\circ j_1(u)$ is a holomorphic function for every $u\in \C^n$. Since $t\circ j_2(z) = j_2(z^2)$, we have $\tilde f\circ j_2(z) = f\circ j_2(z^2)$ and, hence, $\tilde f\circ j_2(z)$ is a holomorphic function for every $z\in \C$.
\end{proof}

\section{Some properties of the functions \texorpdfstring{$\Cos$}{Cos} and \texorpdfstring{$\Sinc$}{Sinc}}
\label{app_CosSinc}

Dividing the identity $\sin 2w = 2\sin w \cos w$ by $2w$, we find that $\sinc 2w = \sinc w \cos w$ for every $w\in\C\setminus \{0\}$. By continuity, this equality remains valid for $w=0$ and it follows from~(\ref{Cos1}) that
\begin{equation}\label{Sinc4}
\Sinc 4\zeta = \Sinc\zeta\Cos\zeta,\quad \zeta\in\C.
\end{equation}
Substituting $\sin w = w\sinc w$ in the identity $\sin^2w+\cos^2w=1$, we obtain $w^2\sinc^2w+\cos^2 w=1$ for every $w\in\C$. In view of~(\ref{Cos1}), this implies that
\begin{equation}\label{Sinc2Cos2}
\zeta \Sinc^2\zeta + \Cos^2\zeta = 1,\quad \zeta\in\C.
\end{equation}
Differentiating the left equality in~(\ref{Cos1}) yields $\sin w = -2 w\Cos'(w^2)$ for every $w\in\C$. Dividing this identity by $w$, we obtain $\sinc w = -2\Cos'(w^2)$ for all $w\in\C\setminus \{0\}$. By continuity, this formula remains valid for $w=0$ and it follows from the right equality in~(\ref{Cos1}) that
\begin{equation}\label{Cos'}
\Sinc\zeta = -2\Cos'\zeta,\quad \zeta\in\C.
\end{equation}

\begin{lemma}\label{l_Sinc}
The functions $\Cos$ and $\Sinc$ are strictly decreasing on the interval $(-\infty,\pi^2]$. For every $\xi<\pi^2$, we have $\Sinc\xi > 0$. If $a\in\R$ and $\xi \geq -a^2$, then $|\Sinc\xi|\leq \ch a$.
\end{lemma}
\begin{proof}
Let the function $f$ on $\R$ be defined by the equality $f(x) = x\cos x -\sin x$. Then we have $\sinc'x = f(x)/x^2$ for every $x>0$. Since $f'(x) = -x\sin x$, we have $f'(x)<0$ for $x\in (0,\pi)$. As $f(0)=0$, we have $f(x)<0$ for $x\in (0,\pi)$ and, therefore, $\sinc$ strictly decreases on $[0,\pi]$. In view of~(\ref{Cos3}), this implies that $\Sinc$ strictly decreases on $[0,\pi^2]$. Since $\Sinc$ is strictly decreasing on $(-\infty,0]$ by~(\ref{CosSinc}), we conclude that $\Sinc$ strictly decreases on $(-\infty,\pi^2]$. In view of the equality $\Sinc(\pi^2) = 0$, it follows that $\Sinc\xi>0$ for $\xi<\pi^2$. By~(\ref{Cos'}), we have $\Cos'\xi = -2^{-1}\Sinc\xi<0$ for $\xi<\pi^2$ and, hence, $\Cos$ is strictly decreasing on $(-\infty,\pi^2]$.
Let $a\in\R$. In view of~(\ref{CosSinc}), (\ref{Cos2}), and the monotonicity of $\Sinc$ established above, we have
\[
\Sinc\xi \leq \Sinc(-a^2)\leq \Cos(-a^2) = \ch a,\quad \xi\in [-a^2,\pi^2].
\]
If $\xi\geq\pi^2$, then we have $|\Sinc\xi| = |\sinc\sqrt{\xi}|\leq 1/\pi<\ch a$. Hence, $|\Sinc\xi|\leq\ch a$ for all $\xi\geq -a^2$.
\end{proof}

\section{Herglotz functions}
\label{s_Herglotz}

A holomorphic function $f$ on $\C_+$ is said to be a Herglotz function if $\Im f(z) > 0$ for every $z\in \C_+$. It is well known (see~\cite{AkhiezerGlazman}, Ch.~6, Sec.~69) that every Herglotz function $f$ admits a unique representation of the form
\begin{equation}\label{herglotz_repr}
f(z) = a + bz + \int_{\R} \left(\frac{1}{t-z}-\frac{t}{t^2+1}\right) d\nu(t),\quad z\in\C_+,
\end{equation}
where $a\in\R$, $b\geq 0$, and $\nu$ is a positive Radon measure on $\R$ such that
\begin{equation}\label{herglotz_meas}
\int_\R \frac{d\nu(t)}{t^2+1} \leq \infty.
\end{equation}
We call $\nu$ the Herglotz measure associated with $f$.

\begin{lemma}\label{l_herglotz}
Let $f$ be a Herglotz function and $\varphi$ be a continuous complex function on $\R$ satisfying the bound $|\varphi(E)|\leq C(1+E^2)^{-2}$, $E\in \R$, for some $C\geq 0$. Then the function $E\to \varphi(E)\Im f(E+i\eta)$ on $\R$ is integrable for every $\eta>0$ and we have
$\int_{\R} \varphi(E)\Im f(E+i\eta) \,dE \to \pi\int_{\R}\varphi(E)\,d\nu(E)$ as $\eta\downarrow 0$,
where $\nu$ is the Herglotz measure associated with $f$.
\end{lemma}
\begin{proof}
Since $\nu$ is the Herglotz measure associated with $f$, representation~(\ref{herglotz_repr}) holds for some $a\in\R$ and $b\geq 0$. We therefore have
\[
\Im f(E+i\eta) = b\eta + \int_\R \frac{\eta\,d\nu(t)}{(t-E)^2+\eta^2},\quad E\in\R,\,\eta>0.
\]
Since
\begin{equation}\label{hergl_est}
\int_\R \frac{|\varphi(E)|\eta\,dE}{(t-E)^2+\eta^2}\leq C \int_\R \frac{\eta\,dE}{(E^2+1)((t-E)^2+\eta^2)} = \frac{C\pi(\eta+1)}{t^2+(\eta+1)^2},
\end{equation}
the function $(E,t)\to \eta((t-E)^2+\eta^2)^{-1}\varphi(E)$ on $\R\times\R$ is integrable with respect to the measure $\lambda\times\nu$, where $\lambda$ is the Lebesgue measure on $\R$. By the Fubini theorem, we conclude that the function $E\to \varphi(E)\Im f(E+i\eta)$ on $\R$ is integrable for every $\eta>0$ and
\begin{equation}\label{hergl_int}
\int_\R \varphi(E)\Im f(E+i\eta)\,dE = b\eta\int_\R\varphi(E)\,dE + \int_\R g_\eta(t)\,d\nu(t),
\end{equation}
where the function $g_\eta$ on $\R$ is given by
\[
g_\eta(t) = \int_\R \frac{\varphi(E)\eta\,dE}{(t-E)^2+\eta^2} = \int_\R \frac{\varphi(t+\eta E)\,dE}{E^2+1}.
\]
By the dominated convergence theorem, $g_\eta(t)\to \pi\varphi(t)$ as $\eta\downarrow 0$ for every $t\in\R$. Since $|g_\eta(t)|\leq 2\pi C(t^2+1)^{-1}$ for $0<\eta<1$ by~(\ref{hergl_est}) and $\nu$ satisfies~(\ref{herglotz_meas}), we can apply the dominated convergence theorem again and conclude that $\int_\R g_\eta(t)\,d\nu(t)\to \pi \int_\R \varphi(t)\,d\nu(t)$ as $\eta\downarrow 0$. In view of~(\ref{hergl_int}), it follows that $\int_\R \varphi(E)\Im f(E+i\eta)\,dE\to \pi \int_\R \varphi(E)\,d\nu(E)$ as $\eta\downarrow 0$.
\end{proof}

\section{Proof of Lemma~\ref{l_analyt0}}
\label{app_analyt}

The proof given below is similar to that of Lemma~2 in~\cite{Smirnov2016}.

Let $\mathrm{Ln}$ be the branch of the logarithm on $\C_\pi$ satisfying $\mathrm{Ln}\,1 = 0$ and $p$ be the holomorphic function on $\C\times \C_\pi$ defined by the relation $p(\kappa,\zeta) = e^{\kappa\,\mathrm{Ln}\,\zeta}$ (hence $p(\kappa,r) = r^\kappa$ for $r>0$). Let $h$ be the holomorphic function on $\C\times\C\times \C_\pi$ such that $h(\kappa,z,\zeta) = p(1/2+\kappa,\zeta)\mathcal X_\kappa(\zeta^2z)$ for all $\kappa,z\in\C$ and $\zeta\in \C_\pi$. By~(\ref{ukappa}), we have
\begin{equation}\label{hu}
h(\kappa,z,r) = \mathfrak u^\kappa(z|r),\quad \kappa,z\in\C,\,r>0.
\end{equation}
We define the holomorphic function $F_2$ on $\C\times\C\times\C_\pi$ by the formula
\[
F_2(\kappa,z,\zeta) = \frac{\pi}{2}(h(\kappa,z,\zeta) + h(-\kappa,z,\zeta))\sinc\vartheta_\kappa,\quad \kappa,z\in\C,\,\zeta\in\C_\pi.
\]
In view of~(\ref{bkappa}) and~(\ref{hu}), the equality $F_2(\kappa,z,r)=\mathfrak b^\kappa(z|r)$ holds for every $\kappa,z\in \C$ and $r>0$.
Further, we define the function $F_1$ on $\C\times\C\times\C_\pi$ by setting
\begin{align}
& F_1(\kappa,z,\zeta) = \frac{h(\kappa,z,\zeta) - h(-\kappa,z,\zeta)}{\kappa}\cos\vartheta_\kappa,\quad\kappa\in \C\setminus\{0\},\nonumber\\
& F_1(0,z,\zeta) = 2\left[\left(\mathrm{Ln}\,\frac{\zeta}{2} + \gamma\right)h(0,z,\zeta) - p(1/2,\zeta)\,\mathcal Y(\zeta^2z)\right] \nonumber
\end{align}
for every $z\in\C$ and $\zeta\in \C_\pi$. It follows immediately from~(\ref{akappa}), (\ref{a0}), (\ref{hu}), and the definition of $F_1$ that $F_1(\kappa,z,r) = \mathfrak a^\kappa(z|r)$ for every $\kappa,z\in\C$ and $r>0$. The function $(z,\zeta)\to F_1(\kappa,z,\zeta)$ is obviously holomorphic on $\C\times \C_\pi$ for every fixed $\kappa\in\C$. The function $\kappa\to F_1(\kappa,z,\zeta)$ is holomorphic on $\C\setminus \{0\}$ and continuous at $\kappa = 0$ (this is ensured by the same calculation as the one used to find the limit in~(\ref{a0})) and is therefore holomorphic on $\C$ for every fixed $z\in\C$ and $\zeta\in \C_\pi$. Hence, $F_1$ is holomorphic on $\C\times\C\times \C_\pi$ by the Hartogs theorem. The uniqueness of $F_1$ and $F_2$ follows from the uniqueness theorem for holomorphic functions.

\section{Proof of Lemma~\ref{l_comb}}
\label{app}

Let $\Xi$ denote the set of all maps from $\Z_+$ to $\Z$ that vanish outside a finite subset of $\Z_+$. We let $\mathbf{0}$ denote the element of $\Xi$ which is identically zero on $\Z_+$: $\mathbf{0}(j)=0$ for all $j\in\Z_+$. For $s\in \Xi$, let $\mathcal K_s = \{j\in\Z_+: s(j)\neq 0\}$. By the definition of $\Xi$, the set $\mathcal K_s$ is finite for every $s\in\Xi$. Let $\Xi_+$ be the set of all $s\in \Xi$ such that $s(j)\geq 0$ for every $j\in\Z_+$. Given $s,t\in\Xi$, we let $s+t$ denote the pointwise sum of $s$ and $t$: $(s+t)(j) = s(j)+t(j)$ for every $j\in \Z_+$. Endowed with this addition, $\Xi$ becomes an Abelian group with zero element $\mathbf{0}$. If $s,t\in\Xi_+$, then $s+t\in\Xi_+$. Given $s\in\Xi_+$, we define the function $e(s)$ on $\C\times\C\times\C_{3\pi/2}$ by the relation
\[
e(s|\alpha,\vartheta,z) = \prod_{j\in I} (\partial_\alpha^jR(\alpha,\vartheta,z))^{s(j)},\quad \alpha,\vartheta\in\C,\,z\in \C_{3\pi/2},
\]
where $I$ is any finite subset of $\Z_+$ such that $\mathcal K_s\subset I$ (clearly, the definition of $e(s)$ does not depend on the choice of $I$). It follows immediately from the above definition that
\begin{equation}\label{e(s+t)}
e(s+t) = e(s)e(t),\quad s,t\in\Xi_+.
\end{equation}
For every $j\in\Z_+$, we define $\delta_j\in\Xi_+$ by the formula
\[
\delta_j(k) = \left\{
\begin{matrix}
1,& k=j,\\
0,& k\in\Z_+\setminus\{j\}.
\end{matrix}
\right.
\]
If $j\in\Z_+$ and $s\in \Xi$, then we set
\[
[s]_j = s-\delta_j+\delta_{j+1}.
\]
If $s\in \Xi_+$ and $j\in \mathcal K_s$, then $[s]_j\in\Xi_+$. We now show that
\begin{align}
& \partial_\alpha e(s|\alpha,\vartheta,z) = \sum_{j\in \mathcal K_s} s(j) e([s]_j|\alpha,\vartheta,z),\label{dalpha}\\
& \partial_\vartheta e(s|\alpha,\vartheta,z) = \sum_{j\in \mathcal K_s} s(j) e(s-\delta_j|\alpha,\vartheta,z)\partial_\alpha^jR(\alpha,\vartheta+\pi/2,z)\label{dtheta}
\end{align}
for every $s\in\Xi_+$, $\alpha,\vartheta\in\C$, and $z\in \C_{3\pi/2}$.
The proof is by induction on the cardinality $\mathrm{Card}\,\mathcal K_s$ of $\mathcal K_s$. Let $\mathrm{Card}\,\mathcal K_s=1$ and $j\in\Z_+$ be such that $\mathcal K_s = \{j\}$. Then we have $e(s|\alpha,\vartheta,z) = (\partial_\alpha^jR(\alpha,\vartheta,z))^{s(j)}$. This implies~(\ref{dalpha}) and~(\ref{dtheta}) because $s(j)>0$ and
\begin{equation}\label{pi/2}
\partial_\vartheta R(\alpha,\vartheta,z)=R(\alpha,\vartheta+\pi/2,z)
\end{equation}
by~(\ref{R}) and~(\ref{R0}). Suppose now that $\mathrm{Card}\,\mathcal K_s>1$. Then we can find $s',s''\in\Xi_+$ such that $s=s'+s''$, the sets $\mathcal K_{s'}$ and $\mathcal K_{s''}$ are both nonempty, and $\mathcal K_{s'}\cap \mathcal K_{s''}=\varnothing$. Since $\mathrm{Card}\,\mathcal K_{s'}< \mathrm{Card}\,\mathcal K_s$ and $\mathrm{Card}\,\mathcal K_{s''}<\mathrm{Card}\,\mathcal K_s$, it follows from the Leibniz rule, the induction hypothesis, and formula~(\ref{e(s+t)}) that
\begin{multline}\nonumber
\partial_\alpha e(s|\alpha,\vartheta,z) = e(s''|\alpha,\vartheta,z)\sum_{j\in \mathcal K_{s'}} s'(j) e([s']_j|\alpha,\vartheta,z) + \\ + e(s'|\alpha,\vartheta,z)\sum_{j\in \mathcal K_{s''}} s''(j) e([s'']_j|\alpha,\vartheta,z).
\end{multline}
Applying~(\ref{e(s+t)}) again and observing that $[s]_j = [s']_j + s'' = s' + [s'']_j$ for every $j\in\Z_+$, we obtain~(\ref{dalpha}). In the same way, (\ref{dtheta}) follows from the Leibniz rule, the induction hypothesis, and formula~(\ref{e(s+t)}). This completes the proof of~(\ref{dalpha}) and~(\ref{dtheta}) for all $s\in\Xi_+$ that are not identically zero. It remains to note that these formulas obviously hold for $s=\mathbf{0}$.\footnote{As usual, we assume that a sum and product of an empty family are equal to zero and unity respectively. In particular, $e(\mathbf{0})$ is identically unity on $\C\times\C\times\C_{3\pi/2}$.}

Given $s,t\in\Xi_+$, we define the function $\mathcal E(s,t)$ on $\C\times\C\times\C_{3\pi/2}$ by the relation
\[
\mathcal E(s,t|\alpha,\vartheta,z) = e(s|\alpha,\vartheta,z)e(t|\alpha,\vartheta+\pi/2,z),\quad \alpha,\vartheta\in\C,\,z\in \C_{3\pi/2}.
\]
Since $\partial_\alpha^jR(\alpha,\vartheta,z)=e(\delta_j|\alpha,\vartheta,z)$ for every $j\in\Z_+$, it follows from~(\ref{periodR}), (\ref{dalpha}), and~(\ref{dtheta}) that
\begin{equation}
\partial_\alpha \mathcal E(s,t|\alpha,\vartheta,z) = \sum_{j\in \mathcal K_s} s(j)\mathcal E([s]_j,t|\alpha,\vartheta,z) + \sum_{j\in \mathcal K_t} t(j)\mathcal E(s,[t]_j|\alpha,\vartheta,z),\nonumber
\end{equation}
\begin{multline}\nonumber
\partial_\vartheta \mathcal E(s,t|\alpha,\vartheta,z) = \sum_{j\in \mathcal K_s} s(j)\mathcal E(s-\delta_j,t+\delta_j|\alpha,\vartheta,z) -\\
- \sum_{j\in \mathcal K_t} t(j)\mathcal E(s+\delta_j,t-\delta_j|\alpha,\vartheta,z)
\end{multline}
In view of~(\ref{pi/2}), this implies that
\begin{equation}\label{dd}
\partial_\alpha\frac{\mathcal E(s,t|\alpha,\vartheta,z)}{R(\alpha,\vartheta,z)^m} = \frac{G_m(s,t|\alpha,\vartheta,z)}{R(\alpha,\vartheta,z)^{m+1}},\quad \partial_\vartheta\frac{\mathcal E(s,t|\alpha,\vartheta,z)}{R(\alpha,\vartheta,z)^m} = \frac{H_m(s,t|\alpha,\vartheta,z)}{R(\alpha,\vartheta,z)^{m+1}}
\end{equation}
for every $s,t\in\Xi_+$, $m = 1,2,\ldots$, and $(\alpha,\vartheta,z)\in \mathscr O$, where the functions $G_m(s,t)$ and $H_m(s,t)$ on $\C\times\C\times\C_{3\pi/2}$ are given by
\begin{equation}
G_m(s,t) = \sum_{j\in \mathcal K_s} s(j)\mathcal E([s]_j+\delta_0,t) + \sum_{j\in \mathcal K_t} t(j)\mathcal E(s+\delta_0,[t]_j)-m\mathcal E(s+\delta_1,t),\label{G_m}
\end{equation}
\begin{multline}
H_m(s,t) = \sum_{j\in \mathcal K_s} s(j)\mathcal E(s-\delta_j+\delta_0,t+\delta_j) -\\- \sum_{j\in \mathcal K_t} t(j)\mathcal E(s+\delta_j+\delta_0,t-\delta_j)-m\mathcal E(s,t+\delta_0).\label{H_m}
\end{multline}

For $k,l\in\Z_+$, we set $\Xi_{k+}=\{s\in\Xi_+: s(j)=0 \mbox{ for }j>k\}$ and
\[
\mathcal S_{k,l} = \{(s,t)\in \Xi_{k+}\times \Xi_{k+}: (s|_{\{0,\ldots,k\}},t|_{\{0,\ldots,k\}})\in \mathcal Q_{k,l}\},
\]
where the set $\mathcal Q_{k,l}$ is given by~(\ref{skl}).

Let $k,l\in\Z_+$, $a\geq 0$, and $f$ be a function on $\C\times\C\times\C_{3\pi/2}$. We say that $f$ is a function of type $(k,l,a)$ if there exists a map $c\colon \mathcal S_{k,l}\to\Z$ such that $\sum_{(s,t)\in \mathcal S_{k,l}}|c(s,t)|\leq a$ and
\begin{equation}\label{sumrepr}
f(\alpha,\vartheta,z) = \sum_{(s,t)\in \mathcal S_{k,l}} c(s,t)\mathcal E(s,t|\alpha,\vartheta,z)
\end{equation}
for every $\alpha,\vartheta\in\C$ and $z\in\C_{3\pi/2}$. In particular, $n\mathcal E(s,t)$ is a function of type $(k,l,|n|)$ for every $(s,t)\in \mathcal S_{k,l}$ and $n\in\Z$. Clearly, the following properties hold:
\begin{enumerate}
\item[(1)] If $f$ is a function of type $(k,l,a)$, then $nf$ is a function of type $(k,l,|n|a)$ for every $n\in\Z$.
\item[(2)] If $f_1$ and $f_2$ are functions of types $(k,l,a_1)$ and $(k,l,a_2)$, then $f_1 + f_2$ is a function of type $(k,l,a_1+a_2)$.
\item[(3)] Let $I$ be a finite set and $F$ and $A$ be maps on $I$ such that $F(\iota)$ is a function of type $(k,l,A(\iota))$ for every $\iota\in I$. Then $\sum_{\iota\in I} F(\iota)$ is a function of type $\left(k,l,\sum_{\iota\in I} A(\iota)\right)$.
\end{enumerate}

Given $k,l\in\Z_+$, we can define a bijection $p$ between $\mathcal S_{k,l}$ and $\mathcal Q_{k,l}$ by setting $p(s,t) = (s|_{\{0,\ldots,k\}},t|_{\{0,\ldots,k\}})$ for every $(s,t)\in \mathcal S_{k,l}$. The right-hand side of~(\ref{sumrepr}) then coincides with the sum in the right-hand side of~(\ref{dkdl}) for $C = c\circ p^{-1}$. Hence, Lemma~\ref{l_comb} is ensured by the following statement.

\begin{lemma}\label{l_comb'}
For every $k,l\in\Z_+$, there is a function $f$ of type $(k,l,2^{k+l}(k+l)!)$ such that
\[
\partial_\alpha^k\partial_\vartheta^l\mathscr F(\alpha,\vartheta,z) = \frac{f(\alpha,\vartheta,z)}{R(\alpha,\vartheta,z)^{k+l+1}}
\]
for every $(\alpha,\vartheta,z)\in \mathscr O$.
\end{lemma}
\begin{proof}
We note that $(\mathbf{0},\delta_0)\in \mathcal S_{0,0}$ and $\mathcal E(\mathbf{0},\delta_0|\alpha,\vartheta,z)=R(\alpha,\vartheta+\pi/2,z)$ for every $\alpha,\vartheta\in\C$ and $z\in\C_{3\pi/2}$. It therefore follows from~(\ref{scrF}) that the statement holds for $k=l=0$ with $f=\mathcal E(\mathbf{0},\delta_0)$. Suppose the statement is true for $k,l\in\Z_+$ and the map $c\colon \mathcal S_{k,l}\to\Z$ is such that
\begin{equation}\label{cbound}
\sum_{(s,t)\in\mathcal S_{k,l}}|c(s,t)|\leq 2^{k+l}(k+l)!
\end{equation}
and
\[
\partial_\alpha^k\partial_\vartheta^l\mathscr F(\alpha,\vartheta,z) = \frac{1}{R(\alpha,\vartheta,z)^{k+l+1}} \sum_{(s,t)\in \mathcal S_{k,l}} c(s,t)\mathcal E(s,t|\alpha,\vartheta,z)
\]
for every $(\alpha,\vartheta,z)\in \mathscr O$. In view of~(\ref{dd}), it follows that
\[
\partial_\alpha^{k+1}\partial_\vartheta^l\mathscr F(\alpha,\vartheta,z) = \frac{g(\alpha,\vartheta,z)}{R(\alpha,\vartheta,z)^{k+l+2}},\quad \partial_\alpha^{k}\partial_\vartheta^{l+1}\mathscr F(\alpha,\vartheta,z) = \frac{h(\alpha,\vartheta,z)}{R(\alpha,\vartheta,z)^{k+l+2}}
\]
for every $(\alpha,\vartheta,z)\in \mathscr O$, where the functions $g$ and $h$ on $\C\times\C\times\C_{3\pi/2}$ are given~by
\begin{equation}\label{gh}
g = \sum_{(s,t)\in \mathcal S_{k,l}} c(s,t) G_{k+l+1}(s,t),\quad h = \sum_{(s,t)\in \mathcal S_{k,l}} c(s,t) H_{k+l+1}(s,t).
\end{equation}
Our statement will be proved by induction on $k$ and $l$ if we show that $g$ and $h$ are functions of type $(k+1,l,a)$ and $(k,l+1,a)$ respectively, where $a=2^{k+l+1}(k+l+1)!$. If $(s,t)\in\mathcal S_{k,l}$, then $([s]_j+\delta_0,t)\in \mathcal S_{k+1,l}$ for every $j\in \mathcal K_s$. Choosing $I=\mathcal K_s$, setting $F(j) = s(j)\mathcal E([s]_j+\delta_0,t)$ and $A(j)=s(j)$ for every $j\in I$, and applying property~(3), we conclude that $\sum_{j\in \mathcal K_s}s(j)\mathcal E([s]_j+\delta_0,t)$ is a function of type $(k+1,l,\sum_{j\in \mathcal K_s}s(j))$. Further, if $(s,t)\in\mathcal S_{k,l}$, then $(s+\delta_0,[t]_j)\in \mathcal S_{k+1,l}$ for every $j\in \mathcal K_t$. Hence, choosing $I=\mathcal K_t$, setting $F(j) = t(j)\mathcal E(s+\delta_0,[t]_j)$ and $A(j)=t(j)$ for every $j\in I$, and applying property~(3), we deduce that $\sum_{j\in \mathcal K_t}t(j)\mathcal E(s+\delta_0,[t]_j)$ is a function of type $(k+1,l,\sum_{j\in \mathcal K_t}t(j))$. Since $(s+\delta_1,t)\in \mathcal S_{k+1,l}$ and $\sum_{j=0}^k (s(j)+t(j))=k+l+1$ for every $(s,t)\in\mathcal S_{k,l}$, it follows from~(\ref{G_m}) and property~(2) that $G_{k+l+1}(s,t)$ is a function of type $(k+1,l,2(k+l+1))$ for every $(s,t)\in\mathcal S_{k,l}$. In a similar way, (\ref{H_m}) and properties~(2) and~(3) imply that $H_{k+l+1}(s,t)$ is a function of type $(k,l+1,2(k+l+1))$ for every $(s,t)\in\mathcal S_{k,l}$. Our claim therefore follows from~(\ref{cbound}), (\ref{gh}), and properties~(1) and~(3).
\end{proof}

\section*{Acknowledgments}
The author is grateful to I.V.~Tyutin and B.L.~Voronov for useful discussions.

\bibliographystyle{acm}
\bibliography{isquare_new}

\end{document}